\documentclass{article}

\usepackage[utf8]{inputenc}

\usepackage[T1]{fontenc}
\usepackage{parskip}
\usepackage{ragged2e}
\usepackage{enumitem}
\usepackage{float}
\usepackage{lscape}
\usepackage{rotating}

\usepackage{fancyhdr}
\usepackage{etoolbox}
\usepackage{geometry}
 \newcommand{\tmpx}{}

\fancypagestyle{sec}{\lhead{\tmpx}}

\usepackage{amsmath}
\usepackage{amsfonts}
\usepackage{amssymb}
\usepackage{mathtools}
\usepackage{amsthm}
\usepackage{bbm}
\usepackage{nccmath} 
\usepackage{scalerel}
\usepackage{actuarialsymbol}
\usepackage{mleftright}
\usepackage{mathrsfs}

\usepackage{multirow}
\usepackage{multicol}
\usepackage{tabularx}

\usepackage{graphicx}
\usepackage[table]{xcolor}
\usepackage{array}
\usepackage{subcaption}

\usepackage{xurl}
\usepackage{thmtools, thm-restate}
\usepackage{enumitem}
\usepackage{titlesec}
\usepackage[absolute,overlay]{textpos}

\usepackage[sectionbib, round]{natbib}
\usepackage[colorlinks, linkcolor = black,  linkcolor = blue,
urlcolor  = blue,
citecolor = blue,
anchorcolor = blue]{hyperref}
\usepackage{etoolbox}
\usepackage{bookmark}

\definecolor{MyGreen}{RGB}{0,153,0}

\titleformat{\subsubsection}
{\normalfont\normalsize\bfseries}{\thesubsubsection}{1em}{}
\titlespacing*{\subsubsection}
{0pt}{3.25ex plus 1ex minus .2ex}{1.5ex plus .2ex}
 
 \titleformat{\paragraph}[runin]
{\itshape}{\theparagraph .}{1em}{}

\newcommand{\stripthanks}[1]{%
  \begingroup
  \let\thanks\@gobble
  #1%
  \endgroup
}

\allowdisplaybreaks

\makeatletter
\newcounter{author}

\newcommand{\authorclean}[1]{%
  \csgdef{author@clean@\the\c@author}{#1}%
}

\renewcommand*\author[1]{%
  \stepcounter{author}%
  \csgdef{author@full@\the\c@author}{#1}%
  \ifcsundef{author@clean@\the\c@author}{%
    \def\temp{#1}%
    \edef\temp{\unexpanded\expandafter{\temp}}%
    \csgdef{author@clean@\the\c@author}{\temp}%
  }{}%
  \ifnum\c@author=1
    \gdef\@author{#1}%
  \else
    \xdef\@author{\unexpanded\expandafter{\@author\and#1}}%
  \fi
}

\newcommand*\email[1]{%
  \csgdef{email@\the\c@author}{#1}}

\newcommand*\orcid[1]{%
  \csgdef{orcid@\the\c@author}{#1}}
  
\newcommand*\address[1]{%
  \csgdef{address@\the\c@author}{#1}}

\AtEndDocument{%
  \xdef\author@count{\the\c@author}%
  \c@author=1
  \print@authors}

\newcommand*\print@authors{%
  \ifnum\c@author>\author@count
  \else
    \print@author{\the\c@author}%
    \advance\c@author by 1
    \expandafter\print@authors
  \fi}

\newcommand*\print@author[1]{%
  \par\medskip
  \begin{tabular}{@{}l@{}}%
    \textsc{\csuse{author@clean@#1}}\\
    \csuse{address@#1}\\
    \textit{E-Mail}: \href{mailto:\csuse{email@#1}}{\csuse{email@#1}}\\
    \textit{ORCiD}:
    \href{\csuse{orcid@#1}}{\csuse{orcid@#1}}
  \end{tabular}}

\makeatother

\numberwithin{equation}{section}

\numberwithin{equation}{section}
\theoremstyle{theorem}
\newtheorem{definition}{Definition}
\numberwithin{definition}{section}
\theoremstyle{theorem}
\newtheorem{theorem}{Theorem}
\numberwithin{theorem}{section}
\theoremstyle{plain}
\newtheorem{proposition}{Proposition}
\numberwithin{proposition}{section}
\theoremstyle{plain}

\numberwithin{corollary}{section}
\newtheorem{lemma}{Lemma}
\numberwithin{lemma}{section}
\theoremstyle{plain}
\newtheorem{remark}{Remark}
\numberwithin{remark}{section}
\mathtoolsset{showonlyrefs}
\usepackage{etoolbox,xparse}

\newtheorem{example}{Example}
\numberwithin{example}{section}

\providecommand{\keywords}[1]
{
  \small	
  \textbf{\textit{Keywords---}} #1
}

\date{\footnotesize{
$^1$Departamento de Matemática y Estadística, Universidad Torcuato Di Tella, Ciudad de Buenos Aires, Argentina\\
    $^2$Institute for Financial and Actuarial Mathematics, Department of Mathematical Sciences, University of Liverpool, Liverpool, United Kingdom\\
    $^3$Institute of Statistics, Biostatistics and Actuarial Science (ISBA), Louvain Institute of Data Analysis and Modeling (LIDAM), Catholic University of Louvain, Louvain-la-Neuve, Belgium%
}}

\author{Pablo Azcue\thanks{Pablo Azcue contributed to the present manuscript before he unfortunately passed away in April 2024. We dedicate this work to him.} $^1$}
\authorclean{Pablo Azcue}
\address{\textit{Departamento de Matemática y Estadística} \\ \textit{Universidad Torcuato Di Tella} \\
\textit{Av. Figueroa Alcorta 7350 (1428BCW)} \\ \textit{Ciudad de Buenos Aires, Argentina}}
\email{pazcue@utdt.edu}
\orcid{https://orcid.org/0000-0003-2624-1841}
\author{Corina Constantinescu$^2$}
\authorclean{Corina Constantinescu}
\address{\textit{Institute for Financial and Actuarial Mathematics} \\ \textit{Department of Mathematical Sciences} \\ \textit{University of Liverpool} \\ \textit{Liverpool, United Kingdom}}
\email{c.constantinescu@liverpool.ac.uk}
\orcid{https://orcid.org/0000-0002-5219-3022}
\author{Jos\'e Miguel Flores-Contró$^3$}
\authorclean{Jos\'e Miguel Flores-Contró}
\address{\textit{Institute of Statistics, Biostatistics and Actuarial Science – ISBA} \\ \textit{Louvain Institute of Data Analysis and Modeling – LIDAM} \\
\textit{Catholic University of Louvain} \\ \textit{Louvain-la-Neuve, Belgium}}
\email{jose.flores@uclouvain.be}
\orcid{https://orcid.org/0000-0001-9332-8811}
\author{Nora Muler$^1$}
\authorclean{Nora Muler}
\address{\textit{Departamento de Matemática y Estadística} \\ \textit{Universidad Torcuato Di Tella} \\
\textit{Av. Figueroa Alcorta 7350 (1428BCW)} \\ \textit{Ciudad de Buenos Aires, Argentina}}
\email{nmuler@utdt.edu}
\orcid{https://orcid.org/0000-0002-8618-719X}

\title{Optimal Cash Transfers and Microinsurance to Reduce Social Protection Costs}

\begin{document}

\vspace{-2ex}

\maketitle

\begin{abstract}

Design and implementation of appropriate social protection strategies is one of the main targets of the United Nation’s Sustainable Development Goal (SDG) 1: No Poverty. Cash transfer (CT) programmes are considered one of the main social protection strategies and an instrument for achieving SDG 1. Targeting consists of establishing eligibility criteria for beneficiaries of CT programmes. In low-income countries, where resources are limited, proper targeting of CTs is essential for an efficient use of resources. Given the growing importance of microinsurance as a complementary tool to social protection strategies, this study examines its role as a supplement to CT programmes. In this article, we adopt the piecewise-deterministic Markov process introduced in \cite{Article:Kovacevic2011} to model the capital of a household, which when exposed to proportional capital losses (in contrast to the classical Cramér–Lundberg model) can push them into the poverty area. Striving for cost-effective CT programmes, we optimise the expected discounted cost of keeping the household\rq s capital above the poverty line by means of injection of capital (as a direct capital transfer). Using dynamic programming techniques, we derive the Hamilton–Jacobi–Bellman (HJB) equation associated with the optimal control problem of determining the amount of capital to inject over time. We show that this equation admits a viscosity solution that can be approximated numerically. Moreover, in certain special cases, we obtain closed-form expressions for the solution. Numerical examples show that there is an optimal level of injection above the poverty threshold, suggesting that efficient use of resources is achieved when CTs are preventive rather than reactive, since injecting capital into households when their capital levels are above the poverty line is less costly than to do so only when it falls below the threshold.
\vspace{0.5cm}


\end{abstract}

\section{Introduction}  \label{Introduction-Section1}

Eradicating poverty in all its forms remains one of the most pressing and complex challenges in global development. In 2015, the international community adopted the 2030 Agenda for Sustainable Development, which outlined 17 Sustainable Development Goals (SDGs) aimed at fostering social equity, economic growth, and environmental resilience. At the core of this agenda is SDG 1: End poverty in all its forms everywhere, which includes targets such as eradicating extreme poverty, reducing overall poverty by at least 50\%, and implementing comprehensive social protection strategies \citep{Book:UnitedNations2015}. As the world faces growing economic volatility, climate-related shocks, and rising inequality, the need for resilient, cost-effective, and scalable poverty alleviation mechanisms has become more urgent than ever.

Social protection strategies are widely regarded as essential instruments for achieving these poverty reduction goals. Following \cite{Article:Slater2011}, they are typically categorised into three pillars: (i) social insurance, such as unemployment or health insurance, which relies on individual contributions; (ii) social assistance, which involves non-contributory subsidies and capital injections for vulnerable populations; and (iii) regulatory standards, which establish legal protections for workers and consumers \citep{Book:Harvey2005, Article:Dfid2006, Article:Farrington2006}. Among these, social assistance (particularly through direct support and targeted transfers) has proven especially effective for low-income households that lack access to formal insurance or financial markets.

This article focuses on social assistance and, in particular, on cash transfer (CT) programmes. In their simplest form, cash transfer programmes provide cash assistance to the poor and/or vulnerable (people living just above the poverty line and facing the risk, in absence of the transfer, of falling into the poverty area). Transfers can be made in small regular payments or in a lump-sum and are usually subsidised by the government \citep{Article:Tavor2002}. However, cash transfer programmes may also be funded by international organisations and non-governmental organisations (NGOs) \citep{Book:WorldBank2012}.  Cash transfer programmes are generally classified as conditional (CCTs) or unconditional (UCTs). The main difference between these is that the former requires beneficiaries to meet certain obligations, such as enrolling children in school or attending regular medical check-ups \citep{Article:Handa2006}, in order to receive transfers, while the latter provides transfers without any additional requirements beyond eligibility.

Our work builds on and extends the piecewise-deterministic Markov process introduced by \cite{Article:Kovacevic2011} to examine the cost-effectiveness of CTs. In its original form, the process portrays the dynamics of a household\rq s capital across time. Within this framework, economic growth arises exclusively from savings, as only the portion of income not consumed can be transformed into productive capital. Capital is interpreted broadly, encompassing not only physical assets but also health, skills, and productive abilities—that is, all elements shaping a household’s capacity to generate income. Households below a critical income level are unable to generate any surplus, implying zero savings and no accumulation of this comprehensive capital stock. Since income is proportional to capital, the absence of investment leads to stagnation. In contrast, once income exceeds subsistence, part of the surplus can be saved and reinvested, strengthening both material and human capital and initiating a cumulative process of growth. Although stylised, this mechanism explains that growth requires an investable surplus, and in contexts of extreme poverty and limited access to external finance, domestic savings are the primary source of such surplus.  While the model assumes that households derive income solely from capital, this should be interpreted as a reduced-form representation of environments where productive assets (e.g., land, livestock, or small businesses) are the primary source of income and are exposed to stochastic shocks. This framework is particularly relevant in the context of climate risk, disaster relief, and social protection policies in developing economies. This model has attracted attention since its publication, with researchers typically approaching it from a ruin-theoretic perspective. \cite{Article:Kovacevic2011} and \cite{Article:Azais2015} use numerical methods to estimate the (infinite-time) trapping probability (the probability of a household\rq s capital falling into the area of poverty at some point in time). Considering that the most desirable outcome in the ruin-theoretic context is to determine closed-form expressions for ruin probabilities \citep{Book:Asmussen2010}, \cite{Article:Henshaw2023b} and \cite{Article:Flores-Contro2025} derive, under certain assumptions, closed-form expressions for the trapping probability and the Gerber-Shiu expected discounted penalty function of the process, respectively. 

Alternative versions to the original formulation of the model have also been studied recently. \cite{Article:Flores-Contro2024b} consider losses that are subtracted from the household\rq s capital, rather than prorated as in the original model. Under this variation of the model, the authors evaluate the impact of (subsidised) microinsurance on the trapping probability. Similarly, \cite{Article:Flores-Contro2024a} consider a risk process that incorporates ideas from the Omega model, originally introduced in \cite{Article:Albrecher2011}. Given this novel model structure, the authors gauge the effects of direct capital CTs on the trapping probability and the probability of households becoming extremely poor.

In contrast with prior studies, this article examines the model through the lens of stochastic optimal control theory. It assesses the expected discounted cost borne by the government to ensure households remain above the poverty threshold (i.e., the expected discounted capital injections or direct capital transfers). Using dynamic programming techniques \citep{Article:Bellman1954}, we derive the Hamilton-Jacobi-Bellman (HJB) equation corresponding to this control problem. When a classical solution satisfying the appropriate boundary conditions exists then it is the optimal value function for the problem. However, such solutions are not always guaranteed, prompting the consideration of weaker solution concepts (namely, \textit{viscosity solutions}, as introduced by \cite{Article:Crandall1983}). Viscosity solutions have since become a standard tool in control optimisation problems; see, for example, \cite{InCo:Soner1988}, \cite{Book:Bardi1997}, and \cite{Book:Fleming2006}. For a comprehensive overview of stochastic control theory in insurance, refer to \cite{Book:Schmidli2007} and \cite{Book:Azcue2014}. In certain specific settings, closed-form solutions are attainable. It is important to emphasise that the present study is purely theoretical in nature. It does not rely on any empirical data; instead, it develops and analyses a conceptual model to explore CT programmes.

Two of the main concerns during the formulation of cash transfer programmes are: the identification of individuals or groups that will be eligible to benefit from the programme (\textit{targeting}) and the affordability of the programme (\textit{sustainability}). Indeed, the efficient allocation of poverty resources to those most in need is the main concern in poverty reduction programmes (including cash transfers programmes) and is an issue that has been at the forefront of policy debates over the last decades \citep{Article:Keen1992}. That is, targeting is seen as crucial for efficient resource allocation and takes on much greater relevance in low-income countries where resources for social protection are limited \citep{Book:Slater2010}. Our results support this idea, suggesting that the cost incurred by the government can be reduced when transfers are preventive rather than reactive, since injecting capital into households when their capital levels are above the poverty threshold is less costly than doing so only when they fall below the threshold (we refer to these strategies as \textit{threshold strategies}). In other words, our results suggest that optimal strategies are precisely threshold strategies. This is of utmost importance, as reducing the expected discounted capital transfers increases viability and sustainability of programmes, two issues that have also been a cause for concern recently \citep{Article:Owusu-Addo2023}. 

Microinsurance\footnote{The broader term \textit{inclusive insurance} was introduced by the International Association of Insurance Supervisors (IAIS) in 2015. Inclusive insurance incorporates excluded or underserved individuals (e.g. women) into its definition \citep{Misc:IAIS2015} --- not only low-income individuals. In contemporary usage, the term inclusive insurance is more prevalent and viewed as more precise.} is insurance aimed at low-income individuals. In other words, low-income individuals pay a premium proportional to the probability of a certain risk occurring in exchange for protection against it. The primary target market for microinsurance comprises people living in poverty or vulnerable individuals living just above the poverty line. Microinsurance products must be tailored to meet the characteristics of these individuals (e.g. microinsurance products need to be easy to understand, feature affordable premiums and address the specific vulnerabilities to which this population is often exposed). Microinsurance benefits are usually most effective when combined with the benefits of other complementary social protection instruments, such as social insurance or CTs \citep{Book:Churchill2012}. For example, \cite{Book:Churchill2006} illustrates how, as part of the reform of the healthcare system in Colombia in 1993, the government offers subsidies that enable the poor to be purchasers of health insurance. This expansion of social protection through microinsurance in Colombia also stimulates competition among microinsurance providers to serve the low-income market. Given the importance of microinsurance in recent years as a complementary tool to social protection strategies, we also analyse its role as a complement to CT programmes in Section \ref{Microinsurance-Section7}. Within the framework of this consolidation, we observe that microinsurance can contribute to reducing the expected discounted capital injections. Our results are consistent with previous findings, which highlight that microinsurance is a very helpful approach to social protection in very different settings (see Chapter 2 of \cite{Book:Churchill2012} for an overview of the potential of microinsurance for social protection).

The remainder of the paper is structured as follows. In Section \ref{TheStochasticControlProblem-Section2}, we introduce the stochastic control problem. In particular, Section \ref{TheStochasticControlProblem-Section2-Subsection21} presents a detailed description of the model studied throughout this manuscript, which corresponds to the original formulation from \cite{Article:Kovacevic2011}. Section \ref{TheStochasticControlProblem-Section2-Subsection22} defines the cost of social protection, compares it with the expected discounted cost to the government resulting from CTs providing perpetual regular transfers instead of lump-sum capital injections, and defines the optimisation problem with its corresponding value function. The Hamilton-Jacobi-Bellman (HJB) equation associated with the control problem of determining the optimal amount of transfer to inject over time and the properties of the optimal value function are presented in Section \ref{AnalysisofV:PropertiesandtheHamilton-Jacobi-BellmanEquation-Section3}. Threshold strategies are introduced in Section \ref{ThresholdStrategies-Section4}. Building upon this framework, Section \ref{Closed-FormSolutionsinaSpecialCase-Section5} presents a case in which the value functions of threshold strategies admit closed-form expressions. Conversely, Section \ref{GeneralCaseAnalysis:AbsenceofClosed-FormSolutions-Section6} addresses scenarios where analytical solutions are unavailable; in these cases, we employ a Monte Carlo-based approach to numerically estimate the corresponding value function. Furthermore, in the examples presented in Sections \ref{Closed-FormSolutionsinaSpecialCase-Section5} and \ref{GeneralCaseAnalysis:AbsenceofClosed-FormSolutions-Section6}, we determine the optimal threshold (defined as the threshold that minimises the cost borne by the government) and its associated value function, which is found to be the optimal value function among all admissible strategies for these examples. Section \ref{Microinsurance-Section7} incorporates microinsurance as a complementary instrument for social protection and analyses its role. For the examples involving microinsurance, we also identify the optimal threshold and evaluate the corresponding value function, which we conclude to be optimal among all admissible strategies. Lastly, concluding remarks are provided in Section \ref{Conclusion-Section8}.

\section{The Stochastic Control Problem} \label{TheStochasticControlProblem-Section2}

\subsection{Description of the Model} \label{TheStochasticControlProblem-Section2-Subsection21}

We assume that, in the absence of lump-sum capital transfers, a household's capital evolves according to the dynamics described by \cite{Article:Kovacevic2011}. At each time $t$, we assume that the income $I_{t}$ of an individual household is split into consumption $C_{t}$ and investment (or savings) $A_{t}$\ as $I_{t}=C_{t}+A_{t}$. The consumption as a function of the income is given by

\vspace{0.3cm}

\begin{align}
C_{t}=\left\{
\begin{array}
[c]{ll}
I_{{\small t}} & \textit{if}~I_{{\small t}}\leq I^{\ast},\\
I^{\ast}+a(I_{t}-I^{\ast}) & \textit{if}~I_{t}>I^{\ast},
\end{array}
\right.
\label{TheStochasticControlProblem-Section2-Equation1}
\end{align}

\vspace{0.3cm}

for some rate of consumption $a\in(0,1)$ and critical income $I^{\ast}$. So,

\vspace{0.3cm}

\begin{align}
A_{t}=\left\{
\begin{array}
[c]{ll}
0 & \textit{if}~I_{{\small t}}\leq I^{\ast},\\
(I_{t}-I^{\ast})(1-a) & \textit{if}~I_{t}>I^{\ast}.
\end{array}
\right.
\label{TheStochasticControlProblem-Section2-Equation2}
\end{align}

\vspace{0.3cm}

The household\rq s capital process $X_{t}$\ grows as $dX_{t}/c=A_{t}dt$\ with some positive constant $c,$ and income is generated through $I_{t}=bX_{t}$ with income generation $b>0$. So putting $r=(1-a)\cdot b \cdot c$ and $x^{\ast}=I^{\ast}/b$ we get the dynamic system $dX_{t}=r[X_{t}-x^{\ast}]^{+}dt$, with $[x]^{+} = \max(x,0)$. The \textit{critical capital} (or poverty line) $x^{\ast}$ is the upper bound of the poverty area: if the initial capital is above $x^{\ast}$ the capital (and consumption and investment) grows exponentially with rate $r$, whereas if the initial capital is below $x^{\ast}$, the capital remains constant and all the income (below the critical income $I^{\ast}$) is consumed.

We further suppose that the capital $X_{t}$ is subject to loss events (e.g. flood, hurricanes and earthquakes). The occurrence of these events follows a Poisson process $\left(N_{t}\right)_{t\geq0}$ with intensity $\lambda$, with the remaining proportion of capital in the $i$th event described by a sequence of i.i.d. random variables $\left(Z_{i}\right)_{i\geq 1}$ with distribution function $G_{Z}$ supported in $(0,1)$ and mean $\mu=\mathbb{E}\left[Z_{i}\right]$, independent to the Poisson process. If $x$ is the initial capital and $(\tau_{i},Z_{i})$ are the time and the remaining proportion of capital at the $i$th event, respectively, the capital process is given by

\vspace{0.3cm}

\begin{align}
\left\{
\begin{array}
[c]{lll}
dX_{t} & = & r \left[X_{t}-x^{\ast}\right]^{+}dt \hspace{0.3cm} \textit{for }\ t\in(\tau_{i},\tau_{i-1}),\\
X_{{\small \tau}_{i}} & = & X_{{\small \tau}_{i}^{-}} \cdot Z_{i}.
\end{array}
\right.
\label{TheStochasticControlProblem-Section2-Equation3}
\end{align}

\vspace{0.3cm}

That is, on the one hand, in between loss events, the household's capital process is given by 

\vspace{0.3cm}

\begin{align}
X_{t}=
\begin{cases}
\left(  X_{\tau_{i-1}}-x^{\ast}\right)  e^{r\left(  t-\tau_{i-1}\right)}+x^{\ast} & \textit{if }X_{\tau_{i-1}}>x^{\ast},\\
X_{\tau_{i-1}} & \textit{otherwise},
\end{cases}
\label{TheStochasticControlProblem-Section2-Equation4}
\end{align}

\vspace{0.3cm}

for $\tau_{i-1}\leq t<\tau_{i}$ and $\tau_{0}=0$. On the other hand, at the jump times $t=\tau_{i}$, the process is given by

\vspace{0.3cm}

\begin{align}
X_{\tau_{i}}=
\begin{cases}
\left[  \left(  X_{\tau_{i-1}}-x^{\ast}\right)  e^{r\left(  \tau_{i}-\tau_{i-1}\right)  }+x^{\ast}\right]  \cdot Z_{i} & \textit{if } X_{\tau_{i-1}}>x^{\ast},\\
X_{\tau_{i-1}} \cdot Z_{i} & \textit{otherwise}.
\end{cases}
\label{TheStochasticControlProblem-Section2-Equation5}
\end{align}

\vspace{0.3cm}


The model highlights the existence of a subsistence-induced nonlinearity in capital accumulation, which generates a poverty trap preventing low-income households from engaging in productive investment. Because all income below a critical income $I^{*}$ is devoted to basic consumption, savings, and thus growth, are only feasible above this level. As a result, initial conditions play a decisive role in determining long-run outcomes, leading to persistent inequality and divergence across households. From a policy perspective, the framework underscores the importance of threshold-crossing interventions, such as cash transfers, improved access to credit, and investments in health and institutional quality, which can shift households from stagnation to self-sustaining growth paths.

More precisely, the sample set is given by

\begin{align}
\Omega=\{(\tau_{i},Z_{i})_{i\geq1}\in\lbrack0,\infty)\times(0,1):\tau_{n}<\tau_{n+1}~\text{and }\lim_{n\rightarrow\infty}\tau_{n}=\infty\}.
\label{TheStochasticControlProblem-Section2-Equation6}
\end{align}

Here, $\mathcal{F}$ is the complete $\sigma$-field generated by the random variables $\tau_{i}:\Omega^{\mathbf{\lambda}}\rightarrow\lbrack0,\infty)$ and $Z_{i}:\Omega^{\lambda}\rightarrow(0,1)$; the filtration $\left(\mathcal{F}_{t}\right)  _{t\geq0},$ where $\mathcal{F}_{t}$ is the complete $\sigma$-field generated by the random variables $\tau_{i}:\Omega^{\mathbf{\lambda}}\rightarrow\lbrack0,\infty)$ and $Z_{i}:\Omega^{\lambda}\rightarrow(0,1)$ for $\tau_{i}\leq t$; and $\mathbb{P}$~is the probability measure defined in $\mathcal{F}$ which satisfies:

\begin{enumerate}

\item $(Z_{i})_{i\geq1}$ is a sequence of i.i.d. random variables with $\mathbb{P}(Z_{i}\leq z)=G_{Z}(z)$;

\item The counting process\textit{\ }$N_{t}:\Omega^{\mathbf{\lambda}}\rightarrow\mathbb{N}_{0}$ defined by$~N_{t}=\#\{n:\tau_{n}\leq t\}~$is a Poisson process of intensity $\lambda$ and;

\item The random variables $Z_{i}$ are independent of the counting process $N_{t}$.

\end{enumerate}

Given a discount rate $\delta>0$ and a continuously function $w:[0,\infty)\rightarrow\mathbb{R}$ differentiable in $[x^{\ast},\infty)$, the discounted infinitesimal generator of the Markov process $X_{t}$ with initial capital $x\geq x^{\ast}$ is given by

\vspace{0.3cm}

\begin{align}
\begin{array}
[c]{lll}
\widetilde{\mathcal{G}}\left(\left(e^{-\delta t}X_{t}\right)_{t\geq0},w\right)   & = & \lim\limits_{t\rightarrow0^{+}}\frac{\mathbb{E}_{x}\left[e^{-\delta t}w(X_{t})\right]-w(x)}{t}\\ \\
& = & r(x-x^{\ast})w^{\prime}(x)-(\lambda+\delta)w(x)+\lambda\int_{0}^{1}w(x \cdot z)dG_{Z}(z).
\end{array}
\label{TheStochasticControlProblem-Section2-Equation7}
\end{align}

\vspace{0.3cm}


\subsection{Lump-Sum Capital Injections vs Perpetual Regular Transfers} \label{TheStochasticControlProblem-Section2-Subsection22}

We consider the case in which the government decides to make lump-sum capital injections if, after a loss event, capital falls below the critical capital level $x^{\ast}$. The capital injection is such that it returns the household\rq s capital to the critical capital value $x^{\ast}$. For an initial capital $x\geq0$ and a positive discount rate $\delta$, we define the \textit{cost of social protection} and denote it as  $C(x)$, as the expected discounted cumulative injections necessary to avoid household\rq s capital from falling below the critical capital $x^{\ast}$. Note that, if $\tau_{1}$ and $Z_{1}$ denote the time and remaining proportion of the capital in the first loss event, respectively, it yields

\vspace{0.3cm}

\begin{align}
\begin{array}
[c]{lll}
C\left(x^{\ast}\right)  & = & \mathbb{E}\left[  \left(\left(1-Z_{1}\right)  x^{\ast}+C\left(x^{\ast}\right)  \right)  e^{-\delta\tau_{1}}\right] \\ \\
& = & \left((1-\mu)x^{\ast}+C\left(x^{\ast}\right)\right)\frac{\lambda}{\lambda+\delta},
\end{array}
\label{TheStochasticControlProblem-Section2-Equation8}
\end{align}

\vspace{0.3cm}

which leads to

\vspace{0.3cm}

\begin{align}
C\left(x^{\ast}\right)  =\frac{\lambda(1-\mu)x^{\ast}}{\delta}.
\label{TheStochasticControlProblem-Section2-Equation9}
\end{align}

\vspace{0.3cm}

Also, for $x\in\lbrack0,x^{\ast})$, we have

\vspace{0.3cm}

\begin{align}
C\left(x\right)=(x^{\ast}-x)+C\left(  x^{\ast}\right)=(x^{\ast}-x)+\frac{\lambda(1-\mu)x^{\ast}}{\delta},
\label{TheStochasticControlProblem-Section2-Equation10}
\end{align}

\vspace{0.3cm}

and for $x>x^{\ast}$, consider the first time in which the capital is less or equal to the poverty line $x^{\ast}$, that is,

\vspace{0.3cm}

\begin{align}
\overline{\tau}=\min\{\tau_{i}:X_{\tau_{i}}\leq x^{\ast}\},
\label{TheStochasticControlProblem-Section2-Equation11}
\end{align}

\vspace{0.3cm}

which yields,

\vspace{0.3cm}

\begin{align}
C\left(x\right)  =\mathbb{E}\left[\left((x^{\ast}-X_{\overline{\tau}})+C\left(x^{\ast}\right)  \right)e^{-\delta\overline{\tau}}\right]=\mathbb{E}\left[\left((x^{\ast}-X_{\overline{\tau}})+\frac{\lambda(1-\mu)x^{\ast}}{\delta}\right)e^{-\delta\overline{\tau}}\right].
\label{TheStochasticControlProblem-Section2-Equation12}
\end{align}

\vspace{0.3cm}

Now let us consider the income $I_{t}=bX_{t}$ generated by the capital process $X_{t}$ with income generation $b > 0$, and let us suppose that the government provides transfers at a rate $I^{\ast}-I_{t}$ when $X_{t}\leq x^{\ast}$. Thus, transfers are immediately consumed and are perpetually paid by the government after the \textit{trapping time} $\overline{\tau}$. We call this social protection strategy \textit{perpetual regular transfers} and denote $D(x)$ as the expected value of the discounted perpetual regular transfers:

\vspace{0.2cm}

\begin{align}
D(x)=\mathbb{E}_{x}\left[\int_{\overline{\tau}}^{\infty}\left(I^{\ast}-I_{s}\right)e^{-\delta s}ds\right].
\label{TheStochasticControlProblem-Section2-Equation13}
\end{align}

\vspace{0.3cm}

\begin{remark}
It is important to note that the government makes transfers from the moment the household\rq s capital falls below the poverty threshold. Moreover, the capital of a household does not grow once it falls below that threshold.
\end{remark}

\vspace{0.3cm}

\begin{proposition} \label{TheStochasticControlProblem-Section2-Proposition1}

The expected value of the discounted perpetual regular transfers, $D(x)$, is given by

\vspace{0.3cm}

\begin{align}
D(x)=\left(\frac{bx^{\ast}}{\delta}\right)\left(\frac{\lambda(1-\mu)}{\delta+\lambda(1-\mu)}\right)+\frac{b}{\delta+\lambda(1-\mu)}(x^{\ast}-x),
\label{TheStochasticControlProblem-Section2-Equation14}
\end{align}

\vspace{0.3cm}

for $x\leq x^{\ast}$, and

\vspace{0.3cm}

\begin{align}
D(x)=\mathbb{E}\left[D(X_{\overline{\tau}})e^{-\delta\overline{\tau}}\right],
\label{TheStochasticControlProblem-Section2-Equation15}
\end{align}

\vspace{0.3cm}

otherwise.

\end{proposition}

\begin{proof}

Let $Z_{0}=1$. If the initial capital is $x\leq x^{\ast}$, then the trapping time $\overline{\tau}=0$ and we have

\vspace{0.3cm}

\begin{align}
D(x) & = \mathbb{E}\left[\sum\limits_{n=1}^{\infty}\int_{\tau_{n-1}}^{\tau_{n}}(I^{\ast}-I_{\tau_{n-1}})e^{-\delta t}dt\right]  \\ \\
& = \sum\limits_{n=1}^{\infty}\mathbb{E}\left[b(x^{\ast}-x \cdot Z_{1}Z_{2}\cdots Z_{n-1})\right] \cdot \mathbb{E}\left[  \int_{\tau_{n-1}}^{\tau_{n}}e^{-\delta
t}dt\right]  \\ \\
& = \sum\limits_{n=1}^{\infty}b(x^{\ast}-x \cdot \mu^{n-1})\mathbb{E}\left[e^{-\delta\tau_{1}}\right]  ^{n-1}\left(\int_{0}^{\infty}\left(\frac{1-e^{-\delta
t}}{\delta}\right)\lambda e^{-\lambda t}dt\right)  \\ \\
& = b\left(\frac{~x^{\ast}}{\delta}-\frac{~x}{\delta+\lambda(1-\mu)}\right).
\label{TheStochasticControlProblem-Section2-Equation16}
\end{align}

\vspace{0.3cm}

In particular,

\vspace{0.3cm}

\begin{align}
D(x^{\ast})=b\left(\frac{~x^{\ast}}{\delta}-\frac{x^{\ast}}{\delta+\lambda(1-\mu)}\right)=\left(\frac{bx^{\ast}}{\delta}\right)\left(\frac{\lambda(1-\mu)}{\delta+\lambda(1-\mu)}\right),
\label{TheStochasticControlProblem-Section2-Equation17}
\end{align}

\vspace{0.3cm}

and

\vspace{0.3cm}

\begin{align}
\begin{array}
[c]{lll}
D(x) & = & b\left(\frac{~x^{\ast}}{\delta}-\frac{~~x^{\ast}}{\delta+\lambda(1-\mu)}\right)+b\left(\frac{~~x^{\ast}}{\delta+\lambda(1-\mu)}-\frac{~x}{\delta+\lambda(1-\mu)}\right)\\ \\
& = & D(x^{\ast})+\left(\frac{b}{\delta+\lambda(1-\mu)}\right)(x^{\ast}-x),
\end{array}
\label{TheStochasticControlProblem-Section2-Equation18}
\end{align}

for $x\leq x^{\ast}$.
\end{proof}

Proposition \ref{TheStochasticControlProblem-Section2-Proposition2} compares the strategy of lump-sum capital transfers up to the poverty line $x^{\ast}$ whenever is necessary with perpetual regular transfers.

\vspace{0.3cm}

\begin{proposition} \label{TheStochasticControlProblem-Section2-Proposition2}

We have that, 

\vspace{0.3cm}

\begin{center}
$D(x^{\ast})-C\left(x^{\ast}\right)\geq0$ if, and only if, $b\geq\delta+\lambda(1-\mu).$
\end{center}

\vspace{0.3cm}

\end{proposition}

\begin{proof}

In the case $x\leq$ $x^{\ast}$, from Proposition \ref{TheStochasticControlProblem-Section2-Proposition1} and \eqref{TheStochasticControlProblem-Section2-Equation10} yields

\vspace{0.3cm}

\begin{align}
\begin{array}
[c]{lll}
D(x)-C\left(  x\right)  & = & b\left(\frac{~x^{\ast}}{\delta}-\frac{~x}{\delta+\lambda(1-\mu)}\right)-\left(  (x^{\ast}-x)+C(x^{\ast})\right) \\ \\
& = & \left(  \frac{b}{\delta+\lambda(1-\mu)}-1\right)  \left(\frac{x^{\ast}\lambda(1-\mu)}{\delta}+(x^{\ast}-x)\right).
\end{array}
\label{TheStochasticControlProblem-Section2-Equation19}
\end{align}

\vspace{0.3cm}

Then, $D(x)-C\left(  x\right)  \geq0,\ $if and only if, $b\geq\delta+\lambda(1-\mu)$. On the other hand, when $x>$ $x^{\ast}$, from Proposition \ref{TheStochasticControlProblem-Section2-Proposition1} and \eqref{TheStochasticControlProblem-Section2-Equation12} we have

\vspace{0.3cm}

\begin{align}
D(x)-C\left(x\right)  =\mathbb{E}\left[\left(D(X_{\overline{\tau}})-C(X_{\overline{\tau}})\right)e^{-\delta\overline{\tau}}\right].
\label{TheStochasticControlProblem-Section2-Equation20}
\end{align}

\vspace{0.3cm}

Hence, from \eqref{TheStochasticControlProblem-Section2-Equation19} we have the result.
\end{proof}

\vspace{-1cm}

\begin{figure}[H]
	\begin{subfigure}[b]{0.5\linewidth}
  		\includegraphics[width=7.5cm, height=7.5cm]{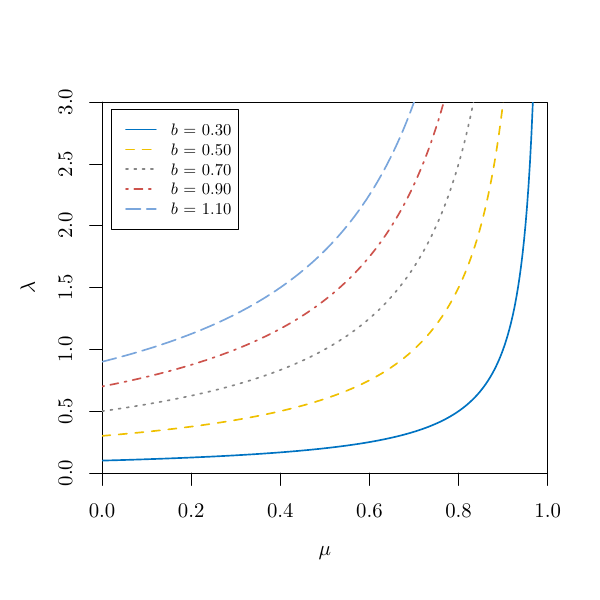}
		\caption{}
  		\label{TheStochasticControlProblem-Section2-Figure2-a}
	\end{subfigure}
	\begin{subfigure}[b]{0.5\linewidth}
  		\includegraphics[width=7.5cm, height=7.5cm]{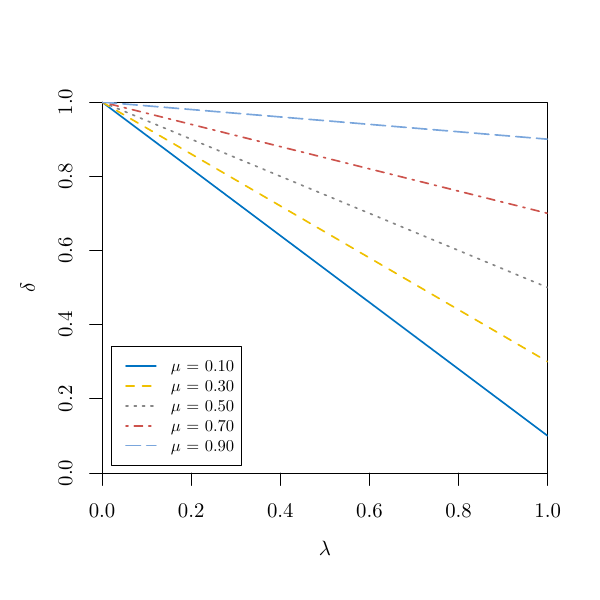}
		\caption{}
  		\label{TheStochasticControlProblem-Section2-Figure2-b}
	\end{subfigure}
	\caption{(a) Upper boundary of the region defined by the constraint $b \geq \delta + \lambda \left(1 - \mu\right)$ derived in Proposition \ref{TheStochasticControlProblem-Section2-Proposition2} with (a) fixed $\delta = 0.2$ and different values of $b$ and (b) fixed $b=1$ and different values of $\mu$.}
	\label{TheStochasticControlProblem-Section2-Figure2}
\end{figure}

Figure \ref{TheStochasticControlProblem-Section2-Figure2} shows the boundary curve defined by the constraint $b \geq \delta + \lambda \left(1 - \mu\right)$, where lump-sum capital injections are more cost-efficient than perpetual regular transfers. In particular, Figure \ref{TheStochasticControlProblem-Section2-Figure2-a} displays the feasible region for $\lambda$ for different values of the income generation rate $b>0$, which is below or on the curve. Similarly, Figure \ref{TheStochasticControlProblem-Section2-Figure2-b} provides the possible region for $\delta$, for different values of the expected remaining proportion of capital $0<\mu<1$, which is all values below or on this boundary. In particular, Figure \ref{TheStochasticControlProblem-Section2-Figure2-a} shows that, the feasible region for $\lambda$ is smaller as the expected losses are more severe (lower $\mu$). That is, there is a risk trade-off between the frequency and severity of the losses. In other words, for $C(x) \leq D(x)$ to hold, we see that as each individual loss is expected to be less severe (higher $\mu$), then the household can tolerate more frequent losses (higher $\lambda$). Moreover, we observe that a household with low income generation rate $b>0$ must experience very low frequency when losses are more severe. On this basis, we can conclude that for households experiencing frequent and severe losses, providing regular, ongoing transfers is a more cost-effective strategy for the government than offering capital injections, except in cases where households possess strong income-generating capacity. On the other hand, Figure \ref{TheStochasticControlProblem-Section2-Figure2-b} demonstrates that the size of the admissible region for $\delta$ is reduced when losses are more frequent (higher $\lambda$) and increased when the expected losses are less severe (higher $\mu$). Lump-sum transfers occur immediately at the loss event while perpetual regular transfers take place during the time the capital lies below the poverty threshold. That is, lump-sum capital injections are usually more costly today, because making the transfer earlier is more valuable than waiting. Indeed, the larger $\delta$, the more sharply future transfers are discounted and the discounted lump-sum transfers grow much higher relative to the discounted perpetual regular transfers. This is the reason why, for instance, we observe that the inequality $C(x) \leq D(x)$ only holds for very small values of $\delta$ when capital losses are more frequent and severe (boundary displayed with the solid blue line for higher values of $\lambda$).

\begin{remark}\label{TheStochasticControlProblem-Section2-Remark2}
From this point onward, we restrict our analysis to lump-sum capital transfers. Specifically, we focus on the parameter region in which lump-sum transfers are less costly for the government than perpetual regular transfers, that is, when $b \geq \delta + \lambda (1 - \mu)$. The parameter values used in the numerical examples presented in the manuscript satisfy this condition. Within this region, perpetual regular transfers are strictly more expensive than lump-sum transfers.
\end{remark}

In the remainder of the paper, we focus on the problem of optimising, from the government\rq s perspective, the transfer of lump-sum capital so the individuals\rq \ capital level is never below the poverty line. In some cases, we will show that it is more effective to provide transfers that raise capital strictly above the poverty threshold (these strategies are explained in greater detail in Section \ref{ThresholdStrategies-Section4}). This presents a trade-off: injecting more than the minimum lump-sum required to bring the household\rq s capital to $x^{\ast}$ means that a portion of future household income will be allocated to investment (or savings), leading to capital growth over time. Consequently, the government may need to inject less capital during the next loss event. Clearly, as we will also show through examples, the optimality of this approach depends on the model parameters and the cumulative distribution function (c.d.f.) of the remaining proportions of capital $Z_{i}$. To investigate this question, we aim to minimise the expected sum of discounted lump-sum capital transfers, subject to the constraint that the household\rq s capital never falls below the critical threshold $x^{\ast}$. {More precisely, we define a \textit{control strategy} as a process $\pi=(S_{t})_{t\geq0}$ where $S_{t}$ is the cumulative amount of capital transfers up to time $t\geq0$. The control strategy $S_{t}$ is \textit{admissible} if it is adapted with respect to de filtration }$\left( \mathcal{F}_{t}\right)  _{t\geq0}${, non-decreasing and right-continuous}. The set of all admissible control strategies with initial capital $x$ is denoted by $\Pi_{x}$. For any $\pi\in\Pi_{x}$, the controlled capital process $X_{t}^{\pi}$ can be written as

\vspace{0.3cm}

\begin{align}
X_{t}^{\pi}=
\begin{cases}
\left(X_{\tau_{i-1}}^{\pi}-x^{\ast}\right)  e^{r\left(t-\tau_{i-1}\right)}+x^{\ast}+S_{t} & \hspace{0.1cm} \textit{for }\tau_{i-1}\leq t<\tau_{i},\\
X_{\tau_{i}^{-}}^{\pi}\cdot Z_{i}+S_{\tau_{i}}-S_{\tau_{i}^{-}} & \textit{for }t=\tau_{i},
\end{cases}
\label{TheStochasticControlProblem-Section2-Equation21}
\end{align}

\vspace{0.3cm}

with $\tau_{0}=0$ and $X_{0}^{\pi}=x+S_{0}$. The capital lump-sum transfers should keep the household out of the area of poverty; that is, the controlled capital process should satisfy $X_{t}^{\pi}\geq x^{\ast}$. The expected discounted capital transfers of the admissible strategy $\pi=(S_{t})_{t\geq0}\in\Pi_{x}$ for a household with initial capital $x\geq0$, is given by

\vspace{0.3cm}

\begin{align}
V^{\pi}(x)=\mathbb{E}\left[ \int_{0^{-}}^{\infty}e^{-\delta t}dS_{t}\right],
\label{TheStochasticControlProblem-Section2-Equation22}
\end{align}

\vspace{0.3cm}

where $\delta>0$ is the discount factor. Here, $dS_t$, includes the possibility of continuous payments as well as lump-sums. The function \eqref{TheStochasticControlProblem-Section2-Equation22} is also known as the value function of an admissible strategy. We consider the following optimisation problem,

\vspace{0.3cm}

\begin{align}
V(x)=\inf\left\{  V^{\pi}(x):\pi\in\Pi_{x}\right\}  \quad\text{ for }x\geq0.
\label{TheStochasticControlProblem-Section2-Equation23}
\end{align}

\vspace{0.3cm}

Note that $V^{\pi}(x)=V^{\pi}(x^{\ast})+(x^{\ast}-x)$ for $x\leq x^{\ast}$ and so $V(x)=V(x^{\ast})+(x^{\ast}-x)$ for $x\leq x^{\ast}$ and that the cost of social protection $C(x)$ is the value function of an admissible strategy (but the cost of perpetual subsidy $D(x)$ it is not). In particular, $V(x)\leq C(x)$.

\section{Analysis of $V$: Properties and the Hamilton-Jacobi-Bellman Equation} \label{AnalysisofV:PropertiesandtheHamilton-Jacobi-BellmanEquation-Section3}

In this section, we associate a Hamilton-Jacobi-Bellman (HJB) equation to the optimisation problem \eqref{TheStochasticControlProblem-Section2-Equation23} and we prove that the optimal value function $V$ is a \textit{viscosity solution} of this equation.

We begin by obtaining some basic properties of $V$. To this end, we first present the following lemma concerning the cost of social protection. The proof is provided in Appendix \ref{ProofofLemma3.1}.

\vspace{0.3cm}

\begin{lemma}  \label{AnalysisofV:PropertiesandtheHamilton-Jacobi-BellmanEquation-Section3-Lemma1}

$C(x)$ is non-increasing and non-negative with $\lim_{x\rightarrow\infty}C(x)=0.$

\end{lemma}

\vspace{0.3cm}

The following proposition establishes basic properties of the optimal value function.

\vspace{0.3cm}

\begin{proposition} \label{AnalysisofV:PropertiesandtheHamilton-Jacobi-BellmanEquation-Section3-Proposition1}

The function $V$ is non-negative, non-increasing and Lipschitz with $\lim_{x\rightarrow\infty}V(x)=0$.

\end{proposition}

\begin{proof}

Given an initial capital $x\geq x^{\ast}$, one possible admissible strategy $\pi_{1}$ is to inject a capital transfer $h$ immediately and then follow any admissible strategy ${\small \bar{\pi}}_{{\small x+h}}$ with initial capital $x+h,$ then

\vspace{0.3cm}

\begin{align}
V(x)\leq V^{\pi_{1}}(x)=V^{\bar{\pi}}(x+h)+h{\small.}
\label{AnalysisofV:PropertiesandtheHamilton-Jacobi-BellmanEquation-Section3-Equation1}
\end{align}

Thus, it yields


\begin{align}
V(x)\leq V(x+h)+h,
\label{AnalysisofV:PropertiesandtheHamilton-Jacobi-BellmanEquation-Section3-Equation2}
\end{align}

and since $V$ is non-increasing we have the Lipschitz result. Also,

\vspace{0.3cm}

\begin{align}
0\leq\lim_{x\rightarrow\infty}V(x)\leq\lim_{x\rightarrow\infty}C(x)=0.
\label{AnalysisofV:PropertiesandtheHamilton-Jacobi-BellmanEquation-Section3-Equation3}
\end{align}
\end{proof}

The HJB equation of this problem is the following first order integro-differential equation (IDE) with a derivative constraint:

\vspace{0.3cm}

\begin{align}
\min\{1+u^{\prime}(x),\mathcal{L}(u)(x)\}=0 \hspace{0.3cm} \textit{for} \hspace{0.3cm} x\geq x^{\ast},
\label{AnalysisofV:PropertiesandtheHamilton-Jacobi-BellmanEquation-Section3-Equation4}
\end{align}

\vspace{0.3cm}

where

\vspace{0.3cm}

\begin{align}
\mathcal{L(}u)(x)&=r(x-x^{\ast})u^{\prime}(x)-(\lambda+\delta)u(x)+\lambda\int_{0}^{x^{\ast}/x}(u(x^{\ast})+x^{\ast}-x\cdot z)dG_{Z}(z) \\ \\ +&\lambda\int_{x^{\ast}/x}^{1}u(x\cdot z)dG_{Z}(z)\text{, for }x\geq x^{\ast}. 
\label{AnalysisofV:PropertiesandtheHamilton-Jacobi-BellmanEquation-Section3-Equation5}
\end{align}

\vspace{0.3cm}

The above IDE is obtained through an heuristic derivation under the assumption that the optimal value function is sufficiently smooth. Its structure reflects the singular nature of the optimisation problem. Such gradient-constrained variational inequalities are standard in singular stochastic control (see, for instance, \cite{Book:Fleming2006} and \cite{Book:Azcue2014} in the insurance context). When the gradient constraint holds with equality, i.e., $1+u^{\prime }(x)=0$ the gradient constraint is binding and this characterises the action region at capital level $x$, where intervention is optimal by means of instantaneous lump-sum capital injections at unit marginal cost. In contrast, when $\mathcal{L}(u)(x)=0$ the capital level $x$ belongs to the inaction region, where no intervention is optimal and the capital level follows the uncontrolled dynamics. Indeed, $\mathcal{L}(u)$ is the discounted generator of the uncontrolled capital process corresponding to the decision to wait (i.e., not to inject capital). Moreover, this structure can be interpreted as the singular analogue of bang-bang controls in non-singular problems, where optimal strategies switch between extreme actions (full intervention versus no intervention). A rigorous justification of the variational inequality is provided in Proposition \ref{AnalysisofV:PropertiesandtheHamilton-Jacobi-BellmanEquation-Section3-Proposition2}, where we prove that the optimal value function $V$ solves \eqref{AnalysisofV:PropertiesandtheHamilton-Jacobi-BellmanEquation-Section3-Equation4} in the viscosity sense. This places the analysis within the standard viscosity-solution approach to stochastic control in the absence of classical smoothness.

Now, we state the \textit{dynamic programming principle}. We skip the proof because it is similar to the one of Lemma 1.2 in \cite{Book:Azcue2014}.

\vspace{0.3cm}

\begin{lemma} \label{AnalysisofV:PropertiesandtheHamilton-Jacobi-BellmanEquation-Section3-Lemma2}

Given $x\geq0$ \ and any finite stopping time $\tau$ we have

\vspace{0.25cm}

\begin{align}
V(x)=\inf\limits_{\pi=(S_{t})_{t\geq0}\in\Pi_{x}}\mathbb{E}_{x}\left[\int\nolimits_{0^{-}}^{\tau}e^{-\delta t}dS_{t}+V(X_{\tau}^{\pi})e^{-\delta\tau}\right].
\label{AnalysisofV:PropertiesandtheHamilton-Jacobi-BellmanEquation-Section3-Equation6}
\end{align}

\vspace{0.3cm}

\end{lemma}

\begin{definition} \label{AnalysisofV:PropertiesandtheHamilton-Jacobi-BellmanEquation-Section3-Definition1}

We say that a locally Lipschitz function $\underline{u}:[x^{\ast},\infty)\rightarrow\mathbbm{R}$ is a \textit{viscosity subsolution} of \eqref{AnalysisofV:PropertiesandtheHamilton-Jacobi-BellmanEquation-Section3-Equation4} at $x\in(x^{\ast},\infty)$ if any continuously differentiable function $\psi:(x^{\ast},\infty)\rightarrow\mathbbm{R}$ with $\psi(x)=$ $\underline{u}(x)$ such that $\underline{u}-\psi$ reaches the minimum at $x$ satisfies

\vspace{0.3cm}

\begin{align}
\min\{1+\psi^{\prime}(x),\mathcal{L}(\psi)(x)\}\leq0.
\label{AnalysisofV:PropertiesandtheHamilton-Jacobi-BellmanEquation-Section3-Equation7}
\end{align}

\vspace{0.3cm}

The function $\psi$ is called a test function for subsolution at $x.$

We say that a locally Lipschitz function $\overline{u}$ $:[x^{\ast},\infty)\rightarrow\mathbbm{R}$ is a \textit{viscosity supersolution} of \eqref{AnalysisofV:PropertiesandtheHamilton-Jacobi-BellmanEquation-Section3-Equation4} at $x\in(x^{\ast},\infty)$ if any continuously differentiable function $\varphi:(x^{\ast},\infty)\rightarrow\mathbbm{R}$ with $\varphi(x)=$ $\overline{u}(x)$ such that $\overline{u}-\varphi$ reaches the maximum at $x$ satisfies

\vspace{0.3cm}

\begin{align}
\min\{1+\varphi^{\prime}(x),\mathcal{L}(\varphi)(x)\}\geq0.
\label{AnalysisofV:PropertiesandtheHamilton-Jacobi-BellmanEquation-Section3-Equation8}
\end{align}

\vspace{0.3cm}

The function $\varphi$ is called a test function for supersolution at $x.$

Finally, we say that a locally Lipschitz function $u$ $:[x^{\ast},\infty)\rightarrow\mathbbm{R}$ is a \textit{viscosity solution} of \eqref{AnalysisofV:PropertiesandtheHamilton-Jacobi-BellmanEquation-Section3-Equation4} if it is both a viscosity subsolution and a viscosity supersolution at any $x\in(x^{\ast},\infty)$.

\end{definition}

\vspace{0.3cm}

\begin{lemma} \label{AnalysisofV:PropertiesandtheHamilton-Jacobi-BellmanEquation-Section3-Lemma3}

Take $u:[x^{\ast},\infty)\rightarrow\lbrack0,\infty)$ continuously differentiable, let us extend $u$ to $[0,\infty)$ as $u(x)=u(x^{\ast})+x^{\ast}-x$. Given $\pi\in\Pi_{x}$, we can write for any finite stopping time $\tau^{\ast},$

\vspace{0.3cm}

\begin{align}
\begin{array}
[c]{lll}
u(X_{\tau^{\ast}}^{\pi})e^{-\delta\tau^{\ast}}-u(x) & = & \int\nolimits_{0}^{\tau^{\ast}}\mathcal{L}(u)(X_{t^{-}}^{\pi})e^{-\delta t}dt-\int_{0^{-}}^{\tau^{\ast}}e^{-\delta t}dS_{t}\\ \\
&  & +\int\nolimits_{0}^{\tau^{\ast}}(1+u^{\prime}(X_{t^{-}}^{\pi}))e^{-\delta t}dS_{t}^{c}\\ \\
&  & +\sum\limits_{\substack{S_{t}\neq S_{t^{-}}\\t\leq\tau^{\ast}}}e^{-\delta t}\left(  \int\nolimits_{0}^{S_{t}-S_{t^{-}}}\left(  1+u^{\prime}(X_{t}^{\pi}-\alpha)\right)  d\alpha\right)  +M_{\tau^{\ast}},
\end{array}
\label{AnalysisofV:PropertiesandtheHamilton-Jacobi-BellmanEquation-Section3-Equation9}
\end{align}

\vspace{0.3cm}

where

\vspace{0.3cm}

\begin{align}
\begin{array}
[c]{cl}
M_{T}= & \sum\limits_{\tau_{i}\leq T}\left(u(Z_{i} \cdot X_{\tau_{i}^{-}}^{\pi})-u(X_{\tau_{i}^{-}}^{\pi})\right)e^{-\delta t} \\ \\
& -\lambda\int\nolimits_{0}^{T}\left(  \int\nolimits_{x^{\ast}/X_{t^{-}}^{\pi}}^{1}\left(u(z\cdot X_{t^{-}}^{\pi})-u(X_{t^{-}}^{\pi})\right)dG_{Z}(z)\right)  e^{-\delta t}dt\\ \\
& -\lambda\int\nolimits_{0}^{T}\left(  \int\nolimits_{0}^{x^{\ast}/X_{t^{-}}^{\pi}}\left(  x^{\ast}-z\cdot X_{t^{-}}^{\pi}+u(x^{\ast})-u(X_{t^{-}}^{\pi})\right)  dG_{Z}(z)\right)e^{-\delta t}dt,
\end{array}
\label{AnalysisofV:PropertiesandtheHamilton-Jacobi-BellmanEquation-Section3-Equation10}
\end{align}

\vspace{0.3cm}

is a martingale with zero expectation.

\end{lemma}

The proof of Lemma \ref{AnalysisofV:PropertiesandtheHamilton-Jacobi-BellmanEquation-Section3-Lemma3} is given in Appendix \ref{ProofofLemma3.3}.

\vspace{0.3cm}

\begin{proposition} \label{AnalysisofV:PropertiesandtheHamilton-Jacobi-BellmanEquation-Section3-Proposition2}

$V$ is a viscosity solution of \eqref{AnalysisofV:PropertiesandtheHamilton-Jacobi-BellmanEquation-Section3-Equation4} in $(x^{\ast},\infty)$.

\end{proposition}

\begin{proof}
	We first prove that $V$ is a viscosity supersolution of \eqref{AnalysisofV:PropertiesandtheHamilton-Jacobi-BellmanEquation-Section3-Equation4}. Given an initial capital $x > x^{\ast}$, let $h>0$ and consider the admissible strategy $\pi $ that does not inject capital unless the capital falls below the poverty line, in which case it restores the capital exactly to $x^{\ast}$. This strategy corresponds to the cost of social protection $C(x)$. Let $X_{t}$ denote the uncontrolled process with initial value $x,$ and $X_{t}^{\pi }$ the controlled process under $\pi $. Let $\tau _{1}$ be the first jump time. By the dynamic programming principle (Lemma \ref{AnalysisofV:PropertiesandtheHamilton-Jacobi-BellmanEquation-Section3-Lemma2}), we have
	
\begin{align}
V(x)\leq \mathbb{E}_{x}\left[ e^{-\delta (\tau _{1}\wedge h)}V(X_{\tau_{1}\wedge h}^{\pi })+e^{-\delta \tau _{1}}(X_{\tau _{1}}-x^{\ast })\mathbbm{1}_{\left\{ X_{\tau _{1}}<x^{\ast },\tau _{1}\leq h\right\} }\right],
\label{AnalysisofV:PropertiesandtheHamilton-Jacobi-BellmanEquation-Section3-Equation11}
\end{align}

where

\begin{align}
X_{\tau _{1}\wedge h}^{\pi }=x^{\ast }\mathbbm{1}_{\left\{ X_{\tau_{1}}<x^{\ast },\tau _{1}\leq h\right\} }+X_{\tau _{1}\wedge h}(\mathbbm{1}_{\left\{ X_{\tau _{1}}\geq x^{\ast },\tau _{1}\leq h\right\} }+\mathbbm{1}_{\left\{ \tau _{1}>h\right\} }).
\label{AnalysisofV:PropertiesandtheHamilton-Jacobi-BellmanEquation-Section3-Equation12}
\end{align}

Recall that, by definition,

\begin{align}
V(z)=V(x^{\ast })+x^{\ast }-z,
\label{AnalysisofV:PropertiesandtheHamilton-Jacobi-BellmanEquation-Section3-Equation13}
\end{align}

for $0\leq z<x^{\ast }$. Hence, inequality \eqref{AnalysisofV:PropertiesandtheHamilton-Jacobi-BellmanEquation-Section3-Equation11} can be rewritten as

\begin{align}
V(x)\leq \mathbb{E}_{x}\left[ e^{-\delta (\tau _{1}\wedge h)}V(X_{\tau_{1}\wedge h})\right].
\label{AnalysisofV:PropertiesandtheHamilton-Jacobi-BellmanEquation-Section3-Equation14}
\end{align}

Let $\varphi$ be a test function for supersolution \eqref{AnalysisofV:PropertiesandtheHamilton-Jacobi-BellmanEquation-Section3-Equation4} at $x$. By definition of test function, $\varphi$ is continuously differentiable, satisfies $\varphi (x)=V(x)$ and $V(z)\geq \varphi (z)$ for all $z\geq x^{\ast }.$ Extending $\varphi $ to $[0,\infty )$ as $\varphi(z)=\varphi (x^{\ast })+x^{\ast }-z$ for $0\leq z<x^{\ast }$ we obtain,

\begin{align}
\begin{array}{lll}
\varphi (x) & = & V(x) \\ 
& \leq  & \mathbb{E}_{x}\left[ e^{-\delta (\tau _{1}\wedge h)}V(X_{\tau_{1}\wedge h})\right]  \\ 
& \leq  & \mathbb{E}_{x}[e^{-\delta (\tau_{1}\wedge h)}\varphi (X_{\tau _{1}\wedge h})].
\end{array}
\label{AnalysisofV:PropertiesandtheHamilton-Jacobi-BellmanEquation-Section3-Equation155}
\end{align}

Hence, from the expression of the discounted infinitesimal generator given in \eqref{TheStochasticControlProblem-Section2-Equation7}, we have

\begin{align}
\lim_{h\rightarrow 0^{+}}\frac{\mathbb{E}_{x}[\varphi (X_{\tau _{1}\wedge h})e^{-\delta (\tau _{1}\wedge h)}]-\varphi (x)}{h}=\mathcal{L}(\varphi)(x)\geq 0.
\label{AnalysisofV:PropertiesandtheHamilton-Jacobi-BellmanEquation-Section3-Equation166}
\end{align}

This establishes the first inequality.

To obtain the gradient constraint, fix an initial capital level $x>x^{\ast }$, and consider an admissible strategy that makes an immediate lump-sum capital injection of size $l>0$ after which a near-optimal admissible strategy is followed. By definition of the value function we get,

\begin{align}
\varphi (x)=V(x)\leq l+V(x+l)\leq l+\varphi (x+l).
\label{AnalysisofV:PropertiesandtheHamilton-Jacobi-BellmanEquation-Section3-Equation177}
\end{align}

Hence, $\varphi (x+l)-\varphi (x)\geq -l$ . Dividing by $l$ and taking $l\searrow 0,$ we obtain

\begin{align}
\varphi ^{\prime }(x)+1\geq 0.
\label{AnalysisofV:PropertiesandtheHamilton-Jacobi-BellmanEquation-Section3-Equation188}
\end{align}

Combining \eqref{AnalysisofV:PropertiesandtheHamilton-Jacobi-BellmanEquation-Section3-Equation166} and \eqref{AnalysisofV:PropertiesandtheHamilton-Jacobi-BellmanEquation-Section3-Equation188}, we conclude that $V$ is a viscosity supersolution of \eqref{AnalysisofV:PropertiesandtheHamilton-Jacobi-BellmanEquation-Section3-Equation4} at $x$. We omit the proof that $V$ is a viscosity subsolution since it is very similar to the one of Proposition 3.1 from \cite{Book:Azcue2014}.
\end{proof}

We now present the following lemma (the proof follows along the same lines as that of Lemma 4.2 from \cite{Book:Azcue2014}, with only minor modifications; we therefore omit it):

\vspace{0.3cm}

\begin{lemma} \label{AnalysisofV:PropertiesandtheHamilton-Jacobi-BellmanEquation-Section3-Lemma4} 

Let $\overline{u}$\ be a non-increasing supersolution of \eqref{AnalysisofV:PropertiesandtheHamilton-Jacobi-BellmanEquation-Section3-Equation4} with $\lim_{x\rightarrow\infty}\overline{u}(x)=0$.\ We can find a sequence of positive functions $\overline{u}_{n}:\mathbb{[}x^{\ast},\infty\mathbb{)}\rightarrow\mathbb{[}0,\infty\mathbb{)}$ such that:

\begin{itemize}

\item[(a)] $\overline{u}_{n}$\ is continuously differentiable, non-increasing with $\lim_{x\rightarrow\infty}\overline{u}_{n}(x)=0$, $\overline{u}_{n}^{\prime}(x)+1\geq0$, $\overline{u}_{n}$\ $\nearrow$ $\overline{u}$\ uniformly and $\overline{u}_{n}^{\prime}(x)$\ converges to $\overline{u}^{\prime}(x)$\ a.e.

\item[(b)] Given any $K\geq x^{\ast}$, there exists a sequence $c_{n}\geq0$ with $\lim_{n\rightarrow\infty}c_{n}\searrow0$ such that $\lambda \overline{u}(0)\geq\mathcal{L}(\overline{u}_{n})(x)\geq-c_{n}$ for $x^{\ast}\leq x\leq K$. 
\end{itemize}

\end{lemma}

\vspace{0.5cm}

\begin{proposition} \label{AnalysisofV:PropertiesandtheHamilton-Jacobi-BellmanEquation-Section3-Proposition3} 

The optimal value function defined in \eqref{TheStochasticControlProblem-Section2-Equation23} is the largest non-increasing viscosity supersolution of \eqref{AnalysisofV:PropertiesandtheHamilton-Jacobi-BellmanEquation-Section3-Equation4}  with limit zero as $x$ goes to infinity.

\end{proposition}

\begin{proof}
	Given $x$ and $\varepsilon>0$, there exists an admissible strategy $\pi_{1}=(S_{t}^{1})_{t\geq0}\in\Pi_{x}$ such that $V^{\pi_{1}}(x)\leq V(x)+\frac{\varepsilon}{2}$. By Lemma \ref{AnalysisofV:PropertiesandtheHamilton-Jacobi-BellmanEquation-Section3-Lemma1}, there exists $\overline{x}>x$ large enough such that the cost of social protection $C(\overline{x})<\frac{\varepsilon}{2}$. Then, we consider

\begin{align}
\tau _{\overline{x}}=\inf \{t:X_{t}^{\pi _{1}}\geq \overline{x}\},
\label{AnalysisofV:PropertiesandtheHamilton-Jacobi-BellmanEquation-Section3-Equation15}
\end{align}

and define the new strategy $\overline{\pi }=\left( \overline{S}_{t}\right)_{t\geq 0}\in \Pi _{x}$, which coincides with $\pi _{1}$ for $t\leq \tau_{\overline{x}}$, and subsequently provides injections of capital up to $x^{\ast}$ whenever the capital is below $x^{\ast }$. Hence,

\begin{align}
V(x)\geq V^{\pi _{1}}(x)-\frac{\varepsilon }{2}=\mathbb{E}_{x}\left[\int_{0^{-}}^{\tau _{\overline{x}}}e^{-\delta t}dS_{t}^{1}+e^{-\delta \tau _{\overline{x}}}V^{\pi _{1}}(\overline{x})\right] -\frac{\varepsilon }{2}\geq\mathbb{E}_{x}\left[ \int_{0^{-}}^{\tau _{\overline{x}}}e^{-\delta t}d\overline{S}_{t}\right] -\frac{\varepsilon }{2}.
\label{AnalysisofV:PropertiesandtheHamilton-Jacobi-BellmanEquation-Section3-Equation16}
\end{align}

We also have, since $C(\overline{x})<\frac{\varepsilon }{2}$ and $X_{_{\tau_{\overline{x}}}}^{\overline{\pi }}=\overline{x}$ (if $\tau_{\overline{x}}<\infty $) that,

\begin{align}
\mathbb{E}_{x}\left[ \int_{0^{-}}^{\tau _{\overline{x}}}e^{-\delta t}d\overline{S}_{t}\right] & \geq \mathbb{E}_{x}\left[ \int_{0^{-}}^{\tau _{\overline{x}}}e^{-\delta t}d\overline{S}_{t}+e^{-\delta \tau _{\overline{x}}}\left( C(\overline{x})-\frac{\varepsilon }{2}\right) \right] \\
& \\
& \geq V^{\overline{\pi }}(x)-\frac{\varepsilon }{2}.
\label{AnalysisofV:PropertiesandtheHamilton-Jacobi-BellmanEquation-Section3-Equation17}
\end{align}

Therefore, we have that $V(x)\geq V^{\overline{\pi }}(x)-\varepsilon $. Let us now demonstrate that $\overline{u}(x)\leq V(x)$, by considering $\overline{u}$, a non-increasing supersolution of \eqref{AnalysisofV:PropertiesandtheHamilton-Jacobi-BellmanEquation-Section3-Equation4} that tends to zero as $x$ goes to infinity (so in particular $\overline{u}(x)$ is bounded by above). We know that the following inequality holds,

\begin{align}
X_{t}^{^{\overline{\pi }}}\leq \left( \overline{x}-x^{\ast }\right)e^{r\left( t-\tau _{\overline{x}}\right) }+\overline{x},
\label{AnalysisofV:PropertiesandtheHamilton-Jacobi-BellmanEquation-Section3-Equation18}
\end{align}

for $t>$ $\tau _{\overline{x}}$. Then, we take $k>0$ and define

\begin{align}
m\left( k\right) :=\left( \overline{x}-x^{\ast }\right) e^{rk}+\overline{x}.
\label{AnalysisofV:PropertiesandtheHamilton-Jacobi-BellmanEquation-Section3-Equation19}
\end{align}

We additionally consider the stopping time $T^{k}:=\tau _{\overline{x}}+k$, for the case in which $\tau _{\overline{x}}<\infty $ and $T^{k}=\infty $ otherwise. Thus, this implies that $X_{t}^{^{\overline{\pi }}}\in \lbrack x^{\ast }, m(k)]$ for all $t\leq T^{k}$. Since the functions $\overline{u}_{n}$ defined in Lemma \ref{AnalysisofV:PropertiesandtheHamilton-Jacobi-BellmanEquation-Section3-Lemma4} are continuously differentiable, we obtain using Lemma \ref{AnalysisofV:PropertiesandtheHamilton-Jacobi-BellmanEquation-Section3-Lemma3} and taking any $s>0,$

\begin{align}
\overline{u}_{n}(X_{T^{k}\wedge s}^{\overline{\pi }})e^{-\delta \left(T^{k}\wedge s\right) }-\overline{u}_{n}(x)\geq\int\nolimits_{0}^{T^{k}\wedge s}\mathcal{L}(\overline{u}_{n})(X_{t^{-}}^{\overline{\pi }})e^{-\delta t}dt-\int\nolimits_{0^{-}}^{T^{k}\wedge s}e^{-\delta t}d\overline{S}_{t}+M_{T^{k}\wedge s},
\label{AnalysisofV:PropertiesandtheHamilton-Jacobi-BellmanEquation-Section3-Equation20}
\end{align}

where $\left( M_{T^{k}\wedge s}\right) _{T\geq 0}$ is a zero-expectation martingale. From Lemma \ref{AnalysisofV:PropertiesandtheHamilton-Jacobi-BellmanEquation-Section3-Lemma4}--(b), we have

\begin{align}
\lambda \overline{u}(0)\geq \mathcal{L}(\overline{u}_{n})(x)\geq -c_{n}.
\label{AnalysisofV:PropertiesandtheHamilton-Jacobi-BellmanEquation-Section3-Equation21}
\end{align}

Then, using the bounded convergence theorem and taking $n\rightarrow \infty $, it yields

\begin{align}
\mathbb{E}_{x}\left[ \overline{u}(X_{\left( T^{k}\wedge s\right) }^{\overline{\pi }})e^{-\delta \left( T^{k}\wedge s\right) }\right] -\overline{u}(x)\geq -\mathbb{E}_{x}\left[ \int\nolimits_{0^{-}}^{T^{k}\wedge s}e^{-\delta t}d\overline{S}_{t}\right].
\label{AnalysisofV:PropertiesandtheHamilton-Jacobi-BellmanEquation-Section3-Equation22}
\end{align}

Let $s\rightarrow \infty $. Then, $T^{k}\wedge s\nearrow T^{k}$ as $s\rightarrow \infty $. Also, $\overline{S}$ is a non-decreasing process, so by the monotone convergence theorem,

\begin{align}
\lim_{s\rightarrow \infty }\mathbb{E}_{x}\left[ \int\nolimits_{0^{-}}^{T^{k}\wedge s}e^{-\delta t}d\overline{S}_{t}\right] =\mathbb{E}_{x}\left[\int\nolimits_{0^{-}}^{T^{k}}e^{-\delta t}d\overline{S}_{t}\right].
\label{AnalysisofV:PropertiesandtheHamilton-Jacobi-BellmanEquation-Section3-Equation23}
\end{align}

Moreover, since $\overline{u}\geq 0$ and is bounded (because $\overline{u}$ is non-increasing and $\lim_{x\rightarrow \infty }\overline{u}(x)=0$), we may apply the bounded convergence theorem to obtain

\begin{align}
\lim_{s\rightarrow \infty }\mathbb{E}_{x}\left[ \overline{u}(X_{\left(T^{k}\wedge s\right) }^{\overline{\pi }})e^{-\delta \left( T^{k}\wedge s\right) }\right] =\mathbb{E}_{x}\left[ \overline{u}(X_{T^{k}}^{\overline{\pi }})e^{-\delta T^{k}}\right].
\label{AnalysisofV:PropertiesandtheHamilton-Jacobi-BellmanEquation-Section3-Equation244}
\end{align}

Therefore, combining \eqref{AnalysisofV:PropertiesandtheHamilton-Jacobi-BellmanEquation-Section3-Equation23} and \eqref{AnalysisofV:PropertiesandtheHamilton-Jacobi-BellmanEquation-Section3-Equation244} we obtain,

\begin{align}
\overline{u}(x)\leq \mathbb{E}_{x}\left[ \overline{u}(X_{T^{k}}^{\overline{\pi }})e^{-\delta T^{k}}\right] +\mathbb{E}_{x}\left[ \int\nolimits_{0^{-}}^{T^{k}}e^{-\delta t}d\overline{S}_{t}\right].
\label{AnalysisofV:PropertiesandtheHamilton-Jacobi-BellmanEquation-Section3-Equation255}
\end{align}

We now let $k\rightarrow \infty $ (so that $T^{k}\rightarrow \infty $). Again, since the process $\overline{S}$ is non-decreasing, by the monotone convergence theorem, we get

\begin{align}
\lim_{k\rightarrow \infty }\mathbb{E}_{x}\left[ \int\nolimits_{0}^{T^{k}}e^{-\delta t}d\overline{S}_{t}\right] =\mathbb{E}_{x}\left[ \int\nolimits_{0^{-}}^{\infty }e^{-\delta t}d\overline{S}_{t}\right]=V^{\overline{\pi }}(x).
\label{AnalysisofV:PropertiesandtheHamilton-Jacobi-BellmanEquation-Section3-Equation266}
\end{align}

Furthermore, as $T^{k}\rightarrow \infty$, and since $\overline{u}\geq 0$ and is bounded by above, the bounded convergence theorem implies that,

\begin{align}
\lim_{k\rightarrow \infty }\mathbb{E}_{x}\left[ \overline{u}(X_{T^{k}}^{\overline{\pi }})e^{-\delta T^{k}}\right] =0.
\label{AnalysisofV:PropertiesandtheHamilton-Jacobi-BellmanEquation-Section3-Equation277}
\end{align}

Finally, combining \eqref{AnalysisofV:PropertiesandtheHamilton-Jacobi-BellmanEquation-Section3-Equation255}, \eqref{AnalysisofV:PropertiesandtheHamilton-Jacobi-BellmanEquation-Section3-Equation266} and \eqref{AnalysisofV:PropertiesandtheHamilton-Jacobi-BellmanEquation-Section3-Equation277}, we conclude that $\overline{u}(x)\leq V^{\overline{\pi }}(x)$. Since we have already shown that $V(x)\geq V^{\overline{\pi }}(x)-\varepsilon $ for arbitrary $\varepsilon >0$, it follows that $\overline{u}(x)\leq V(x)$ and that completes the proof.
\end{proof}

The following verification theorem shows that if the value function associated with an admissible strategy is a viscosity supersolution of the HJB equation \eqref{AnalysisofV:PropertiesandtheHamilton-Jacobi-BellmanEquation-Section3-Equation4}, then this function coincides with the optimal value function and the corresponding admissible strategy is optimal. In particular, the theorem does not require checking whether this function is the largest viscosity supersolution. Instead, the result of Proposition \ref{AnalysisofV:PropertiesandtheHamilton-Jacobi-BellmanEquation-Section3-Proposition3}, combined with the definition of the optimal value function, yields the desired identification. This argument does not rely on a comparison principle or on uniqueness of viscosity solutions and it is our main tool for identifying the correct value function in the present setting, where a uniqueness result for viscosity solutions is not available, since the standard comparison principle fails for this integro-differential HJB equation. Verification theorems of this type are a standard way to overcome the lack of uniqueness and to select the relevant solution among the viscosity solutions. This approach is well established in the insurance literature on optimisation problems driven by compound Poisson processes; see, for instance, \cite[Chapter 5]{Book:Azcue2014}, where the value function is characterised through a verification argument without relying on a global comparison principle or a uniqueness result.

\vspace{0.3cm}

\begin{theorem} \label{AnalysisofV:PropertiesandtheHamilton-Jacobi-BellmanEquation-Section3-Theorem1}
Consider a family of admissible strategies $(\pi ^{x})_{x\geq x^{\ast }}$ such that $\pi ^{x}\in \Pi _{x}$ for any initial surplus $x\geq x^{\ast }$. If the function $V^{\pi ^{x}}(x)$ is a non-increasing viscosity supersolution of \eqref{AnalysisofV:PropertiesandtheHamilton-Jacobi-BellmanEquation-Section3-Equation4} with $\lim_{x\rightarrow \infty }V^{\pi ^{x}}(x)=0$, then $V^{\pi ^{x}}(x)$ is the optimal value function.
\end{theorem}

\begin{proof}
By definition of the optimal value function, we have $V^{\pi ^{x}}(x)\geq V(x)$ for all $x\geq 0$. Moreover, both $V$ and $V^{\pi ^{x}}$ are non-increasing functions and satisfy that the limit is zero as $x$ goes to infinity. Since $V^{\pi ^{x}}$is, by assumption a viscosity supersolution of \eqref{AnalysisofV:PropertiesandtheHamilton-Jacobi-BellmanEquation-Section3-Equation4}, Proposition \ref{AnalysisofV:PropertiesandtheHamilton-Jacobi-BellmanEquation-Section3-Proposition3} implies that  $V^{\pi ^{x}}(x)\leq V(x)$ because $V$ is the largest viscosity supersolution within the class of non-increasing functions vanishing at infinity. Combining both inequalities, we conclude that $V^{\pi^{x}}(x)=V(x)$, which completes the proof.
\end{proof}

\vspace{0.3cm}

The way in which $V$\ solves the HJB equation gives us the optimal strategy for any capital level $x\geq x^{\ast}$. Roughly speaking, we have the following:

\vspace{0.3cm}

\begin{enumerate}

\item $V^{\prime}(x)+1=0$: Provide capital transfers and;

\item $\mathcal{L}(V)(x)=0$: Do not provide capital transfers.

\end{enumerate}

\vspace{0.5cm}

\begin{definition} \label{AnalysisofV:PropertiesandtheHamilton-Jacobi-BellmanEquation-Section3-Definition2}
Suppose there exists a closed set $B=\{x:V^{{\small \prime}}(x)+1=0\}\subset\lbrack x^{\ast},\infty)$\ such that the optimal strategy satisfies: if the capital $x\in B,$ the capital transfer is $y(x)-x$, where${\small}$

\begin{align}
{\small y(x)=}\max{\small \{y>x\geq x^{\ast}:V(y)-V(x)+(y-x)=0\}},
\label{AnalysisofV:PropertiesandtheHamilton-Jacobi-BellmanEquation-Section3-Equation24}
\end{align}

whereas if the capital $x\notin B,$ no transfer is paid up to the first time where the capital process exits the closed set $B$. These strategies are called band strategies with action zone $B\ $and non-action zone $C=[0,\infty)-B$. Note that $[0,x^{\ast}]\subset B.$
\end{definition}

\section{Threshold Strategies} \label{ThresholdStrategies-Section4}

In this section, we restrict our attention to a specific class of admissible strategies, namely threshold transfer strategies, and analyse their associated value functions. In view of the structure of the HJB equation, these transfer strategies form the simplest class of admissible controls and therefore provide a natural starting point for the analysis and are the natural candidates for optimality. The analysis in this section provides a detailed characterisation of the value functions generated by such strategies.

A threshold transfer strategy with threshold $y\geq x^{\ast}$ is defined as a programme in which the government provides a lump-sum transfer of amount $y-x$ whenever the household\rq s capital falls below the threshold $y$. Under such programme, the capital is immediately raised to the threshold level $y$. Conversely, if the household\rq s capital lies above the threshold $y$, no transfer is granted. This strategy, when $y>x^{\ast}$, seeks to maintain a buffer above the critical capital $x^{\ast}$, allowing households to grow their capital. Let us denote with $\pi_{x}^{y}$ $\in\Pi_{x}$ the admissible strategy with threshold $y$ and $V_{y}(x)$ its corresponding value function. Note that, in particular, $V_{x^{\ast}}(x)=C(x)$. We now develop and formalise the properties of this family of admissible strategies.

In the following three lemmas we examine basic properties of the function $V_{y}$. The first two lemmas establish fundamental properties, while the third lemma addresses a continuity result in the special case $y=x^{\ast}$.

\vspace{0.3cm}

\begin{lemma} \label{ThresholdStrategies-Section4-Lemma1}
Given $y\geq x^{\ast}$ the function $V_{y}(x)$ is bounded, non-increasing with respect to the variable $x$ for $x \geq y$ and $V_{y}(x)=V_{y}(y)+y-x$ for $x<y$.
\end{lemma}

\vspace{0.3cm}

\begin{lemma} \label{ThresholdStrategies-Section4-Lemma2}
The function $V_{y}(x)$ for threshold $y>x^{\ast}$ is Lipschitz with respect to the variable $x$ in $[y, +\infty)$ and $C(x)=V_{x^{\ast}}(x)$ is Lipschitz with respect to the variable $x$ in any set $[w, \infty)$ with $w > x^{*}$.
\end{lemma}

\vspace{0.3cm}

\begin{lemma} \label{ThresholdStrategies-Section4-Lemma3}
The function $C(x)$ is continuous.
\end{lemma}

\vspace{0.3cm}

The proofs of these lemmas are provided in Appendices \ref{ProofofLemma4.1}, \ref{ProofofLemma4.2} and \ref{ProofofLemma4.3}, respectively. On the other hand, the proof of Proposition \ref{ThresholdStrategies-Section4-Proposition1}, which states properties of $V_{y}$, is omitted as it is essentially identical to that of Lemma \ref{AnalysisofV:PropertiesandtheHamilton-Jacobi-BellmanEquation-Section3-Lemma1}.

\vspace{0.3cm}

\begin{proposition} \label{ThresholdStrategies-Section4-Proposition1}
The function $V_{y}$ satisfies $\lim_{x\rightarrow\infty}V_{y}(x)=0.$
\end{proposition}

\vspace{0.3cm}

Let us define

\begin{align}
\mathcal{L}^{y}\mathcal{(}W)(x)& :=r(x-x^{\ast})W^{\prime}(x)-(\delta+\lambda)W(x)+\lambda\int_{y/x}^{1}W(x\cdot z)dG_{Z}(z) \\ \\ 
&+ \lambda\int_{0}^{y/x}\left(\left(y-x\cdot z\right)+ W(y)\right)dG_{Z}(z),
\label{ThresholdStrategies-Section4-Equation1}
\end{align}

and the associated IDE,

\begin{align}
\mathcal{L}^{y}\mathcal{(}W)(x)=0.
\label{ThresholdStrategies-Section4-Equation2}
\end{align}


We now show that the value function $V_{y}$ associated with the threshold transfer strategy that maintains the capital at or above level $y$ is the unique solution of Equation \eqref{ThresholdStrategies-Section4-Equation2} for $x>y$, subject to the value-matching condition at the threshold $W(y)=V_{y}(y)$ and with limit zero as the capital goes to infinity. The value-matching condition connects the continuation value before any transfer is required ($x>y$) with the value obtained after the threshold strategy restores the capital to the level $y$. The value $V_{y}(y)$ is therefore not prescribed in advance and should not be interpreted as a Dirichlet boundary condition. Rather, it is determined jointly with $W$ through this matching relation. When a closed-form solution of Equation \eqref{ThresholdStrategies-Section4-Equation2} is not available, this value will be approximated in the numerical examples using Monte Carlo simulations (as it is explained in Section \ref{GeneralCaseAnalysis:AbsenceofClosed-FormSolutions-Section6}). In Section \ref{Closed-FormSolutionsinaSpecialCase-Section5}, since $Z_{i\text{ }}$ is assumed to follow a particular case of the Beta distribution, a closed-form solution is available, which allows us to compute $V_{y}(y)$ explicitly (see Proposition \ref{Closed-FormSolutionsinaSpecialCase-Section5-Proposition1}). When the threshold is $y = x^{\ast }$, since $V_{x^{\ast }}(x)=C(x)$ this value can be obtained directly regardless the distribution of $Z_{i}$, as explained in Subsection \ref{TheStochasticControlProblem-Section2-Subsection22}.

\vspace{0.3cm}

\begin{proposition} \label{ThresholdStrategies-Section4-Proposition2}

Fix $y\geq x^{\ast }$, the value function $V_{y}$ associated with the threshold strategy at level $y$ is given by

\begin{align}
V_{y}(x)=%
\begin{cases}
\left( y-x\right) +V_{y}(y) \hspace{0.3cm} \text{for}\ 0<x\leq y, \\ 
W(x) \hspace{1.7cm} \text{for }x>y,%
\end{cases}%
\label{ThresholdStrategies-Section4-Equation3}
\end{align}

where $W$ is the unique classical solution of \eqref{ThresholdStrategies-Section4-Equation2} on $(y,\infty)$ satisfying $\lim_{x\to\infty}W(x)=0$ and the value-matching condition $W(y)=V_{y}(y)$.



If $y=x^{\ast }$ and $x=x^{\ast },$ then $V_{x^{\ast }}=C$ and in this case, although $C^{\prime }(x^{\ast })$ might not exist, the boundary value $C(x^{\ast })$ satisfies

\begin{align}
-(\delta +\lambda )C(x^{\ast })+\lambda \int_{0}^{1}\left( \left(x^{\ast}-x^{\ast}\cdot z\right) +C(x^{\ast })\right) dG_{Z}(z)=0.
\label{ThresholdStrategies-Section4-Equation55}
\end{align}

\end{proposition}

\begin{proof}

Fix a threshold level $y\geq x^{\ast}$. Let us first prove that $V_{y}$ is a classical solution of the IDE \eqref{ThresholdStrategies-Section4-Equation2}. Consider an initial capital $x\geq y $ such that $V_{y}^{\prime }(x)$ exists. Under the threshold strategy, no transfers are made before the first down-crossing below the threshold $y$ caused by a shock. Let $\tau_{1}$ be the first shock time. We denote by $X_{\tau_{1}}$ the uncontrolled capital level immediately after the shock and before any possible transfer. If the post-shock level satisfies $X_{\tau_1}>y$, the process continues without intervention, whereas if $X_{\tau_1}<y$, the threshold strategy immediately injects the amount $y-X_{\tau_1}$ and the controlled process is restored to the level $y$. Hence, by the dynamic programming principle, we have

\begin{align}
V_{y}(x)=\mathbb{E}\left[ V_{y}\left( X_{h\wedge \tau _{1}}\right)e^{-\delta \left( h\wedge \tau _{1}\right) }\right],
\label{ThresholdStrategies-Section4-Equation5}
\end{align}

for any $h>0$. Then, by (\ref{TheStochasticControlProblem-Section2-Equation7}), we get

\begin{align}
\begin{array}{lll}
0 & = & \lim_{h\rightarrow 0^{+}}\frac{\mathbb{E}\left[ V_{y}\left(X_{h\wedge \tau _{1}}\right) e^{-\delta \left( h\wedge \tau _{1}\right) }\right] -V_{y}(x)}{h} \\ 
& = & r\left( x-x^{\ast }\right) V_{y}^{\prime }(x)-(\lambda +\delta)V_{y}(x)+\lambda \int_{0}^{1}V_{y}(x\cdot z)dG_{Z}(z).
\end{array}
\label{ThresholdStrategies-Section4-Equation6}
\end{align}

Thus, at every point $x \in (y, \infty)$ where $V_{y}$ is differentiable, $V_{y}$ satisfies the IDE \eqref{ThresholdStrategies-Section4-Equation2}, and we can write
\begin{align}
V_{y}^{\prime }(x)=\frac{(\lambda +\delta )V_{y}(x)-\lambda\int_{0}^{1}V_{y}(x\cdot z)dG_{Z}(z)}{r\left( x-x^{\ast }\right)}.
\label{ThresholdStrategies-Section4-Equation7}
\end{align}

By Lemma \ref{ThresholdStrategies-Section4-Lemma2} and \ref{ThresholdStrategies-Section4-Lemma3}, the function $V_{y}(x)$ is absolutely continuous for $x \in [y,\infty )$ and differentiable in a full measure set in $[y,\infty )$. Define for $w>y$,

\begin{align}
g(w):=\frac{(\lambda +\delta )V_{y}(w)-\lambda \int_{0}^{1}V_{y}(w\cdot z)dG_{Z}(z)}{r\left( w-x^{\ast }\right) }.
\label{ThresholdStrategies-Section4-Equation8}
\end{align}

Since $V_{y}(x)$ is continuous in $x$, the function $g$ is continuous on $(y,\infty )$ and, if $y>x^{\ast}$, it is continuous at $w=y$ as well. Also, if $y>x^{\ast}$, absolute continuity of $V_{y}$ implies that we can write

\begin{align}
V_{y}(x)=\int_{y}^{x}g(w)dw+V_{y}(y),
\label{ThresholdStrategies-Section4-Equation9}
\end{align}

which shows that $V_{y}(x)$ is continuously differentiable and is a classical solution of the IDE on $[y,\infty)$.

If $y=x^{\ast }$, then $V_{x^{\ast }}(x)=C(x)$, the cost of social protection defined in Subsection \ref{TheStochasticControlProblem-Section2-Subsection22}. Although $g$ is not defined at $w=x^{\ast }$, the same argument applies on any interval $[x^{\ast}+\varepsilon ,x]$ for $x>x^{\ast }$ and any $0<\varepsilon <x-x^{\ast }$. Nevertheless, one can proceed as in the case $V_{y}(x)$ for $y>x^{\ast }$ yielding

\begin{align}
V_{x^{\ast }}(x)=C(x)=\int_{x^{\ast }+\varepsilon }^{x}g(w)dw+C(x^{\ast}+\varepsilon).
\label{ThresholdStrategies-Section4-Equation10}
\end{align}

Letting $\varepsilon \searrow 0$, this argument shows that $C(x)$ is continuously differentiable and a classical solution of the IDE in $(x^{\ast},\infty )$. Moreover, from \eqref{TheStochasticControlProblem-Section2-Equation9} we have

\begin{align}
-(\delta+\lambda)C\left( x^{\ast}\right) +\lambda\int_{0}^{1}\left( \left(x^{\ast}-x^{\ast}\cdot z\right) +C\left( x^{\ast}\right) \right) dG_{Z}(z)=0.
\label{ThresholdStrategies-Section4-Equation11}
\end{align}

We now prove the uniqueness result. Let $u$ be any classical solution of the IDE on $\left( y,\infty \right) $satisfying the value-matching condition $u(y)=V_{y}(y)$ and $\lim_{x\rightarrow \infty }u(x)=0$. Let us extend $u$ to $\left[ 0,y\right] $ as $u(z)=V_{y}(z)=z-y+V_{y}(y)$ for $0\leq z\leq y$. We are going to prove that $u=V_{y}$ on $\left(y,\infty \right) $ and hence that $V_{y}$ is the unique solution associated with these value-matching conditions. Fix $x>y$, consider the threshold strategy $\pi :=\pi _{x}^{y}$ and let

\begin{align}
\tau ^{y}=\inf \left\{ t\geq 0:X_{t}\leq y\right\}.
\label{ThresholdStrategies-Section4-Equation12}
\end{align}

Take $T>0$, using Lemma \ref{AnalysisofV:PropertiesandtheHamilton-Jacobi-BellmanEquation-Section3-Lemma3}, we obtain

\begin{align}
\mathbb{E}\left[e^{-\delta \left( T\wedge \tau ^{y}\right) }u(X_{T\wedge \tau^{y}}^{\pi })\right]-u(x)=\mathbb{E}\left[\int\nolimits_{0}^{T\wedge \tau ^{y}}\mathcal{L}^{y}(u)(X_{t^{-}})e^{-\delta t}dt-\mathbbm{1}_{\left\{ T\wedge \tau^{y}=\tau ^{y}\right\} }(y-X_{\tau ^{y}})e^{-\delta \left( T\wedge \tau^{y}\right)}\right].
\label{ThresholdStrategies-Section4-Equation13}
\end{align}

Since $\mathcal{L}^{y}(u)(x)=0$ for all $x>y$,

\begin{align}
u(x)=\mathbb{E}\left[e^{-\delta \left( T\wedge \tau ^{y}\right) }\left(u(X_{T\wedge \tau ^{y}}^{\pi })+\mathbbm{1}_{\left\{ T\wedge \tau^{y}=\tau ^{y}\right\} }(y-X_{\tau ^{y}})\right)\right].
\label{ThresholdStrategies-Section4-Equation14}
\end{align}

Using the value-matching condition $u(y)=V_{y}(y)$ and that $u(z)=V_{y}(z)$ for $z\leq y$,

\begin{align}
u(x)=\mathbb{E}\left[\mathbbm{1}_{\left\{ T\wedge \tau ^{y}=\tau^{y}\right\} }e^{-\delta \tau ^{y}}V_{y}(X_{\tau ^{y}})\right]+\mathbb{E}\left[\mathbbm{1}_{\left\{ T< \tau ^{y}\right\} }e^{-\delta T}u(X_{T})\right].
\label{ThresholdStrategies-Section4-Equation15}
\end{align}

Letting $T\rightarrow \infty $, and using bounded convergence (since $u$ is bounded and vanishes at infinity),

\begin{align}
\lim_{T\rightarrow \infty }\mathbb{E}\left[\mathbbm{1}_{\left\{ T<\tau ^{y}\right\} }e^{-\delta T}u(X_{T})\right] =0,
\label{ThresholdStrategies-Section4-Equation16}
\end{align}

and also

\begin{align}
\lim_{T\rightarrow \infty }\mathbb{E}\left[\mathbbm{1}_{\left\{ T\wedge \tau^{y}=\tau ^{y}\right\} }e^{-\delta \tau ^{y}}V_{y}(X_{\tau^{y}})\right]=\mathbb{E}\left[e^{-\delta \tau ^{y}}V_{y}(X_{\tau ^{y}})\right].
\label{ThresholdStrategies-Section4-Equation17}
\end{align}

Moreover, since no transfers are made under the threshold strategy prior to $\tau ^{y}$, we obtain

\begin{align}
\mathbb{E}\left[e^{-\delta \tau ^{y}}V_{y}(X_{\tau ^{y}})\right]=V_{y}(x).
\label{ThresholdStrategies-Section4-Equation18}
\end{align}

Therefore, $u(x)=V_{y}(x)$.
\end{proof}

\vspace{0.3cm}

Alternatively, one can also characterise $V_{y}$ in $[y,\infty)$ as a unique fixed-point of an operator, as established by the Proposition \ref{ThresholdStrategies-Section4-Proposition3}.

\vspace{0.3cm}

\begin{proposition}  \label{ThresholdStrategies-Section4-Proposition3}

The operator $\mathcal{T}$: $\mathcal{W}\rightarrow\mathcal{W}$ defined as

\begin{align}
\begin{array}
[c]{ccc}
T(W)(x) & = & \mathbb{E}\left[(y-X_{\tau_{1}}+W(y))\mathbbm{1}_{\left\{  X_{\tau_{1}}<y\right\}}e^{-\delta\tau_{1}}\right]  +\mathbb{E}\left[W\left(X_{\tau_{1}}\right) \mathbbm{1}_{\left\{  X_{\tau_{1}}\geq y\right\}  } e^{-\delta\tau_{1}}\right],
\end{array}
\label{ThresholdStrategies-Section4-Equation19}
\end{align}

where

\begin{align}
\mathcal{W=}\left\{  W:[y,\infty)\rightarrow\lbrack0,\infty)\text{ bounded and non-negative functions}\right\},
\label{ThresholdStrategies-Section4-Equation20}
\end{align}

with the norm $\left\Vert W\right\Vert =\sup_{[y,\infty)}W(x)$ is a contraction and has the function $V_{y}$ as the unique fixed-point.

\end{proposition}

\begin{proof}

We have the following,

\begin{align}
T(W_{1})(x)-T(W_{2})(x)  &  =\left(  W_{1}(y)-W_{2}(y)\right)  \mathbb{E}\left[e^{-\delta\tau_{1}}\mathbbm{1}_{\{X_{\tau_{1}}<y\}}\right] \\ \\
&  +\mathbb{E}\left[\left(W_{1}(X_{\tau_{1}})-W_{2}(X_{\tau_{1}})\right)\mathbbm{1}_{\left\{  X_{\tau_{1}}\geq y\right\}}e^{-\delta\tau_{1}}\right] \\ \\
&  \leq\left\Vert W_{1}-W_{2}\right\Vert \mathbb{E}\left[e^{-\delta\tau_{1}}\right]  =\left\Vert W_{1}-W_{2}\right\Vert \left(\frac{\lambda}{\lambda+\delta}\right).
\label{ThresholdStrategies-Section4-Equation21}
\end{align}

Thus, it is a contraction and has a unique fixed-point. Since $T\left(V_{y}\right)=V_{y}$ we have the result.
\end{proof}

In the following remark, we describe the procedure used to identify the optimal value function in the examples presented in the subsequent sections. We begin by determining the optimal threshold within the class of threshold strategies. The corresponding value function is then taken as a candidate for optimality, whose optimality is subsequently assessed through a verification argument based on the Hamilton--Jacobi--Bellman (HJB) equation. While this procedure confirms optimality in all the examples considered in this paper, we do not establish in general that the optimal strategy must have a threshold structure.

\begin{remark} \label{ThresholdStrategies-Section4-Remark1}

For any threshold $y \geq x^\ast$, consider $V_y(x)$ as the value function associated with the threshold strategy with level $y$. We first identify a candidate optimal threshold strategy by minimising the cost within the class of threshold strategies. If several minimisers exist, we choose the smallest one. Hence,


\begin{align}
	V_{y^\star}(x)=\inf_{y\geq x^\ast} V_y(x), \qquad x\geq 0.
\end{align}



This minimisation provides a candidate value function for the original control problem. Its optimality among all admissible strategies is then established by verifying the conditions of Theorem \ref{AnalysisofV:PropertiesandtheHamilton-Jacobi-BellmanEquation-Section3-Theorem1}. More precisely, from Proposition \ref{ThresholdStrategies-Section4-Proposition2}, this reduces to verifying that

\begin{align}
	\mathcal{L}(V_{y^{\star }})(x)=-r(x-x^{\ast })-\delta (V_{y^{\star}}(y^{\star })+y^{\star }-x)+\lambda x(1-\mu )\geq 0,
\end{align}

for all $x\in \lbrack x^{\ast },y^{\star })$ and $V_{y^{\star }}^{\prime}(x)\geq -1$ for all $x\in \lbrack y^{\star },\infty )$.

\end{remark}

\section{Closed-Form Solutions for Threshold Strategies in the Special Case $Z_{i}\sim Beta(\alpha,1)$} \label{Closed-FormSolutionsinaSpecialCase-Section5}

In this section, we derive explicit expressions for the value function associated with the threshold strategies introduced in Section \ref{ThresholdStrategies-Section4}, under the assumption that the remaining proportion of capital follows a particular case of the Beta distribution; that is, when $Z_{i}\sim Beta(\alpha,1)$, case for which the distribution function is $G_{Z}(z)=z^{\alpha}$ and the p.d.f. is $g_{Z}(z) = \alpha z^{\alpha -1}$ for $0<z<1$, where $\alpha >0$. More precisely, in this case we were able to derive an explicit solution to the IDE \eqref{ThresholdStrategies-Section4-Equation2} in accordance with the criteria established in Proposition \ref{ThresholdStrategies-Section4-Proposition2}. Moreover, using this explicit solution we can obtain the value-matching condition $V_{y}(y)$ for threshold levels $y > x^{\ast}$ (recall that, when the threshold is $ y = x^{\ast}$, $V_{x^{\ast}}(x^{\ast}) = C(x^{\ast})$ is known regardless of the distribution of $Z_{i}$).

Although these solutions are derived for a specific class of admissible strategies, as noted earlier, they serve as natural candidates for the global solution of the HJB Equation \eqref{AnalysisofV:PropertiesandtheHamilton-Jacobi-BellmanEquation-Section3-Equation4}. Furthermore, following the procedure described in Remark \ref{ThresholdStrategies-Section4-Remark1}, after optimising the threshold parameter $y$ (i.e., after calculating $y^\star$), we verify whether the corresponding candidate value function satisfies the conditions of the verification theorem (Theorem \ref{AnalysisofV:PropertiesandtheHamilton-Jacobi-BellmanEquation-Section3-Theorem1}). In all the examples considered in this section, these conditions are satisfied, and therefore the value function associated with the optimal threshold strategy is indeed the optimal value function among all admissible strategies.


In the following proposition, we derive a closed-form expression for the value function associated with a threshold strategy in the aforementioned case, assuming the remaining proportions of capital are $Beta(\alpha,1)$---distributed. Remark \ref{Closed-FormSolutionsinaSpecialCase-Section5-Remark1} then considers the particular case in which the threshold coincides with the poverty line, i.e., $y = x^{\ast}$, so that $V_{x^{\ast}}(x)=C(x)$. Moreover, the remark shows how $C(x)$ in this special case can also be derived using the Gerber-Shiu function.

\vspace{0.3cm}

\begin{proposition} \label{Closed-FormSolutionsinaSpecialCase-Section5-Proposition1}

Consider a household capital process defined as in \eqref{TheStochasticControlProblem-Section2-Equation4} and \eqref{TheStochasticControlProblem-Section2-Equation5}, with initial capital $x\ge 0$, capital growth rate $r$, intensity $\lambda > 0$ and remaining proportions of capital with distribution $Beta(\alpha, 1)$ where $\alpha >0$; that is, $Z_{i}\sim Beta(\alpha, 1)$. The value function corresponding to a threshold strategy with level $y > x^{\ast}$ is given by

\begin{align}
V_{y}(x)=
\begin{cases}
\left(y-x\right)\ + \mathcal{F}_{0}\left(\frac{x^{*}}{y}\right)\overline{A}(y)
\hspace{1.4cm} \textit{if} & 0\leq x\leq y,\\ \\
\mathcal{F}_{0}\left(\frac{x^{*}}{x}\right)\overline{A}\left(y\right)\left(\frac{y}{x}\right)^{b} \hspace{2.1cm} \textit{if} & \hspace{-0.1cm}  x>y,
\end{cases}
\label{Closed-FormSolutionsinaSpecialCase-Section5-Equation1}
\end{align}


where $\delta \ge 0$ is the force of interest for valuation, $\mathcal{F}_{0}(z)=\ _{2} F_{1}\left(b, b - c + 1; b - a + 1; z\right)$, $\mathcal{F}_{1}(z)= \ _{2} F_{1}\left(b + 1, b - c + 1; b - a + 1; z\right)$, ${ }_{2} F_{1}\left(\cdot \right)$ is Gauss\rq s Hypergeometric Function as defined in \eqref{Closed-FormSolutionsinaSpecialCase-Section5-Equation6}, \scriptsize $a=\left(-(\delta+\lambda-\alpha r)-\sqrt{(\delta+\lambda-\alpha r)^{2}+4r\alpha\delta}\right)/2r$\normalsize, \scriptsize$b=\left(-(\delta+\lambda-\alpha r)+\sqrt{(\delta+\lambda-\alpha r)^{2}+4r\alpha\delta}\right)/2r$\normalsize, \scriptsize $c=\alpha$\normalsize \ and $\overline{A}\left(y\right)  = \frac{\lambda y}{\left(\alpha + 1\right)\left[\delta\mathcal{F}_{0}\left(\frac{x^{*}}{y}\right) + \frac{rb\left(y-x^{*}\right)}{y}\mathcal{F}_{1}\left(\frac{x^{*}}{y}\right)\right]}$.


\end{proposition}

\begin{proof}

For $x > y$, under the assumption $Z_{i}\sim Beta(\alpha, 1)$, the IDE \eqref{ThresholdStrategies-Section4-Equation2} can be written such that

\vspace{0.3cm}

\begin{align}
    \begin{split}
        0 = r(x-x^{*})W'(x)
        -(\delta + \lambda)W(x)
        +\lambda \left(\frac{y}{x}\right)^{\alpha} \left[\frac{y + (\alpha + 1) W(y)}{\alpha + 1}\right]+\lambda \int_{y/x}^{1}W(x\cdot z)\alpha z^{\alpha-1}dz.\\ \\
        \label{Closed-FormSolutionsinaSpecialCase-Section5-Equation2}
    \end{split}
\end{align}

\vspace{-0.3cm}

Differentiating both sides of \eqref{Closed-FormSolutionsinaSpecialCase-Section5-Equation2} with respect to $x$, and performing some algebraic manipulations, yields the following second-order ordinary differential equation (ODE):

\vspace{0.3cm}

\begin{align}
    \begin{split}
        0 = r(x^{2}-xx^{*})W''(x)+\left[(r(1+\alpha)-\delta-\lambda)x-r\alpha x^{*}\right]W'(x)-\alpha \delta W(x).
        \label{Closed-FormSolutionsinaSpecialCase-Section5-Equation3}
    \end{split}
\end{align}

\vspace{0.3cm}
Let us consider the change of variable $t:=\frac{x}{x^{*}}$ and define $f(t):=W(t x^{\ast})$, Equation \eqref{Closed-FormSolutionsinaSpecialCase-Section5-Equation3} reduces to Gauss\rq s hypergeometric differential equation \citep{Book:Slater1960}

\vspace{0.3cm}

\begin{align}
    t(1-t)\cdot f''(t) + [c - (1+a+b)t] f'(t) - ab f(t) =0,
    \label{Closed-FormSolutionsinaSpecialCase-Section5-Equation4}
\end{align}

\vspace{0.3cm}

for \footnotesize$a=\left(-(\delta+\lambda-\alpha r)-\sqrt{(\delta+\lambda-\alpha r)^{2}+4r\alpha\delta}\right)/2r$\normalsize, \footnotesize$b=\left(-(\delta+\lambda-\alpha r)+\sqrt{(\delta+\lambda-\alpha r)^{2}+4r\alpha\delta}\right)/2r$\normalsize \ and \footnotesize$c= \alpha$\normalsize, with regular singular points at $t=0, 1, \infty$ (corresponding to $x=0,x^{*},\infty$, respectively). A general solution of \eqref{Closed-FormSolutionsinaSpecialCase-Section5-Equation4} in the neighborhood of the singular point $t=\infty$ is given by

\vspace{0.3cm}

\begin{align}
    f(t):=W(x)= &A_{1}\left(\frac{x^{\ast}}{x}\right)^{a} { }_{2} F_{1}\left(a, a-c+1 ; a-b+1 ; \frac{x^{\ast}}{x}\right)\\&+A_{2}\left(\frac{x^{\ast}}{x}\right)^{b} { }_{2} F_{1}\left(b, b-c+1 ; b-a+1 ; \frac{x^{\ast}}{x}\right),
    \label{Closed-FormSolutionsinaSpecialCase-Section5-Equation5}
\end{align}
\normalsize

\vspace{0.3cm}

for arbitrary constants $A_{1},A_{2} \in \mathbb {R}$ (see for example, Equations (15.5.7) and (15.5.8) of \cite{Book:Abramowitz1964}). Here,

\vspace{0.3cm}

\begin{align}
    { }_{2} F_{1}(a, b ; c ; z)=\sum_{n=0}^{\infty} \frac{(a)_{n}(b)_{n}}{(c)_{n}} \frac{z^{n}}{n !}
    \label{Closed-FormSolutionsinaSpecialCase-Section5-Equation6}
\end{align}

is Gauss\rq s hypergeometric function and $(a)_{n}=\frac{\Gamma(a+n)}{\Gamma(n)}$ denotes the Pochhammer symbol \citep{Book:Seaborn1991}.

The constants $A_1$ and $A_2$ are obtained from the boundary conditions at $y$ and at infinity. The condition $\lim_{x\to\infty} W(x) = 0$ implies $A_{1}=0$. Evaluating \eqref{Closed-FormSolutionsinaSpecialCase-Section5-Equation2} at $x=y$ yields

\vspace{0.3cm}

\begin{align}
    W'\left(y^{+}\right) = \frac{1}{r\left(y - x^{*}\right)}\left(\delta W(y) - \frac{\lambda y}{\alpha + 1}\right).
    \label{Closed-FormSolutionsinaSpecialCase-Section5-Equation7}
\end{align}

\vspace{0.3cm}

Then, using the differential properties of Gauss’s hypergeometric function, namely $\tfrac{d}{dz} { }_{2} F_{1}(a, b ; c ; z) = \tfrac{ab}{c} { }_{2} F_{1}(a + 1, b + 1; c + 1; z)$, we obtain the following,

\begin{align}
&A_{2}\left(-\frac{b}{y}\left(\frac{x^{\ast}}{y}\right)^{b}{ }_{2} F_{1}\left(b, b - c + 1; b - a + 1; \frac{x^{\ast}}{y}\right) \right. \\& \left. - \left(\frac{x^{\ast}}{y}\right)^{b} \frac{bx^{\ast}\left(b - c + 1\right)}{\left(b - a + 1\right) y^{2}} { }_{2} F_{1}\left(b + 1, b - c + 2; b - a + 2; \frac{x^{\ast}}{y}\right) \right) \\ & = \frac{1}{r\left(y - x^{\ast}\right)}\left(A_{2} \delta \left(\frac{x^{\ast}}{y}\right)^{b} { }_{2} F_{1}\left(b, b - c + 1 ; b - a + 1 ; \frac{x^{\ast}}{y}\right) - \frac{\lambda y}{\alpha + 1}\right).
\label{Closed-FormSolutionsinaSpecialCase-Section5-Equation8}
\end{align}

Hence, solving for $A_{2}$, it yields,

\begin{align}
A\left(y\right) := A_{2} & = \frac{\lambda y \left(\frac{y}{x^{*}}\right)^{b}}{\left(\alpha + 1\right)\left[\delta\mathcal{F}_{0}\left(\frac{x^{*}}{y}\right) + \frac{rb\left(y-x^{*}\right)}{y}\mathcal{F}_{1}\left(\frac{x^{*}}{y}\right)\right]},
\label{Closed-FormSolutionsinaSpecialCase-Section5-Equation9}
\end{align}


where $\mathcal{F}_{0}(z):=\ _{2} F_{1}\left(b, b - c + 1; b - a + 1; z\right)$ and $\mathcal{F}_{1}(z):= \ _{2} F_{1}\left(b + 1, b - c + 1; b - a + 1; z\right)$. Then, defining $\overline{A}\left(y\right):=A\left(y\right)\left(\frac{x^{*}}{y}\right)^{b}$ yields the value function given by \eqref{Closed-FormSolutionsinaSpecialCase-Section5-Equation1}.
\end{proof}

\begin{remark} \label{Closed-FormSolutionsinaSpecialCase-Section5-Remark1}

We now consider the particular case in which the threshold level coincides with the poverty line, i.e., $y = x^{\ast}$. In this case, the value function is: $V_{x^{\ast}}(x) = C(x)$. That is, the value function is equal to the cost of social protection, introduced in Subsection \ref{TheStochasticControlProblem-Section2-Subsection22}. Hence, setting $y = x^{\ast}$ in \eqref{Closed-FormSolutionsinaSpecialCase-Section5-Equation1} yields,

\vspace{0.3cm}

\begin{align}
V_{x^{\ast}}(x) = C\left(x\right)  =\left\{
\begin{array}
[c]{ccc}%
\left(  x^{\ast}-x\right)  +\frac{\lambda x^{\ast}}{\left(  \alpha+1\right)\delta} & \hspace{0.7cm} \textit{if} & 0\leq x\leq x^{\ast},\\ \\
\frac{\lambda x^{\ast}\mathcal{F}_{0}\left(\frac{x^{*}}{x}\right)}{\left(  \alpha+1\right)\delta \mathcal{F}_{0}(1) }{\left(  \frac{x^{\ast}}{x}\right)  ^{b}} & \hspace{0.5cm} \textit{if} & x>x^{\ast},
\end{array}
\right.
\label{Closed-FormSolutionsinaSpecialCase-Section5-Equation10}
\end{align}

\vspace{0.3cm}

where $\mathcal{F}_{0}(z)=\ _{2} F_{1}\left(b, b - c + 1; b - a + 1; z\right)$, \footnotesize $ a=\left(-(\delta+\lambda-\alpha r)-\sqrt{(\delta+\lambda-\alpha r)^{2}+4r\alpha\delta}\right)/2r$, $b=\left(-(\delta+\lambda-\alpha r)+\sqrt
{(\delta+\lambda-\alpha r)^{2}+4r\alpha\delta}\right)/2r$ \normalsize and \footnotesize $c=\alpha$\normalsize. 

Note that one can also obtain Equation \eqref{Closed-FormSolutionsinaSpecialCase-Section5-Equation10} for $x>x^{*}$ by means of the Gerber-Shiu expected discounted penalty function, recently derived by \cite{Article:Flores-Contro2025}. Indeed, for $x>x^{*}$, the cost of social protection $C(x)$ is given by

\begin{align}
C\left(x\right)= \mathbbm{E}\left[\mid X_{\overline{\tau}}-x^{*}\mid e^{-\delta \overline{\tau}};\overline{\tau}<\infty\right] + C(x^{*}) \cdot \mathbbm{E}\left[e^{-\delta\overline{\tau}};\overline{\tau}<\infty\right],
\label{Closed-FormSolutionsinaSpecialCase-Section5-Equation11}
\end{align}

where $\mathbbm{E}\left[e^{-\delta \overline{\tau}};\overline{\tau}<\infty\right]$ is equivalent to $\mathbbm{E}\left[e^{-\delta \overline{\tau}} \mathbbm{1}_{\{\overline{\tau}<\infty\}} \right]$. Under the assumption of $Beta(\alpha,1)$---distributed remaining proportions of capital; that is, $Z_{i}\sim Beta(\alpha,1)$, Equation \eqref{Closed-FormSolutionsinaSpecialCase-Section5-Equation11} yields to \eqref{Closed-FormSolutionsinaSpecialCase-Section5-Equation10}, as $\mathbbm{E}\left[\mid X_{\overline{\tau}}-x^{*}\mid e^{-\delta \overline{\tau}};\overline{\tau}<\infty\right]=x^{*}/\left(\alpha + 1\right) \cdot \mathbbm{E}\left[e^{-\delta \overline{\tau}};\overline{\tau}<\infty\right]$ (see, for example, Section 5 of \cite{Article:Flores-Contro2025}) and $\mathbbm{E}\left[e^{-\delta \overline{\tau}};\overline{\tau}<\infty\right]$ is the Laplace transform of the trapping time, given by Equation (10) from \cite{Article:Flores-Contro2025}.

\end{remark}

\begin{figure}[H]
  	\centering
        \begin{subfigure}[b]{0.32\linewidth}
  	\includegraphics[width=4.5cm, height=4.5cm]{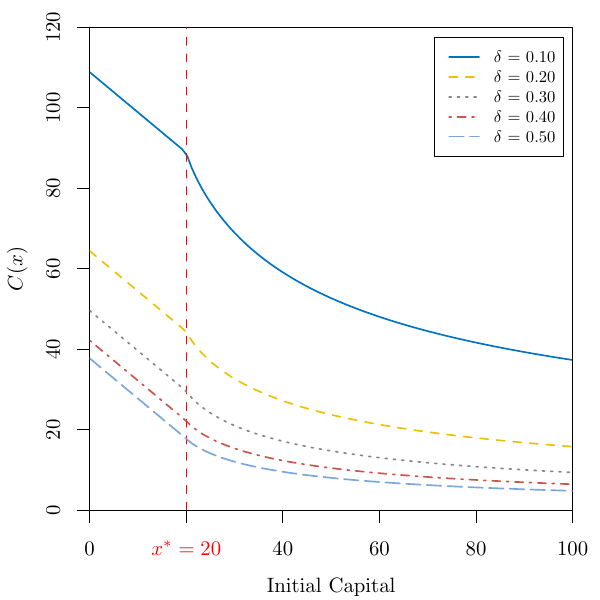}
	\caption{}
  	\label{Closed-FormSolutionsinaSpecialCase-Section5-Figure1-a}
	\end{subfigure}
 	 \hspace{0.01cm}
         \begin{subfigure}[b]{0.32\linewidth}
 	 \includegraphics[width=4.5cm, height=4.5cm]{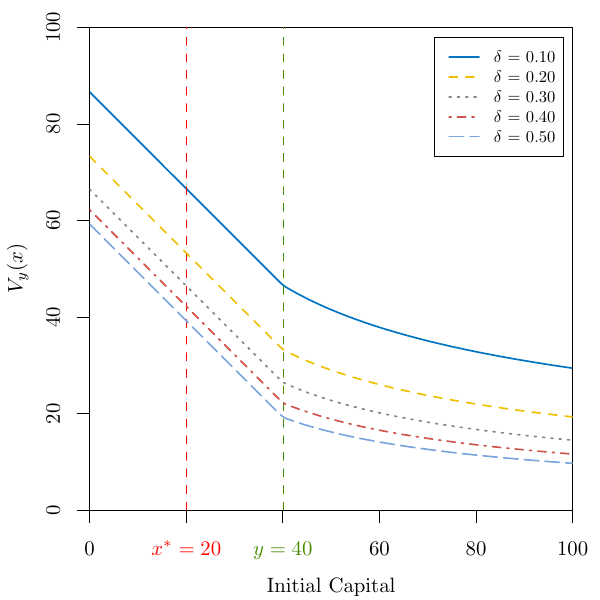}
	 \caption{}
  	\label{Closed-FormSolutionsinaSpecialCase-Section5-Figure1-b}
	\end{subfigure}
  	 \hspace{0.01cm}
          \begin{subfigure}[b]{0.32\linewidth}
 	 \includegraphics[width=4.5cm, height=4.5cm]{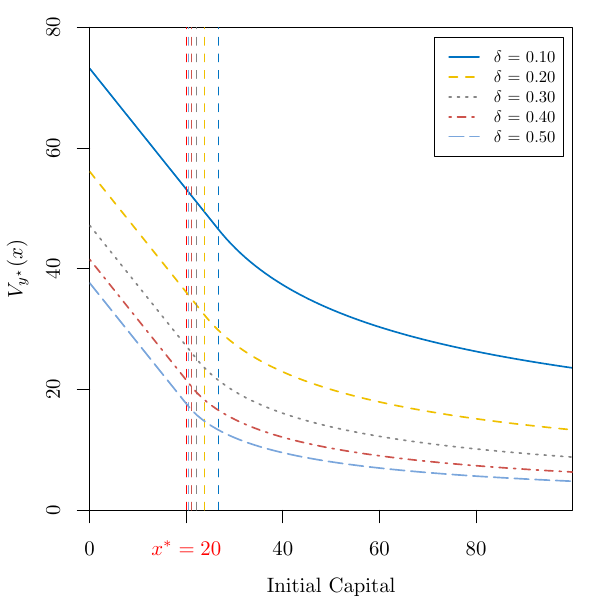}
	 \caption{}
  	\label{Closed-FormSolutionsinaSpecialCase-Section5-Figure1-c}
	\end{subfigure}
	 \\ \vspace{0.1cm}
         \begin{subfigure}[b]{0.32\linewidth}
 	 \includegraphics[width=4.5cm, height=4.5cm]{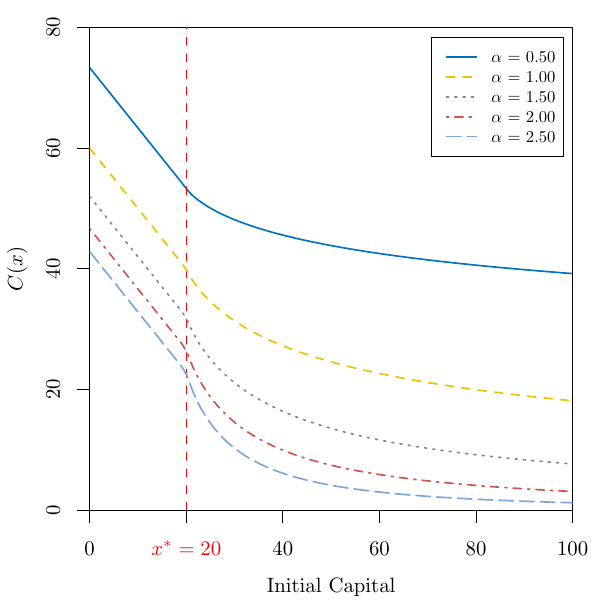}
	 \caption{}
  	\label{Closed-FormSolutionsinaSpecialCase-Section5-Figure1-d}
	\end{subfigure}
 	 \hspace{0.01cm}
         \begin{subfigure}[b]{0.32\linewidth}
 	 \includegraphics[width=4.5cm, height=4.5cm]{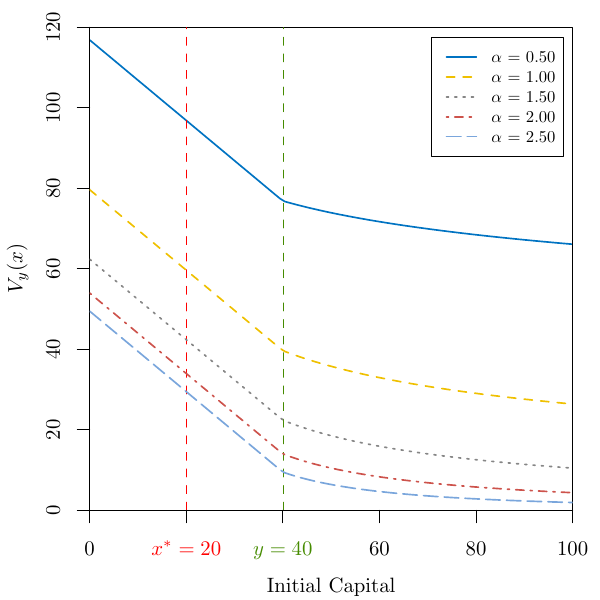}
	 \caption{}
  	\label{Closed-FormSolutionsinaSpecialCase-Section5-Figure1-e}
	\end{subfigure}
  	\hspace{0.01cm}
         \begin{subfigure}[b]{0.32\linewidth}
 	 \includegraphics[width=4.5cm, height=4.5cm]{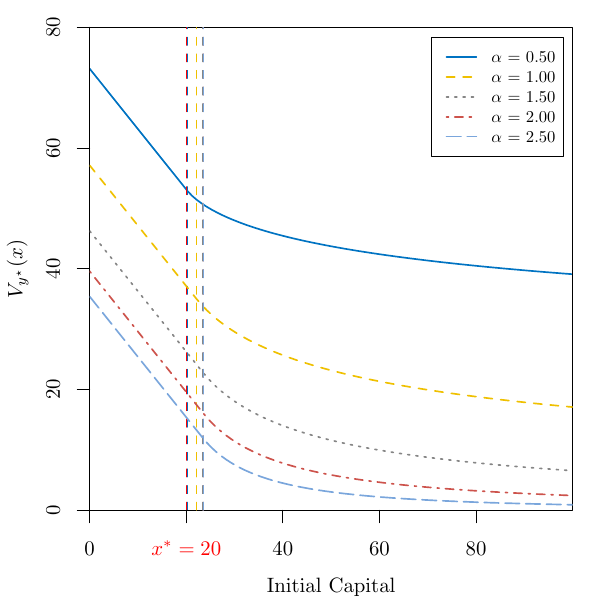}
	 \caption{}
  	\label{Closed-FormSolutionsinaSpecialCase-Section5-Figure1-f}
	\end{subfigure}
	 \\ \vspace{0.1cm}
\caption{The first row (a)--(c) shows the cost of social protection and the value function of a threshold strategy when $Z_{i} \sim Beta(1.25, 1)$, $a = 0.10$, $b = 3$, $c = 0.40$, $\lambda = 1$, $x^{*} = 20$ for $\delta = 0.10, 0.20, 0.30, 0.40, 0.50$. The second row (d)--(f) displays the cost of social protection and the value function of a threshold strategy when $Z_{i} \sim Beta(\alpha, 1)$, $a = 0.10$, $b = 3$, $c = 0.40$, $\lambda = 1$, $\delta = 0.25$, $x^{*} = 20$, $y=40$ for $\alpha = 0.50, 1.00, 1.50, 2.00, 2.50$. Panels (a) and (d) display $C(x)$ corresponding to $V_{y}(x)$ with threshold $y=20$, while (b) and (e) show $V_{y}(x)$ with threshold $y=40$. Similarly, panels (c) and (f) show $V_{y^{\star}}(x)$ with the optimal threshold $y^{\star}$.}
\label{Closed-FormSolutionsinaSpecialCase-Section5-Figure1}
\end{figure}

In the context of Remark \ref{ThresholdStrategies-Section4-Remark1}, we therefore restrict attention to the class of threshold strategies and seek the optimal threshold within this class. Letting $h(y):=V_{y}(x^{\ast})$, the candidate optimal threshold $y^{\star}$ is obtained as a solution to the first-order condition $h^{\prime}(y)= 0$, using that in this case we have the formula for this derivative from the closed form solution \eqref{Closed-FormSolutionsinaSpecialCase-Section5-Equation1}. Since this condition cannot be solved analytically, the solution is computed numerically for $y^{\star}$. This procedure yields a candidate value function $V_{y^{\star}}(x)$ for optimality. Upon verifying that the conditions stated in Remark \ref{ThresholdStrategies-Section4-Remark1} are satisfied, we find that, in all examples considered in this section, the value function associated to threshold $y^{\star}$ fulfills these conditions. Therefore, we conclude that $V_{y^{\star}}(x)$ is optimal both within the class of threshold strategies and over the set of all admissible strategies.

\begin{figure}[H]
	\begin{subfigure}[b]{\linewidth}
	\centering
  		\includegraphics[width=7.5cm, height=7.5cm]{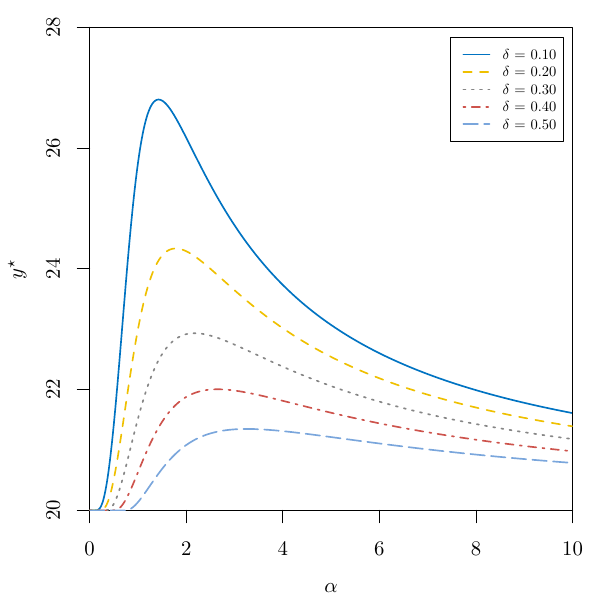}
  		\label{Closed-FormSolutionsinaSpecialCase-Section5-Figure2}
	\end{subfigure}
	\caption{The optimal threshold $y^{\star}$ when $Z_{i} \sim Beta(\alpha, 1)$, $a = 0.10$, $b = 3$, $c = 0.40$, $\lambda = 1$, $x^{*} = 20$ for different values of the discount rate $\delta$.}
	\label{Closed-FormSolutionsinaSpecialCase-Section5-Figure2}
\end{figure}

Figure \ref{Closed-FormSolutionsinaSpecialCase-Section5-Figure1} presents some examples of value functions corresponding to the class of threshold strategies with different threshold levels, under the assumption that $Z_{i}\sim Beta(\alpha,1)$. Specifically, we consider: (i) Figures \ref{Closed-FormSolutionsinaSpecialCase-Section5-Figure1-a} and \ref{Closed-FormSolutionsinaSpecialCase-Section5-Figure1-d}, where the threshold coincides with the poverty line, $y = x^{\ast}$, for which the value function is given by \eqref{Closed-FormSolutionsinaSpecialCase-Section5-Equation10}; (ii) Figures \ref{Closed-FormSolutionsinaSpecialCase-Section5-Figure1-b} and \ref{Closed-FormSolutionsinaSpecialCase-Section5-Figure1-e}, where the threshold exceeds the poverty line, $y > x^{\ast}$, for which the value function is given by \eqref{Closed-FormSolutionsinaSpecialCase-Section5-Equation1}; and (iii) Figures \ref{Closed-FormSolutionsinaSpecialCase-Section5-Figure1-c} and \ref{Closed-FormSolutionsinaSpecialCase-Section5-Figure1-f} where the threshold is set to the optimal level, $y = y^{\star}$, with the corresponding value function again given by \eqref{Closed-FormSolutionsinaSpecialCase-Section5-Equation1} (with the optimal threshold in these examples verified to satisfy the conditions stated in Remark \ref{ThresholdStrategies-Section4-Remark1}). The functions are clearly decreasing, in agreement with the earlier results. Furthermore, as expected, the functions decrease with respect to both the discount rate $\delta$ and the shape parameter $\alpha$. This reflects the fact that a lower discount rate assigns greater weight to future transfers, thereby leading to higher costs. Similarly, a larger value of $\alpha$ corresponds to a higher expected remaining proportion of capital (higher $\mu$), which in turn reduces the total amount of discounted transfers. A comparison of Figures \ref{Closed-FormSolutionsinaSpecialCase-Section5-Figure1-a} and \ref{Closed-FormSolutionsinaSpecialCase-Section5-Figure1-b} with \ref{Closed-FormSolutionsinaSpecialCase-Section5-Figure1-c}, as well as Figures \ref{Closed-FormSolutionsinaSpecialCase-Section5-Figure1-d} and \ref{Closed-FormSolutionsinaSpecialCase-Section5-Figure1-e} with \ref{Closed-FormSolutionsinaSpecialCase-Section5-Figure1-f}, highlights the potential cost savings for the government when adopting the strategy corresponding to the optimal threshold $y^{\star}$. Note that the optimal threshold in Figure \ref{Closed-FormSolutionsinaSpecialCase-Section5-Figure1-c} is $y^{\star} = 26.66, 23.82, 22.16, 21.10, 20.41$, corresponding to the case $\delta = 0.10, 0.20, 0.30, 0.40, 0.50$, respectively (for completeness, the optimal thresholds shown in Figure \ref{Closed-FormSolutionsinaSpecialCase-Section5-Figure1-f} for $\alpha = 0.50, 1.00, 1.50, 2.00, 2.50$ are $y^{\star} = 20.27, 22.16, 23.32, 23.56, 23.42$, respectively). As noted previously, the conditions described in Remark \ref{ThresholdStrategies-Section4-Remark1} were verified to hold, and hence we conclude that the resulting value function $V_{y^{\star}}(x)$ is the optimal among all admissible strategies.

Figure \ref{Closed-FormSolutionsinaSpecialCase-Section5-Figure2} shows the sensitivity of the optimal threshold $y^{\star}$ with respect to the discount factor $\delta$ and the shape parameter $\alpha$. This figure shows that the optimal threshold $y^{\star}$ is not monotone in $\alpha$. This can be interpreted as follows: when the expected remaining proportion of capital is low (small $\alpha$), the optimal threshold $y^{\star}$ becomes less relevant, since the household will almost surely fall into poverty after the first capital loss. However, as the expected remaining proportion of capital increases (with higher $\alpha$), the choice of an optimal threshold gains importance: the household is less likely to fall into poverty after the first loss and thus has the possibility of rebuilding its capital. It is important to note, however, that another effect arises when $\alpha$ becomes sufficiently large: capital losses have a smaller impact on the household, which in turn may reduce the importance of selecting an optimal threshold, as the losses themselves become less significant.


\section{General Case Analysis: Absence of Closed-Form Solutions} \label{GeneralCaseAnalysis:AbsenceofClosed-FormSolutions-Section6}

In general, it is not straightforward to derive explicit formulas for the value function of a threshold strategy when more general cases are considered (e.g., when the distribution of the remaining proportion of capital differs from the $Beta(\alpha,1)$ specification, or when a microinsurance cover is taken into account, as discussed in Section \ref{Microinsurance-Section7}). Monte Carlo simulation is an alternative way to produce estimates and is particularly useful when dealing with cases for which closed-form formulas are not available. In this section, we introduce a simple and efficient methodology that allows to generate fairly accurate approximations for these value functions. This approach will also allow us to estimate the optimal threshold, which, as previously emphasised, represents a key aspect of the problem.

\subsection{Methodology} \label{GeneralCaseAnalysis:AbsenceofClosed-FormSolutions-Section6-Subsection1}

We begin by computing $V_{y}(y)$. Let us consider trajectories of $X_t^{\pi_y}$, with initial capital $y$, and the sequence $(\tau_{n}, Z_{n})$, up to the time $\tau^y=\min \left\{t: X_t^{\pi_y}<y\right\}$. For each trajectory $\omega_i$, we find $\left(\tau^y\left(\omega_i\right), J_y\left(\omega_i\right)\right)$, with $J_y\left(\omega_i\right)=y\left(\omega_i\right)-X_{\tau^y}^{\pi_y}\left(\omega_i\right)$. Thus, we have the following,

\vspace{0.3cm}

\begin{align}
V_{y}(y) \approx \frac{1}{N} \sum_{i=1}^N\left[J_y\left(\omega_i\right) e^{-\delta \tau^y\left(\omega_i\right)}+ V_{y}(y)e^{-\delta\tau^y\left(\omega_i\right)}\right],
\label{GeneralCaseAnalysis:AbsenceofClosed-FormSolutions-Section6-Equation1}
\end{align}

\vspace{0.3cm}

and therefore, one can approximate $V_{y}(y)$ as follows,

\vspace{0.3cm}

\begin{align}
V_{y}(y) \approx \frac{\frac{1}{N} \sum\limits_{i=1}^NJ_y\left(\omega_i\right) e^{-\delta \tau^y\left(\omega_i\right)}}{1-\frac{1}{N} \sum\limits_{i=1}^N e^{-\delta \tau^y\left(\omega_i\right)}}.
\label{GeneralCaseAnalysis:AbsenceofClosed-FormSolutions-Section6-Equation2}
\end{align}

\vspace{0.3cm}

Then, we compute $V_{y}(y)$ for $x>y$. We focus on the trajectories of $X_t^{\pi_y}$, with initial capital $x>y$, and the sequence $\left(\tau_{n}, Z_{n}\right)$, up to the time $\tau^y=\min \left\{t: X_t^{\pi_y}<y\right\}$. Similarly, for each trajectory $\omega_{i}$, we find $\left(\tau^y\left(\omega_i\right), J_y\left(\omega_i\right)\right)$, with $J_y\left(\omega_i\right)=y\left(\omega_i\right)-X_{\tau^y}^{\pi_y}\left(\omega_i\right)$. Hence, we can approximate $V_{y}(y)$ as follows,

\vspace{0.3cm}

\begin{align}
V_{y}(y) \approx \frac{1}{N} \sum\limits_{i=1}^N\left[J_y\left(\omega_i\right) e^{-\delta \tau^y\left(\omega_i\right)}+ V_{y}(y)e^{-\delta\tau^y\left(\omega_i\right)}\right],
\label{GeneralCaseAnalysis:AbsenceofClosed-FormSolutions-Section6-Equation3}
\end{align}

\vspace{0.3cm}

where the function $V_{y}(y)$ denotes the approximation \eqref{GeneralCaseAnalysis:AbsenceofClosed-FormSolutions-Section6-Equation2}. 

\begin{remark}
	If $\tau^{y}$ is infinite, the value function equals zero, since no capital transfer is required. In the methodology described in Section \ref{GeneralCaseAnalysis:AbsenceofClosed-FormSolutions-Section6-Subsection1}, we truncate time: if $\tau^{y}>T$, for sufficiently large $T$, we treat $\tau^{y}$ as infinite. This approximation is reasonable due to the discounting effect.
\end{remark} 

\begin{example} \label{GeneralCaseAnalysis:AbsenceofClosed-FormSolutions-Section6-Example1}
	Let us consider the situation in which the remaining proportion of capital follows a Kumaraswamy distribution; that is,  $Z_{i} \sim Kumaraswamy(p, q)$, case for which the c.d.f is $G_{Z}(z) = 1-(1-z^{p})^{q}$ and the p.d.f. is $g_{z}(z)=pqz^{p-1}(1-z^{p})^{q-1}$ for $0<z<1$, where $p>0$ and $q>0$. Clearly, under this assumption, deriving an analytical solution to the IDE \eqref{ThresholdStrategies-Section4-Equation2} becomes significantly more difficult; hence, alternative methods as the Monte Carlo simulation methodology presented in this section are required to derive the value function of a threshold strategy.
\end{example}

Figure \ref{GeneralCaseAnalysis:AbsenceofClosed-FormSolutions-Section6-Figure1} displays the cost of social protection and the value function of a threshold strategy for the case described in Example \ref{GeneralCaseAnalysis:AbsenceofClosed-FormSolutions-Section6-Example1}. That is, it presents some examples of value functions corresponding to the class of threshold strategies with different threshold levels, under the assumption that $Z_{i} \sim Kumaraswamy(p, q)$. Just as for Section \ref{Closed-FormSolutionsinaSpecialCase-Section5}, Figures \ref{GeneralCaseAnalysis:AbsenceofClosed-FormSolutions-Section6-Figure1-a} and \ref{GeneralCaseAnalysis:AbsenceofClosed-FormSolutions-Section6-Figure1-d} correspond to the case where the threshold coincides with the poverty line, $y = x^{\ast}$. Figures \ref{GeneralCaseAnalysis:AbsenceofClosed-FormSolutions-Section6-Figure1-b} and \ref{GeneralCaseAnalysis:AbsenceofClosed-FormSolutions-Section6-Figure1-e} illustrate the case where the threshold exceeds the poverty line, $y > x^{\ast}$. Finally, Figures \ref{GeneralCaseAnalysis:AbsenceofClosed-FormSolutions-Section6-Figure1-c} and \ref{GeneralCaseAnalysis:AbsenceofClosed-FormSolutions-Section6-Figure1-f} correspond to the optimal threshold, $y = y^{\star}$. Again, as for Section \ref{Closed-FormSolutionsinaSpecialCase-Section5}, a candidate for optimal threshold in these examples has been verified to satisfy the conditions stated in Remark \ref{ThresholdStrategies-Section4-Remark1}. To determine $y^{\star}$, we again consider the first-order condition $h^{\prime}(y) = 0$. However, unlike in Section \ref{Closed-FormSolutionsinaSpecialCase-Section5}, where this derivative can be computed analytically, here it is evaluated numerically due to the absence of closed-form expressions. The 99\% confidence interval for the functions is included for reference. The figures exhibit the same behavior as those presented in the examples of Section \ref{Closed-FormSolutionsinaSpecialCase-Section5}. It is worth noting, however, that Figures \ref{GeneralCaseAnalysis:AbsenceofClosed-FormSolutions-Section6-Figure1-d}–\ref{GeneralCaseAnalysis:AbsenceofClosed-FormSolutions-Section6-Figure1-f} display higher costs than Figures \ref{Closed-FormSolutionsinaSpecialCase-Section5-Figure1-d}–\ref{Closed-FormSolutionsinaSpecialCase-Section5-Figure1-f}. This result is not unexpected, since the expected value of the remaining proportion of capital modelled with the Kumaraswamy distribution ($\mu^{{\scaleto{\text{ {\fontfamily{qcr}\selectfont KUMARASWAMY}}}{2.5pt}}} = \left(q\cdot \Gamma\left(1+1/p\right)\cdot\Gamma\left(q\right)\right)/\Gamma\left(1+1/p+q\right)$) used in this example is lower than that obtained with the Beta distribution ($\mu^{{\scaleto{\text{ {\fontfamily{qcr}\selectfont BETA}}}{2.5pt}}} = \alpha/(\alpha + 1)$) used in Section \ref{Closed-FormSolutionsinaSpecialCase-Section5}. The optimal thresholds in Figures \ref{GeneralCaseAnalysis:AbsenceofClosed-FormSolutions-Section6-Figure1-c} (for $\delta = 0.10, 0.20, 0.30, 0.40, 0.50$) and \ref{GeneralCaseAnalysis:AbsenceofClosed-FormSolutions-Section6-Figure1-f} (for $p = 0.50, 1.00,$ $1.50, 2.00, 2.50$) are given by $y^{\star} = 29.28, 25.73, 23.79, 21.45, 20$ and $y^{\star} = 20, 20, 20,$$21.18, 23.85$, respectively. As noted previously, we also verify (numerically) that the conditions stated in Remark \ref{ThresholdStrategies-Section4-Remark1} are satisfied for all examples considered in this section. This proves that $V_{y^{\star}}(x)$ is optimal both within the class of threshold strategies and over the set of all admissible strategies. On the other hand, Figure \ref{GeneralCaseAnalysis:AbsenceofClosed-FormSolutions-Section6-Figure2} illustrates the sensitivity of the optimal threshold $y^{\star}$ with respect to the discount factor $\delta$ and the parameter $p$. The pattern of sensitivity observed here is consistent with the behavior reported in Section \ref{Closed-FormSolutionsinaSpecialCase-Section5}.

\begin{figure}[H]
  	\centering
        \begin{subfigure}[b]{0.32\linewidth}
  	\includegraphics[width=4.5cm, height=4.5cm]{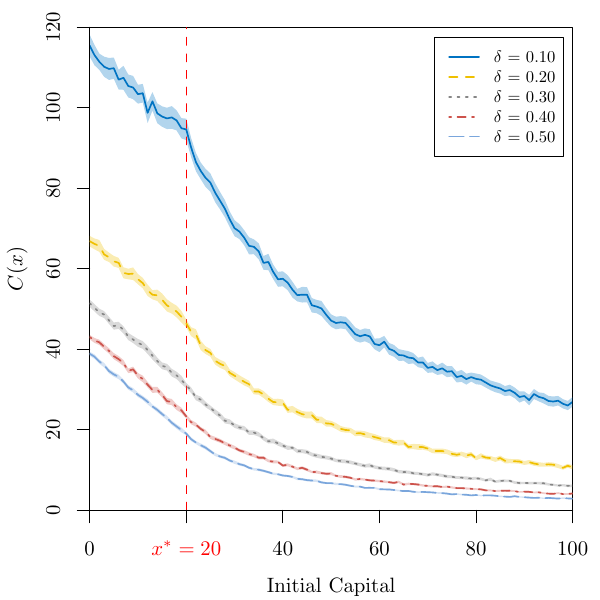}
	\caption{}
  	\label{GeneralCaseAnalysis:AbsenceofClosed-FormSolutions-Section6-Figure1-a}
	\end{subfigure}
 	 \hspace{0.01cm}
         \begin{subfigure}[b]{0.32\linewidth}
 	 \includegraphics[width=4.5cm, height=4.5cm]{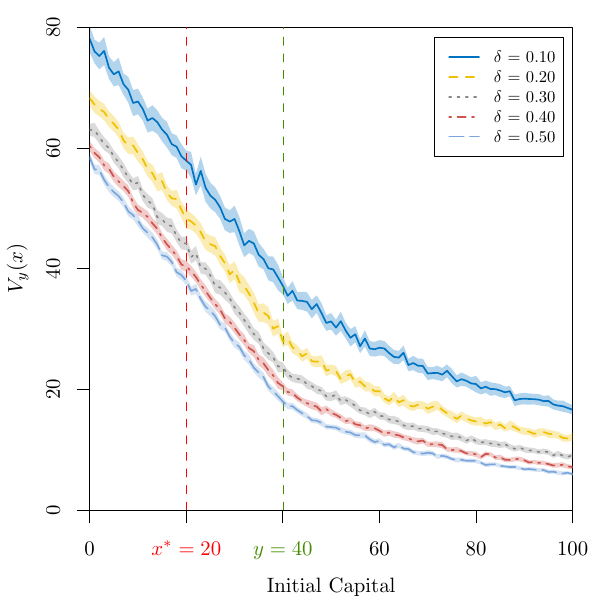}
	 \caption{}
  	\label{GeneralCaseAnalysis:AbsenceofClosed-FormSolutions-Section6-Figure1-b}
	\end{subfigure}
  	 \hspace{0.01cm}
          \begin{subfigure}[b]{0.32\linewidth}
 	 \includegraphics[width=4.5cm, height=4.5cm]{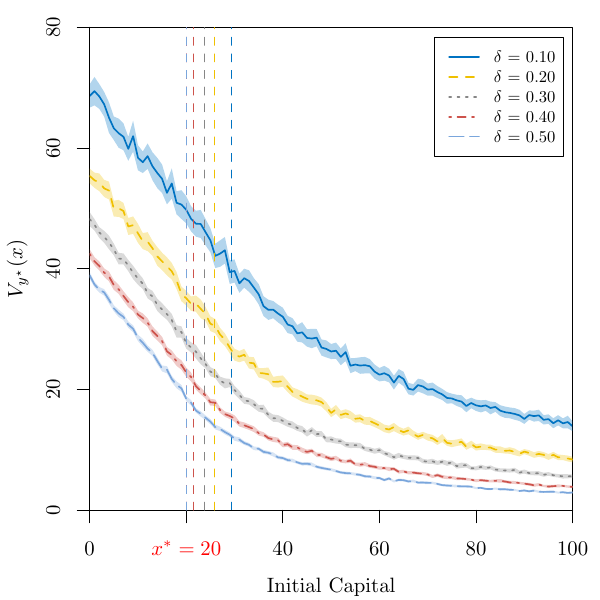}
	 \caption{}
  	\label{GeneralCaseAnalysis:AbsenceofClosed-FormSolutions-Section6-Figure1-c}
	\end{subfigure}
	 \\ \vspace{0.1cm}
         \begin{subfigure}[b]{0.32\linewidth}
 	 \includegraphics[width=4.5cm, height=4.5cm]{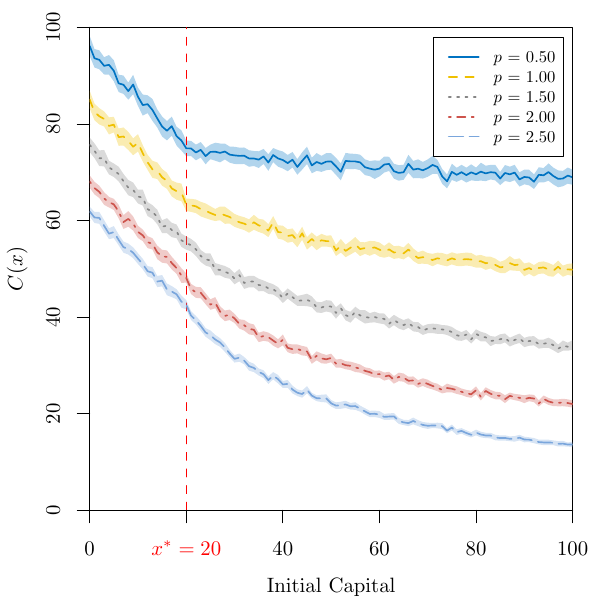}
	 \caption{}
  	\label{GeneralCaseAnalysis:AbsenceofClosed-FormSolutions-Section6-Figure1-d}
	\end{subfigure}
 	 \hspace{0.01cm}
         \begin{subfigure}[b]{0.32\linewidth}
 	 \includegraphics[width=4.5cm, height=4.5cm]{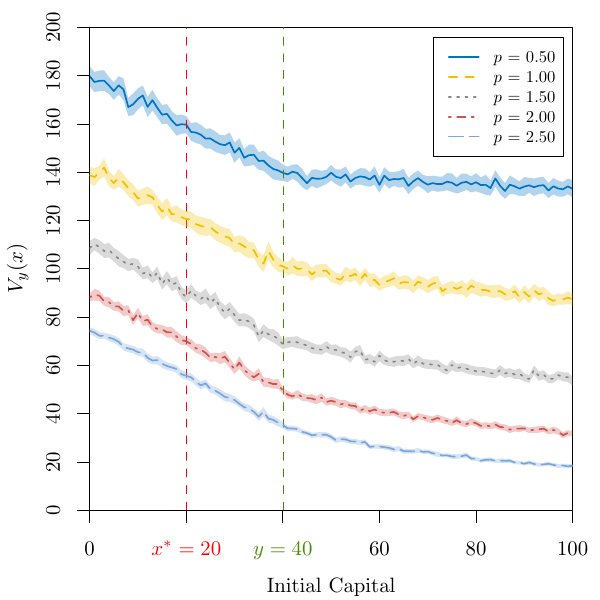}
	 \caption{}
  	\label{GeneralCaseAnalysis:AbsenceofClosed-FormSolutions-Section6-Figure1-e}
	\end{subfigure}
  	\hspace{0.01cm}
         \begin{subfigure}[b]{0.32\linewidth}
 	 \includegraphics[width=4.5cm, height=4.5cm]{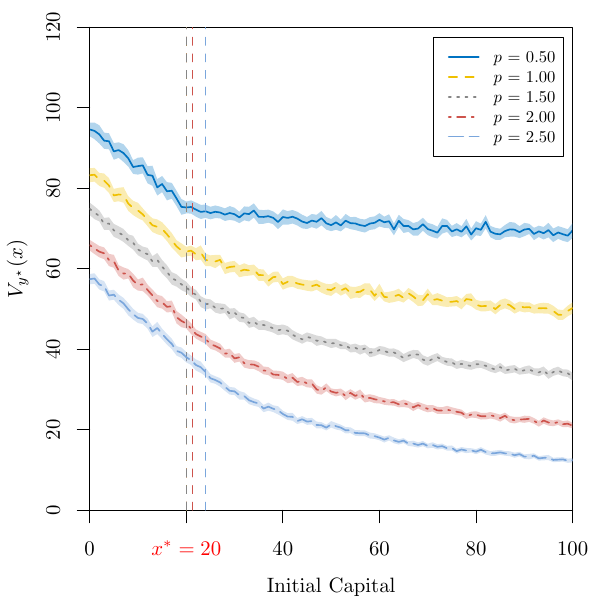}
	 \caption{}
  	\label{GeneralCaseAnalysis:AbsenceofClosed-FormSolutions-Section6-Figure1-f}
	\end{subfigure}
	 \\ \vspace{0.1cm}
	\caption{The first row (a)--(c) shows the cost of social protection and the value function of a threshold strategy when $Z_{i} \sim Kumaraswamy(3, 4)$, $a = 0.10$, $b = 3$, $c = 0.40$, $\lambda = 1$, $x^{*} = 20$ for $\delta = 0.10, 0.20, 0.30, 0.40, 0.50$. The second row (d)--(f) displays the cost of social protection and the value function of a threshold strategy when $Z_{i} \sim Kumaraswamy(p, 4)$, $a = 0.10$, $b = 3$, $c = 0.40$, $\lambda = 1$, $\delta = 0.25$, $x^{*} = 20$, $y=40$ for $p = 0.50, 1.00, 1.50, 2.00, 2.50$. Panels (a) and (d) display $C(x)$ corresponding to $V_{y}(x)$ with threshold $y=20$, while (b) and (e) show $V_{y}(x)$ with threshold $y=40$. Similarly, panels (c) and (f) show $V_{y^{\star}}(x)$ with the optimal threshold $y^{\star}$.}
	\label{GeneralCaseAnalysis:AbsenceofClosed-FormSolutions-Section6-Figure1}
\end{figure}

\begin{figure}[H]
	\begin{subfigure}[b]{\linewidth}
\centering
  		\includegraphics[width=7.5cm, height=7.5cm]{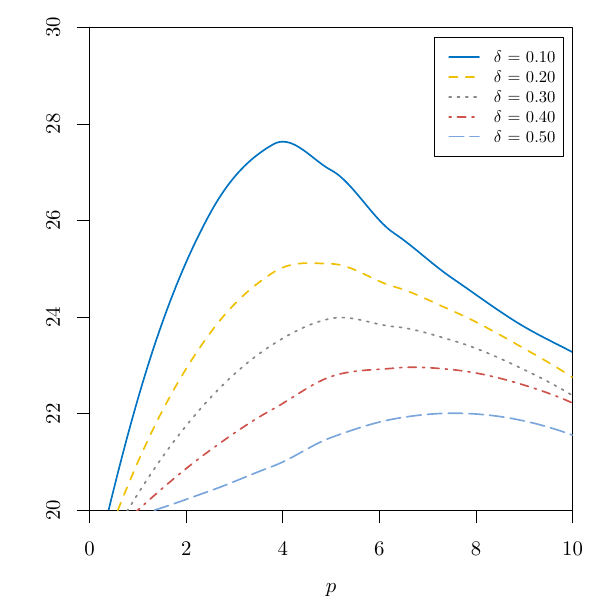}
  		\label{GeneralCaseAnalysis:AbsenceofClosed-FormSolutions-Section6-Figure2}
	\end{subfigure}
	\caption{The optimal threshold $y^{\star}$ when $Z_{i} \sim Kumaraswamy(p, 4)$, $a = 0.10$, $b = 3$, $c = 0.40$, $\lambda = 1$, $x^{*} = 20$ for different values of the discount rate $\delta$.}
	\label{GeneralCaseAnalysis:AbsenceofClosed-FormSolutions-Section6-Figure2}
\end{figure}

\section{Microinsurance} \label{Microinsurance-Section7}

We consider microinsurance as a complementary instrument within social protection strategies (see, for example, \cite{Book:Churchill2012}). This section examines its role alongside CT programmes by considering three types of coverage: (i) proportional, (ii) excess-of-loss (XL), and (iii) total-loss.

Let $R:[0,1]\rightarrow\lbrack0,1]$ be the retained loss function, satisfying $0\leq R(u)\leq u$ and $R(0)=0$. This function represents the portion of the loss borne by the household when there is a loss of $u=1-z\in\lbrack0,1]$ per unit of capital. Consequently, following a loss, the capital of a household changes from $X_{\tau_{i}}$ to $\left(1-R(u)\right)\cdot X_{\tau_{i}}$ after the loss. The insurer calculates the premium rate using the \textit{expected value principle}. That is,

\vspace{0.3cm}

\begin{align}
p_{R}=\left(1+\gamma\right)\cdot \lambda \cdot \mathbb{E}\left[1-Z-R(1-Z)\right],
\label{Microinsurance-Section7-Equation1}
\end{align}

\vspace{0.3cm}

where $\gamma>0$ is the safety loading per unit of capital. The critical capital is now given by

\vspace{0.3cm}

\begin{align}
x^{\ast R}=\left(\frac{b}{b-p_{R}}\right)x^{\ast}\geq x^{\ast},
\label{Microinsurance-Section7-Equation2}
\end{align}

\vspace{0.3cm}

since $x^{\ast}=I^{\ast}/b$ and $x^{\ast R}=I^{\ast}/\left(b-p_{R}\right)$, where $I^{\ast}$ represents the critical income and, to ensure that poverty does not occur with certainty, we further assume that $b > p_{R}$ (see \cite{Article:Kovacevic2011}). In the framework considered here, the insurance contract with coverage $R$ is fixed before the implementation of the transfer policy. Hence, the coverage level and the corresponding premium are not decision variables in the optimisation problem.

The increase of the critical capital from $x^\ast$ to $x^{\ast R}$ reflects the fact that the household must remain above the critical income level after paying the insurance premium. Therefore, households with capital levels satisfying

\vspace{0.3cm}

\begin{align}
x^\ast<X_t<x^{\ast R}
\end{align}

\vspace{0.3cm}

are above the original critical level, but their surplus income is not sufficient to afford the fixed insurance contract $R$ while keeping their net income above $I^\ast$. Under the microinsurance policy considered here, the transfer mechanism is adjusted to cover this additional region: households are supported up to the higher level $x^{\ast R}$, making the prescribed insurance contract affordable.

The payment of the premium also modifies the growth rate of the household capital, which becomes

\vspace{0.3cm}

\begin{align}
r^{R}=(1-a)\cdot (b-p_{R})\cdot c=(1-a)\cdot b\left(\frac{b-p_{R}}{b}\right)\cdot c=r\left(\frac{b-p_{R}}{b}\right)<r,
\label{Microinsurance-Section7-Equation3}
\end{align}

\vspace{0.3cm}

since $r=(1-a)\cdot b\cdot c$. Taking this into account, the optimal problem of capital lump-sum transfers with microinsurance can be viewed as the problem without insurance coverage, discussed in previous sections, but now considering $r^{R}<r$ instead of $r$, the critical capital $x^{\ast R}>x^{\ast}$ instead of $x^{\ast}$, and the cumulative distribution function of the remaining proportion of capital $G^{R}_{W}\leq G_{Z}$ instead of $G_{Z}$. Here, $G^{R}_{W}$ denotes the c.d.f. of the remaining proportion of capital, $W=1-R(1-Z)$, retained by the household after a loss. In this case, the infinitesimal generator of the process is given by

\vspace{0.3cm}

\begin{align}
A^{R}(f)(x) & = r^{R}(x-x^{\ast R})^{+}f^{\prime}(x)-\lambda f(x)+\lambda\int_{0}^{1}f\left(  x \cdot (1-R(1-z)\right)  dG_{Z}(z)\\ \\
& = r^{R}(x-x^{\ast R})^{+}f^{\prime}(x)-\lambda f(x)+\lambda\int_{0}^{1}f\left(  x\cdot w\right)  dG^{R}_{W}(w),
\label{Microinsurance-Section7-Equation4}
\end{align}

\vspace{0.3cm}

where

\vspace{0.3cm}

\begin{align}
G^{R}_{W}(w)=\mathbbm{P}(W\leq w)=\mathbbm{P}(1-R(1-Z)\leq w). 
\label{Microinsurance-Section7-Equation5}
\end{align}

\vspace{0.3cm}

Let us call: $h(z):=1-R(1-z)$. This function satisfies $h(z)\geq z$. In the case that the function $h$: $[0,1]\rightarrow\lbrack1-R(1),1]$ is invertible and non-decreasing, and calling $w=h(z)$, if $w\geq1-R(1)$, it yields

\vspace{0.3cm}

\begin{align}
G^{R}_{W}(w)=\mathbbm{P}(W\leq w)=\mathbbm{P}(1-R(1-Z)\leq w)=\mathbbm{P}(Z\leq h^{-1}(w)).
\label{Microinsurance-Section7-Equation6}
\end{align}

\vspace{0.3cm}

Hence,

\vspace{0.3cm}

\begin{align}
G^{R}_{W}(w)=\left\{
\begin{array}
[c]{ccc}
G_{Z}(h^{-1}(w)) & \textit{if} & 1-R(1)\leq w\leq1,\\
0 & \textit{if} & w<1-R(1).
\end{array}
\right.
\label{Microinsurance-Section7-Equation7}
\end{align}

\vspace{0.3cm}

If $h$ is non-invertible we will obtain $G^{R}_{W}$ in each special case.

This formulation allows us to investigate whether the increase in the intervention threshold from $x^\ast$ to $x^{\ast R}$ might, in some cases, be compensated by the benefits generated by the insurance contract, namely the modified growth rate $r^R$ and the reduction of future losses described by the post-insurance loss distribution $G_W^R$.

\subsection{Proportional Microinsurance} \label{Microinsurance-Section7-Subsection1}

In the proportional case, let $\eta$ denote the fraction of the household\rq s capital lost after a proportional loss of $u$. That is, $R(u)=\eta u$, where $\eta=0$ and $\eta=1$ mean total and no insurance cover, respectively (with $R(1)=\eta$). The premium rate per unit of capital in this case is given by

\vspace{0.3cm}

\begin{align}
p_{R}=\left(1+\gamma\right)\cdot\lambda \cdot \mathbb{E}\left[1-Z-\eta(1-Z)\right]=\left(  1+\gamma\right)\cdot \lambda \cdot(1-\eta) \cdot (1-\mathbb{E}\left[Z\right]).
\label{Microinsurance-Section7-Equation8}
\end{align}

\vspace{0.3cm}

Thus, we have,

\vspace{0.3cm}

\begin{align}
h(z)=1-R(1-z)=1-\eta(1-z)=w,
\label{Microinsurance-Section7-Equation9}
\end{align}

\vspace{0.3cm}

and therefore, $h$ is invertible and non-decreasing with

\vspace{0.3cm}

\begin{align}
h^{-1}(w)=\frac{1}{\eta}\left(w+\eta-1\right).
\label{Microinsurance-Section7-Equation10}
\end{align}

\vspace{0.3cm}

Then, from \eqref{Microinsurance-Section7-Equation7} we get

\vspace{0.3cm}

\begin{align}
G^{R}_{W}(w)=\left\{
\begin{array}
[c]{ccc}%
G_{Z}\left(\frac{1}{\eta}\left(w+\eta-1\right)\right) & \textit{if} & 1-\eta\leq w\leq1,\\
0 & \textit{if} & w<1-\eta,
\end{array}
\right.
\label{Microinsurance-Section7-Equation11}
\end{align}

\vspace{0.3cm}

which is continuous if $G_{Z}$ continuous. Moreover, we also have the following,

\vspace{0.3cm}

\begin{align}
\int_{0}^{1}f\left(x\cdot w\right)dG_{W}^{R}(w)=\int_{1-\eta}^{1}f\left(x\cdot w\right) dG_{Z}\left(\frac{1}{n}\left(w+\eta-1\right)\right).
\label{Microinsurance-Section7-Equation12}
\end{align}

\vspace{0.3cm}

\begin{example}

We consider a household that faces capital losses, where the remaining proportion of capital follows a $Beta(\alpha,1)$ distribution; that is, $Z_{i}\sim Beta(\alpha,1)$. This corresponds to the case discussed in Section \ref{Closed-FormSolutionsinaSpecialCase-Section5}. In addition, we now assume that the household acquires proportional microinsurance coverage. Under this setting, the value function associated with a threshold strategy (together with its optimal threshold) can be estimated via the Monte Carlo procedure described in detail in Section \ref{GeneralCaseAnalysis:AbsenceofClosed-FormSolutions-Section6}. Specifically, let $Y\sim Unif[0,1]$. Then $Z=Y^{1/\alpha}$ follows a $Beta(\alpha,1)$ distribution. Under proportional microinsurance coverage, the remaining proportion of capital is given by 

\vspace{0.3cm}

\begin{align}
W=1-R(1-Z)=1-\eta(1-Z).
\label{Microinsurance-Section7-Equation13}
\end{align}

\vspace{0.3cm}

Hence, $1-\eta \leq W\leq1$, and can equivalently be written as 

\vspace{0.3cm}

\begin{align}
W=1-\eta(1-Y^{1/\alpha}).
\label{Microinsurance-Section7-Equation14}
\end{align}

\vspace{0.3cm}

Therefore, in the Monte Carlo simulations, we now consider the sequence $(\tau_{n},W_{n})$, where $W_{n}$ is given by \eqref{Microinsurance-Section7-Equation14}.

Figure \ref{Microinsurance-Section7-Figure1} presents the value function of a threshold strategy with the optimal threshold, $y^{\star}$, for a household with proportional microinsurance coverage. Figures \ref{Microinsurance-Section7-Figure1-a} and \ref{Microinsurance-Section7-Figure1-b} show that the value function increases with both a higher fraction of the household\rq s capital lost after a proportional loss $\eta$ and a higher loading factor $\gamma$, respectively. This result is expected, as lower insurance coverage (i.e., larger values of $\eta$, where $\eta = 1$ indicates no insurance coverage) requires governments to inject larger amounts of capital. Similarly, higher values of $\gamma$ increase the premium rate \eqref{Microinsurance-Section7-Equation8}, which in turn reduces the household\rq s capital growth rate \eqref{Microinsurance-Section7-Equation3} and raises the critical capital level \eqref{Microinsurance-Section7-Equation2}, thereby increasing the need for government injections.

Although lower insurance coverage reduces the premiums paid by households (thereby reducing the need for government injections), purchasing proportional microinsurance coverage plays an important role in reducing the size of government injections required. Moreover, all the lines in Figure \ref{Microinsurance-Section7-Figure1} lie below the blue solid line in Figure \ref{Closed-FormSolutionsinaSpecialCase-Section5-Figure1-c}, indicating that adding proportional microinsurance coverage for households helps reduce the government’s cost of capital transfers. The optimal thresholds in Figures \ref{Microinsurance-Section7-Figure1-a} (for $\eta = 0.60, 0.70, 0.80, 0.90, 1.00$) and \ref{Microinsurance-Section7-Figure1-b} (for $\gamma = 0.50, 1.50, 2.50, 3.50, 4.50$) are given by $y^{\star} = 28.09, 27.70, 27.38, 28.15, 26.16$ and $y^{\star} = 28.38, 30.67, 33.42, 39.98, 45.56$, respectively. It can be verified that the conditions stated in Remark \ref{ThresholdStrategies-Section4-Remark1} are satisfied. Therefore, by Theorem \ref{AnalysisofV:PropertiesandtheHamilton-Jacobi-BellmanEquation-Section3-Theorem1}, the corresponding value function is optimal among all admissible strategies.

\begin{remark}
	It is important to note that when the insurance premium $p_{R}$ is close in value to the income generation rate $b$ (while still satisfying the assumption that $b>p_{R}$), the denominator of Equation \eqref{Microinsurance-Section7-Equation2} tends to zero, and consequently, the new poverty line $x^{*R}$ approaches infinity. In this case, and for that reason, microinsurance does not reduce the social protection costs borne by the government. However, this situation is not very realistic, as a household would not allocate all of its income to the payment of an insurance premium. In fact, in our analyses, we observed that cost reductions failed to materialise only in those scenarios where the premium was extremely close to the income generation rate. Conversely, social protection costs were effectively reduced for reasonably high but still feasible premium levels (corresponding to high values of $\gamma$ in our example). Therefore, microinsurance contributes to lowering social protection costs in realistic settings where premiums remain economically attainable for households.
\end{remark}

\begin{figure}[H]
	\begin{subfigure}[b]{0.5\linewidth}
  		\includegraphics[width=7.5cm, height=7.5cm]{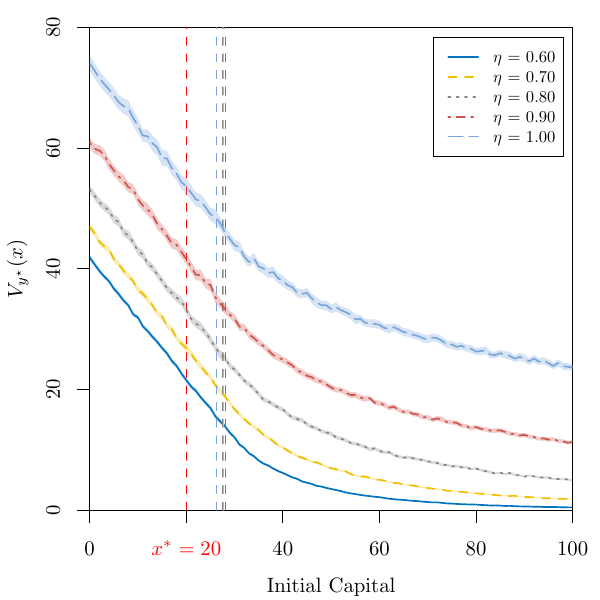}
		\caption{}
  		\label{Microinsurance-Section7-Figure1-a}
	\end{subfigure}
	\begin{subfigure}[b]{0.5\linewidth}
  		\includegraphics[width=7.5cm, height=7.5cm]{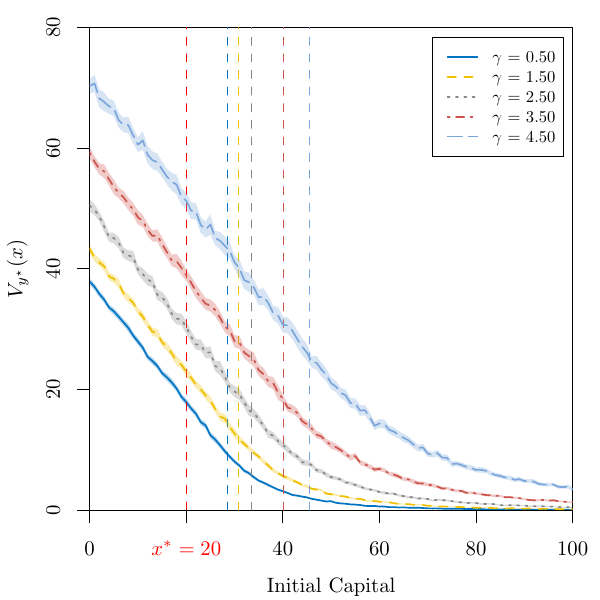}
		\caption{}
  		\label{Microinsurance-Section7-Figure1-b}
	\end{subfigure}
	\caption{The value function of a threshold strategy with the optimal threshold, $V_{y^{\star}}(x)$, of a household with proportional microinsurance coverage when $Z_{i} \sim Beta(\alpha,1)$, $a = 0.10$, $b = 3$, $c = 0.40$, $\lambda = 1$, $\delta = 0.10$, $x^{*} = 20$, $\alpha = 1.25$ for (a) fixed $\gamma = 0.50$ and different values of $\eta$ and (b) fixed $\eta = 0.50$ and different values of $\gamma$.}
	\label{Microinsurance-Section7-Figure1}
\end{figure}

\end{example}

\subsection{Excess-of-Loss Microinsurance} \label{Microinsurance-Section7-Subsection2}

Excess-of-Loss (XL) microinsurance covers losses only above a specified threshold (the \textit{retention limit}) per unit of capital, leaving smaller losses to be managed by the insured. More precisely, in the XL case with retention limit $l\in\lbrack0,1]$, if there is a loss of $u=1-z\in\lbrack0,1]$ per unit of capital, the household loses

\vspace{0.3cm}

\begin{align}
R(u)=u\cdot \mathbbm{1}_{\{u\leq l\}}+l\cdot  \mathbbm{1}_{\{u>l\}}=\min\{u,l\},
\label{Microinsurance-Section7-Equation15}
\end{align}

\vspace{0.3cm}

per unit of capital. Thus, the microinsurance provider covers the following fraction of the capital:

\vspace{0.3cm}

\begin{align}
u-R(u) = 0\cdot \mathbbm{1}_{\{u\leq l\}}+(u-l)\cdot \mathbbm{1}_{\{u>l\}}= (u-l)\cdot \mathbbm{1}_{\{u>l\}}.
\label{Microinsurance-Section7-Equation16}
\end{align}

\vspace{0.3cm}

The premium rate per unit of capital, calculated according to the expected value principle, is therefore given by

\vspace{0.3cm}

\begin{align}
p_{R} = \left(1+\gamma\right)\cdot\lambda\cdot\mathbb{E}\left[(1-Z-l)\cdot \mathbbm{1}_{\{Z<1-l\}}\right]= \left(1+\gamma\right)\cdot \lambda\cdot\int_{0}^{1-l}(1-z-l)dG_{Z}(z).
\label{Microinsurance-Section7-Equation17}
\end{align}

\vspace{0.3cm}

It follows that

\vspace{0.3cm}

\begin{align}
h(z) =1-\min\{1-z,l\} =\max\{1-l,z\},
\label{Microinsurance-Section7-Equation18}
\end{align}

\vspace{0.3cm}

for which no inverse exists. From the definition of $G_W^R(w)$, the c.d.f of $W=\max\{1-l,Z\}$ is given by,

\vspace{0.3cm}

\begin{align}
G_{W}^{R}(w)=\left\{
\begin{array}
[c]{ccc}%
0 & \textit{if} & w<1-l,\\
G_{Z}(w) & \textit{if} & w\geq1-l,
\end{array}
\right.
\label{Microinsurance-Section7-Equation19}
\end{align}

\vspace{0.3cm}

which has an upward jump of $G_{Z}(1-l)$ at $1-l$ if $>0$. We also obtain,

\vspace{0.3cm}

\begin{align}
\int_{0}^{1}f\left(x\cdot w\right)  dG^{R}(w)=f\left(x\left(1-l\right)\right)\cdot G_{Z}(1-l)+\int_{1-l}^{1}f\left(x \cdot w\right)dG_{Z}(w).
\label{Microinsurance-Section7-Equation20}
\end{align}

\begin{example}

We revisit the case when $Z_{i} \sim Beta(\alpha,1)$, as discussed in Section \ref{Closed-FormSolutionsinaSpecialCase-Section5}. Under excess-of-loss (XL) microinsurance coverage, the remaining proportion of capital is given by


\begin{align}
W & =\max\{1-l,Z\}\\
& = Z \cdot \mathbbm{1}_{\{1-l\leq Z\}}+(1-l)\cdot \mathbbm{1}_{\{1-l>Z\}} \\
& =Y^{1/\alpha}\cdot \mathbbm{1}_{\{1-l\leq Y^{1/\alpha}\}}+(1-l)\cdot \mathbbm{1}_{\{1-l>Y^{1/\alpha}\}},
\label{Microinsurance-Section7-Equation21}
\end{align}

\vspace{0.3cm}

where $Y\sim Unif[0,1]$.

Figure \ref{Microinsurance-Section7-Figure2} displays the value function of a threshold strategy with the optimal threshold $y^{\star}$ for a household with excess-of-loss microinsurance coverage. As expected, and consistent with the proportional microinsurance case, the value function increases with both a higher retention limit per unit of capital $l$ and a higher loading factor $\gamma$.

Furthermore, as anticipated, the value function for the XL microinsurance coverage is higher than that of the proportional cover. This is consistent with the fact that the XL policy only covers the losses above the retention level per unit of capital, whereas the proportional policy covers a fixed fraction $1-\eta$ of the loss. Consequently, the government would be expected to provide larger capital injections to households with XL microinsurance coverage.

Figure \ref{Microinsurance-Section7-Figure2-b} also shows that the value function for the XL microinsurance is less sensitive with respect to the loading factor $\gamma$ than in the proportional case, since the expected ceded loss is lower under the XL coverage. As highlighted before, because the expected ceded loss is higher under proportional microinsurance, the corresponding value function tends to be lower, as the required government injections are smaller. As for the proportional case, all the lines in Figure \ref{Microinsurance-Section7-Figure2} lie below the blue solid line in Figure \ref{Closed-FormSolutionsinaSpecialCase-Section5-Figure1-c}, indicating that adding XL microinsurance coverage for households helps reduce the government\rq s cost of capital injections. The optimal thresholds in Figure \ref{Microinsurance-Section7-Figure2-a} (for $l = 0.60, 0.70, 0.80, 0.90, 1.00$) and \ref{Microinsurance-Section7-Figure2-b} (for $\gamma = 0.50, 1.50, 2.50,$ $3.50, 4.50$) are given by $y^{\star} = 28.57, 28.24, 26.46, 26.20, 26.25$ and $y^{\star} = 28.75, 29.89, 30.90, 33.57, 34.61$, respectively. After verifying that the conditions in Remark \ref{ThresholdStrategies-Section4-Remark1} are satisfied, we conclude that the resulting value function $V_{y^{\star}}(x)$ is optimal among all admissible strategies.

\begin{figure}[H]
	\begin{subfigure}[b]{0.5\linewidth}
  		\includegraphics[width=7.5cm, height=7.5cm]{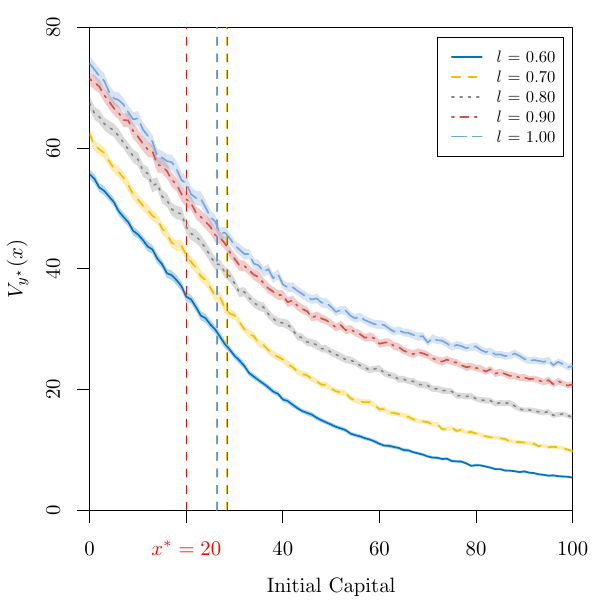}
		\caption{}
  		\label{Microinsurance-Section7-Figure2-a}
	\end{subfigure}
	\begin{subfigure}[b]{0.5\linewidth}
  		\includegraphics[width=7.5cm, height=7.5cm]{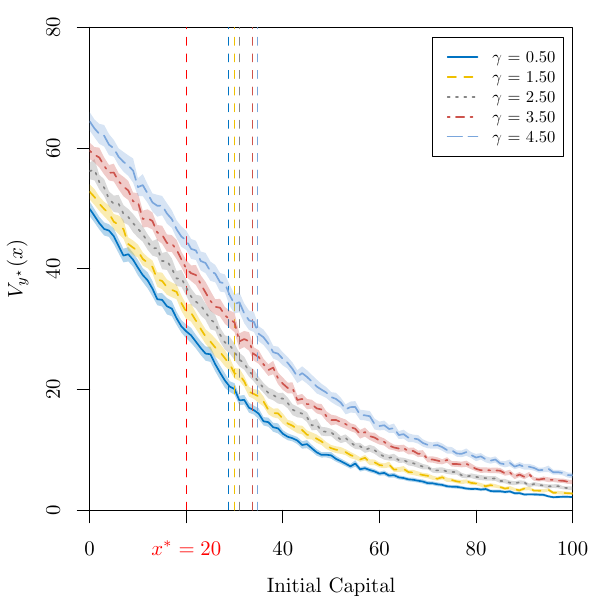}
		\caption{}
  		\label{Microinsurance-Section7-Figure2-b}
	\end{subfigure}
	\caption{The value function of a threshold strategy with the optimal threshold, $V_{y^{\star}}(x)$, of a household with excess-of-loss microinsurance coverage when $Z_{i} \sim Beta(\alpha,1)$, $a = 0.10$, $b = 3$, $c = 0.40$, $\lambda = 1$, $\delta = 0.10$, $x^{*} = 20$, $\alpha = 1.25$ for (a) fixed $\gamma = 0.5$ and different values of $l$ and (b) fixed $l = 0.5$ and different values of $\gamma$.}
	\label{Microinsurance-Section7-Figure2}
\end{figure}

\end{example}

\subsection{Total-Loss Microinsurance} \label{Microinsurance-Section7-Subsection3}

In a total-loss microinsurance policy, the microinsurance provider covers everything from a proportional loss of $L\in\lbrack0,1]$ onwards. That is, the household loses $R(u)=$ $u\cdot \mathbbm{1}_{\{u\leq L\}}$ per unit of capital after experiencing a proportional loss of $u=1-z$ while the microinsurance provider covers (per unit of capital),

\vspace{0.3cm}

\begin{align}
u-R(u)=0\cdot \mathbbm{1}_{\{u\leq L\}}+u\cdot\mathbbm{1}_{\{u>L\}}=(1-z)\cdot \mathbbm{1}_{\{1-z>L\}}=(1-z)\cdot \mathbbm{1}_{\{z<1-L\}}.
\label{Microinsurance-Section7-Equation22}
\end{align}

\vspace{0.3cm}

Hence, from \eqref{Microinsurance-Section7-Equation1}, the premium rate per unit of capital is given by,

\vspace{0.3cm}

\begin{align}
p_{R} = \left(1+\gamma\right)\cdot \lambda \cdot\mathbb{E}\left[(1-Z)\cdot \mathbbm{1}_{\{z<1-L\}}\right] = \left(1+\gamma\right)\cdot \lambda \cdot \int_{0}^{1-L}(1-z)dG_{Z}(z),
\label{Microinsurance-Section7-Equation23}
\end{align}

\vspace{0.3cm}

which is greater than the premium rate of the excess-of-loss insurance when $l=L$. Moreover,

\vspace{0.3cm}

\begin{align}
h(z) = 1-(1-z)\cdot \mathbbm{1}_{\{z\geq1-L\}} = z\cdot \mathbbm{1}_{\{z\geq1-L\}}+\mathbbm{1}_{\{z<1-L\}},
\label{Microinsurance-Section7-Equation24}
\end{align}

\vspace{0.3cm}

which is non-invertible. Then, from the definition of $G_W^R(w)$, it yields,

\vspace{0.3cm}

\begin{align}
G_{W}^{R}(w) 
&  =\left\{
\begin{array}
[c]{lll}%
0 & \textit{if} & 0\leq w<1-L,\\
G_{Z}(w)-G_{Z}(1-L) & \textit{if} & 1-L\leq w<1,\\
1 & \textit{if} & w=1,
\end{array}
\right.
\label{Microinsurance-Section7-Equation25}
\end{align}

\vspace{0.3cm}

which is continuous at $w=1-L$ and has an upward jump of $G_{Z}(1-L)$ at $w=1$. In addition, we have the following,

\vspace{0.3cm}

\begin{align}
\int_{0}^{1}f\left(x\cdot w\right)dG_{W}^{R}(w) = \int_{1-L}^{1}f\left(x\cdot w\right)dG_{W}^{R}(w) = \int_{1-L}^{1}f\left(x\cdot w\right)dG_{Z}(w)+f\left(x\right)\cdot G_{Z}(1-l).
\label{Microinsurance-Section7-Equation26}
\end{align}

\vspace{0.3cm}

\begin{example}

Let $Z_{i} \sim Beta(\alpha,1)$. Under total-loss insurance coverage, the remaining proportion of capital can be expressed as

\vspace{0.3cm}

\begin{align}
W=1-R(1-Z)=Y^{1/\alpha}\cdot \mathbbm{1}_{\{Y^{1/\alpha}\geq1-L\}}+\mathbbm{1}_{\{Y^{1/\alpha}<1-L\}},
\label{Microinsurance-Section7-Equation27}
\end{align}

\vspace{0.3cm}

with $Y\sim Unif[0,1]$.

The value function of a threshold strategy with the optimal threshold $y^{\star}$ for a household under a total-loss microinsurance policy is shown in Figure \ref{Microinsurance-Section7-Figure3}. The total-loss coverage exhibits behaviour consistent with the proportional and excess-of-loss microinsurance policies, as the value function increases with both the loss level $L$ and the loading factor $\gamma$. Under this type of microinsurance, extremely large losses (those exceeding values of $L$ which are close to one) are fully covered by the insurer. Consequently, the government\rq s expenditure to catastrophic shocks affecting households is reduced.

\begin{figure}[H]
	\begin{subfigure}[b]{0.5\linewidth}
  		\includegraphics[width=7.5cm, height=7.5cm]{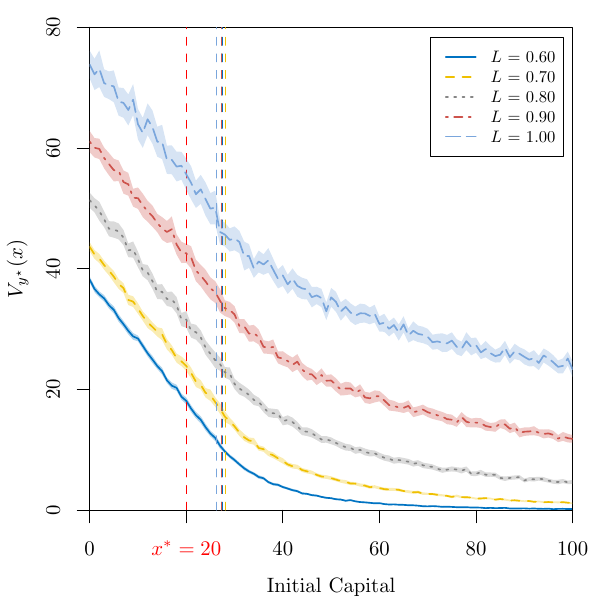}
		\caption{}
  		\label{Microinsurance-Section7-Figure3-a}
	\end{subfigure}
	\begin{subfigure}[b]{0.5\linewidth}
  		\includegraphics[width=7.5cm, height=7.5cm]{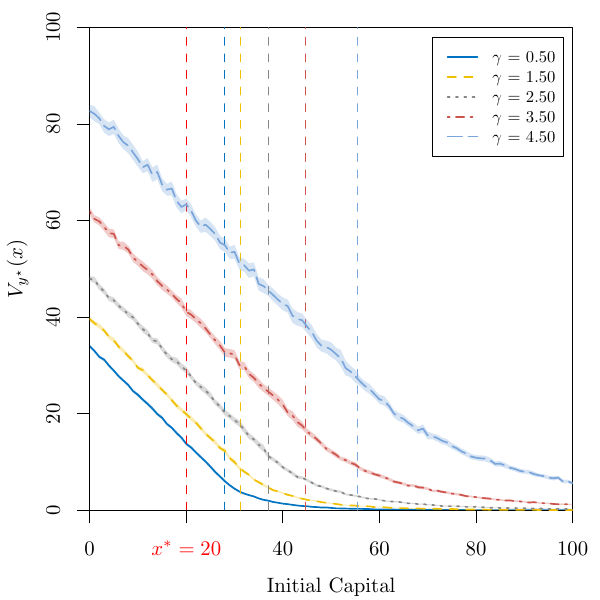}
		\caption{}
  		\label{Microinsurance-Section7-Figure3-b}
	\end{subfigure}
	\caption{The value function of a threshold strategy with the optimal threshold, $V_{y^{\star}}(x)$, of a household with total-loss microinsurance coverage when $Z_{i} \sim Beta(\alpha,1)$, $a = 0.10$, $b = 3$, $c = 0.40$, $\lambda = 1$, $\delta = 0.10$, $x^{*} = 20$, $\alpha = 1.25$ for (a) fixed $\gamma = 0.5$ and different values of $L$ and (b) fixed $L = 0.5$ and different values of $\gamma$.}
	\label{Microinsurance-Section7-Figure3}
\end{figure}

This mechanism explains why, for higher values of $\eta$ and $L$, the value function under the proportional and total-loss coverages are very similar, as the government is required to inject a smaller share of the household\rq s loss. That is, under the proportional policy, a small portion of the household\rq s losses is absorbed by the insurer (for higher values of $\eta$), while the likelihood of a catastrophic event exceeding the threshold $L$ (and thus triggering full coverage under the total-loss policy) is relatively low (for higher values of $L$). On the other hand, the value function under total-loss microinsurance is lower than that under proportional coverage for lower values of $\eta$ and $L$, as it becomes more likely that the household experiences losses exceeding $L$, which are fully covered by the insurer. As a result, the government needs to provide smaller capital injections. The optimal thresholds in Figure \ref{Microinsurance-Section7-Figure3-a} (for $L = 0.60, 0.70, 0.80, 0.90, 1.00$) and \ref{Microinsurance-Section7-Figure3-b} (for $\gamma = 0.50, 1.50, 2.50,$ $3.50, 4.50$) are given by $y^{\star} = 27.47, 28.07, 27.26, 27.24, 26.22$ and $y^{\star} = 27.76, 31.25, 36.87, 44.56, 55.46$, respectively. The conditions in Remark \ref{ThresholdStrategies-Section4-Remark1} are satisfied, thereby establishing the optimality of the resulting value function among all admissible strategies.

\end{example}

\section{Conclusion} \label{Conclusion-Section8}

This article examines social assistance programmes, focusing in particular on cash transfers (CTs), from a stochastic control perspective. It employs the piecewise-deterministic Markov process introduced by \cite{Article:Kovacevic2011} to model household capital and aims to optimise direct government transfers, defined as the expected discounted cost of maintaining a household\rq s capital above the poverty line through capital injections. The Hamilton-Jacobi-Bellman (HJB) equation associated with the stochastic control problem of determining the optimal transfer amount over time is derived. In certain special cases, the equation is solved analytically, yielding closed-form expressions for the expected discounted cost.

Our results indicate the existence of an optimal injection level above the poverty threshold, suggesting that resources are used more efficiently when CTs are preventive rather than reactive, since injecting capital into households when their capital is above the poverty line is less costly than doing so only after it falls below the poverty threshold.

Building on this framework, we further examine the combination of microinsurance and CT programmes to assess how these two instruments can jointly enhance social protection outcomes. While CTs provide immediate support to households by maintaining their capital above certain threshold (such as the poverty line or a higher reference level), microinsurance offers risk-transfer mechanisms that reduce the magnitude and impact of adverse shocks on household capital in exchange for a premium paid by the household. Our results, consistent with previous studies employing similar techniques and models (e.g., \cite{Article:Flores-Contro2025}), suggest that blending these two instruments can generate complementary effects, leading to more cost-efficient CT policies and reducing the overall fiscal burden on governments.

It is important to emphasise that this study is theoretical in nature and does not rely on empirical data. As with any model-based analysis, certain limitations should be acknowledged. In particular, the framework does not account for potential behavioural changes among households that might occur once they become aware that their capital level is approaching the threshold for cash transfer eligibility, which is a tendency that has been observed in other studies (see, for example, \cite{Article:Firpo2014} and \cite{Article:Habimana2021}). Of course, this represents one of several limitations that could be highlighted. A natural avenue for future research is therefore the collection and analysis of empirical data to determine whether the theoretical results obtained here are supported by real-world evidence. Such empirical validation would not only strengthen the robustness of the model but also provide valuable insights for designing effective and adaptive social protection strategies. 

Another interesting direction for future research would be to consider alternative optimisation criteria, such as maximising the expected discounted consumption above a minimum level, complementing the cost-minimisation approach considered in this article.

A further extension would be to allow the microinsurance contract itself to depend on the household\rq s wealth. This would lead to a more general framework in which the optimal transfer policy is combined with the design of the insurance contract, allowing the coverage level (and the corresponding premium) to be adjusted according to the household\rq s wealth.

Lastly, an additional promising avenue for future research is to determine the optimal level of insurance a household should purchase in order to minimise the expected discounted cost of maintaining its capital above the poverty line (or above a specified threshold). Every household receives different income levels, faces heterogeneous risks, and starts with distinct initial capital levels. Hence, the optimal level of coverage will vary accordingly, making this framework useful for understanding inequality and designing targeted policies.

\section*{Funding}  \label{Funding-Section} 

José Miguel Flores-Contró gratefully acknowledges funding from the FWO and F.R.S.-FNRS under the Excellence of Science (EOS) programme, project ASTeRISK (40007517).

\section*{Acknowledgements}  \label{Acknowledgements-Section}

The authors are grateful to the anonymous referees for their numerous helpful and constructive comments on an earlier version of this manuscript.


\section*{Appendices}
\addcontentsline{toc}{section}{Appendices}
\renewcommand{\thesubsection}{\Alph{subsection}}
\setcounter{subsection}{0}

\numberwithin{equation}{subsubsection}

\subsection{Mathematical Proofs}\label{Appendix A: Mathematical Proofs}

\subsubsection{Proof of Lemma \ref{AnalysisofV:PropertiesandtheHamilton-Jacobi-BellmanEquation-Section3-Lemma1}} \label{ProofofLemma3.1}

\begin{proof}
	Let $(X_{t})_{t\geq 0}$ be a process as defined in Subsection \ref{TheStochasticControlProblem-Section2-Subsection21} and let $(X_{t}^{0})_{t\geq 0}$ be the auxiliary process obtained by suppressing the growth above $x^{\ast }$ (that is, by setting $r=0$ in the definition of the original process), both starting at initial capital $x$. Let $C^{0}(x)$ be the corresponding cost of social protection. Since $X_{t}^{0}\leq X_{t}$ pathwise for all $t$, it follows that $C^{0}(x)\geq C(x)$ because starting from a larger path can only reduce the need for capital injections and also that $C^{0}(x^{*}) = C(x^{*})$. Hence, it is enough to prove that $\lim_{x\rightarrow +\infty }C^{0}(x)=0.$

	In the auxiliary model $(X_{t}^{0})_{t\geq 0}$, the capital remains constant between jumps. Starting from an initial capital level $x>x^{\ast }$, the first time it falls below $x^{\ast }$occurs at the jump index
	
	\begin{align}
		J_{x}=\min \left\{ j\geq 1:x\cdot Z_{1}\cdot Z_{2}\cdot\cdot\cdot Z_{j}<x^{\ast }\right\} =\min\left\{ j\geq 1:\sum_{i=1}^{j}\left(-\log \left(Z_{i}\right)\right)>\log \left(\frac{x}{x^{\ast }}\right)\right\}.
		\label{Appendix A: Mathematical Proofs-Equation1}
	\end{align}
	
	Define
	
	\begin{align}
		Y_{i}:=-\log \left(Z_{i}\right)\text{,}\quad R_{j}:=\sum_{i=1}^{j}Y_{i}\quad \text{ and } \quad a_{x}:=\log \left(\frac{x}{x^{\ast }}\right)\text{.}
		\label{Appendix A: Mathematical Proofs-Equation2}
	\end{align}
	
	Then, $R_{j}$ is a random walk with i.i.d. increments $Y_{i}>0$ (because $0<Z_{i}<1$) and $J_{x}$ is the hitting time of level $a_{x}$ of the random walk $R_{j}$. At time $J_{x}$, the required capital injection satisfies $x^{\ast }-x\cdot Z_{1}\cdot Z_{2}\cdot\cdot\cdot Z_{j}<x^{\ast }$. Hence,
	
	\begin{align}
		C^{0}(x)=\mathbb{E}\left[ e^{-\delta \tau _{J_{x}}}(x^{\ast}-x\cdot Z_{1}\cdot Z_{2}\cdot\cdot\cdot Z_{J} + C(x^{*}))\right] \leq \left(x^{\ast }+ C(x^{*})\right)\mathbb{E}\left[ e^{-\delta \tau_{J_{x}}}\right],
		\label{Appendix A: Mathematical Proofs-Equation3}
	\end{align}
	
where $\tau _{J_{x}}$ denotes the time at which the $J_{x}$-th jump occurs. Let us define $\tau _{0}=0$, conditioning on $J_{x}$ and since $\tau _{J_{x}}=\sum_{i=1}^{\tau_{J_{x}}}\left( \tau _{i}-\tau _{i-1}\right) $, where the interarrival times are i.i.d. exponential with rate $\lambda $, we obtain

	\begin{align}
		\mathbb{E}\left[ e^{-\delta \tau _{J_{x}}}\right] =\mathbb{E}\left[ \left(\frac{\mathbb{\lambda }}{\lambda +\delta }\right)^{J_{x}}\right].
		\label{Appendix A: Mathematical Proofs-Equation4}
	\end{align}
	
	Since $Y_{i}>0$, the random walk $R_{j}$ diverges to $+\infty $ almost surely as $j\rightarrow +\infty $ (see, e.g., \cite{Book:Durrett2019}). As $\lim_{x\rightarrow +\infty }a_{x}=+\infty $, the hitting time $J_{x}\rightarrow +\infty $ almost \ surely. Since $0<\lambda /(\lambda+\delta )<1$, dominated convergence implies
	
		\begin{align}
			\lim_{x\rightarrow +\infty }\mathbb{E}\left[ \left( \frac{\mathbb{\lambda }}{\lambda +\delta }\right) ^{J_{x}}\right] =0.
		\label{Appendix A: Mathematical Proofs-Equation5}
	\end{align}
	
This proves that $\lim_{x\rightarrow +\infty }C^{0}(x)=0$ and hence the same holds for $C(x).$
\end{proof}

\subsubsection{Proof of Lemma \ref{AnalysisofV:PropertiesandtheHamilton-Jacobi-BellmanEquation-Section3-Lemma3}} \label{ProofofLemma3.3}

\begin{proof}

Since $S_{t}$ is non-decreasing and left continuous, it can be written as

\begin{align}
S_{T}=\int\nolimits_{0}^{T}dS_{t}^{c}+\sum_{\substack{S_{t}\neq S_{t^{-}
}\\t\leq T}}(S_{t}-S_{t^{-}}),
\label{Appendix A: Mathematical Proofs-Equation12}
\end{align}

where $S_{t}^{c}$ is a continuous and non-decreasing function. Take a non-negative continuously differentiable function $u$\ in $\mathbf{[}x^{\ast}\mathbf{,\infty)}$. Note that $dS^{c}_{t}$ represents the continuous increments to $S$. Since the function $e^{-\delta t}u(x)$ is continuously differentiable in $\mathbf{[}x^{\ast}\mathbf{,\infty)}$, using the expression \eqref{Appendix A: Mathematical Proofs-Equation12} and the change of variables formula for finite variation processes (see, for instance, \cite{Book:Protter1992}), we can write

\begin{align}
\begin{array}
[c]{ll}
u(X_{\tau^{\ast}}^{\pi})e^{-\delta\tau^{\ast}}-u(x) & = \int\nolimits_{0^{-}}^{\tau^{\ast}}e^{-\delta t}d\left(u(X_{t}^{\pi})\right)  -\delta\int\nolimits_{0}^{\tau^{\ast}}u(X_{t^{-}}^{\pi})e^{-\delta
t}dt\\ \\
& = \int\nolimits_{0}^{\tau^{\ast}}e^{-\delta t}u^{\prime}(X_{t^{-}}^{\pi})\ r(X_{t^{-}}^{\pi}-x^{\ast})dt-\delta\int\nolimits_{0}^{\tau^{\ast}
}e^{-\delta t}u(X_{t^{-}}^{\pi})dt\\ \\
& +\sum\limits_{\tau_{i}\leq\tau^{\ast}}e^{-\delta t}\left(u(Z_{i}\cdot X_{\left(\tau_{i}\right)  ^{-}}^{\pi})-u(X_{\left(  \tau_{i}\right)  ^{-}}^{\pi})\right) \\ \\
& +\int\nolimits_{0}^{\tau^{\ast}}e^{-\delta t}u^{\prime}(X_{t^{-}}^{\pi})dS_{t}^{c}+\sum\limits_{_{\substack{S_{t}\neq S_{t^{-}}\\t\leq\tau^{\ast}}
}}e^{-\delta t}\left(  u(X_{t}^{\pi})-u(X_{t}^{\pi}-(S_{t}-S_{t^{-}}))\right).
\end{array}
\label{Appendix A: Mathematical Proofs-Equation13}
\end{align}

One can also write,

\begin{align}
\begin{array}
[c]{l}
\int\nolimits_{0}^{\tau^{\ast}}u^{\prime}(X_{t^{-}}^{\pi})e^{-\delta t} dS_{t}^{c}+\sum\limits_{_{\substack{S_{t}\neq S_{t^{-}}\\t\leq\tau^{\ast}}}}e^{-\delta t}\left(  u(X_{t}^{\pi})-u(X_{t}^{\pi}-(S_{t}-S_{t^{-}}))\right)\\ \\
\begin{array}
[c]{cl}
= & \int\nolimits_{0}^{\tau^{\ast}}u^{\prime}(X_{t^{-}}^{\pi})e^{-\delta t}dS_{t}^{c}+\sum\limits_{_{\substack{S_{t}\neq S_{t^{-}}\\t\leq\tau^{\ast}}}}\left(\int\nolimits_{0}^{S_{t}-S_{t^{-}}}u^{\prime}(X_{t}^{\pi}-\alpha)d\alpha\right)  e^{-\delta t}\\ \\
= & -\int_{0^{-}}^{\tau^{\ast}}e^{-\delta t}dS_{t}+\int\nolimits_{0}^{\tau^{\ast}}(1+u^{\prime}(X_{t^{-}}^{\pi}))e^{-\delta t}dS_{t}^{c}\\ \\
& +\sum\limits_{_{\substack{S_{t}\neq S_{t^{-}}\\t\leq\tau^{\ast}}}}\left(\int\nolimits_{0}^{S_{t}-S_{t^{-}}}\left(  1+u^{\prime}(X_{t}^{\pi}-\alpha)\right)  d\alpha\right)  e^{-\delta t}.
\end{array}
\end{array}
\label{Appendix A: Mathematical Proofs-Equation14}
\end{align}

Analogously to Proposition 2.12 in \cite{Book:Azcue2014}, we have that $M_{T}$ is a martingale with zero expectation. Hence, from \eqref{Appendix A: Mathematical Proofs-Equation13} and \eqref{Appendix A: Mathematical Proofs-Equation14} we obtain the result.
\end{proof}

\subsubsection{Proof of Lemma \ref{ThresholdStrategies-Section4-Lemma1}} \label{ProofofLemma4.1}

\begin{proof}

It is straightforward to see that the function $V_{y}$ is non-increasing and, by definition, that $V_{y}(x)=y-x+V_{y}(y)$, for $x<y$. We now show that it is bounded. For $x\geq y$, we have the following,

\vspace{0.3cm}

\begin{align}
\begin{array}
[c]{lll}
V_{y}(x) & \leq & \mathbb{E}\left[\sum\limits_{k=1}^{\infty}ye^{-\delta\tau_{k}}\right]\\ \\
& = & y\sum\limits_{k=1}^{\infty}\prod\limits_{i=1}^{k}\mathbb{E}\left[e^{-\delta(\tau_{i}-\tau_{i-1})}\right]\\ \\
& = & y\frac{\lambda+\delta}{\delta}.
\end{array}
\label{Appendix A: Mathematical Proofs-Equation26}
\end{align}

\vspace{0.3cm}

Hence, the result follows.
\end{proof}

\subsubsection{Proof of Lemma \ref{ThresholdStrategies-Section4-Lemma2}} \label{ProofofLemma4.2}

\begin{proof}

Let us take $y>x^{\ast}$ and $x_{2}>x_{1}\geq y$ with $x_{2}-x_{1}\leq(x_{2}-x^{\ast}/2)$. Thus, this yields,

\begin{align}
\frac{1}{2}\leq1-\frac{x_{2}-x_{1}}{x_{2}-x^{\ast}}\leq1.
\label{Appendix A: Mathematical Proofs-Equation27}
\end{align}

Furthermore, since $V_{y}$ is non-increasing we have,

\begin{align}
0\leq V_{y}(x_{1})-V_{y}(x_{2})\text{.}
\label{Appendix A: Mathematical Proofs-Equation28}
\end{align}

Now, consider $T$, such that $\left(x_{1}-x^{\ast}\right)e^{rT}+x^{\ast}=x_{2}$. Hence, 

\begin{align}
T=\frac{1}{r}\ln\left(\frac{x_{2}-x^{\ast}}{x_{1}-x^{\ast}}\right).
\label{Appendix A: Mathematical Proofs-Equation29}
\end{align}

Therefore, we have that,

\vspace{0.3cm}

\begin{align}
V_{y}(x_{1})\leq V_{y}(x_{2})\cdot \mathbbm{P}(\tau_{1}>T)+V_{y}(0)\cdot \mathbbm{P}(\tau_{1}\leq T),
\label{Appendix A: Mathematical Proofs-Equation30}
\end{align}

\vspace{0.3cm}

which yields,

\begin{align}
\begin{array}
[c]{lll}
V_{y}(x_{1})-V_{y}(x_{2}) & \leq & \left(  1-\left(  1-\frac{x_{2}-x_{1}}{x_{2}-x^{\ast}}\right)^{\lambda\frac{1}{r}}\right)\left(V_{y}(0)-V_{y}(x_{2})\right) \\ \\
& \leq & \left(1-\left(1-\frac{x_{2}-x_{1}}{x_{2}-x^{\ast}}\right)
^{\lambda\frac{1}{r}}\right)  V_{y}(0).
\end{array}
\label{Appendix A: Mathematical Proofs-Equation31}
\end{align}

We next define the following function,

\begin{align}
g(h)=1-\left(  1-\frac{h}{x_{2}-x^{\ast}}\right)  ^{\lambda\frac{1}{r}},
\label{Appendix A: Mathematical Proofs-Equation32}
\end{align}

with derivative given by,

\begin{align}
g^{\prime}(h)=\frac{\lambda}{r}\left(1-\frac{h}{x_{2}-x^{\ast}}\right)^{\frac{\lambda}{r}-1}\left(\frac{1}{x_{2}-x^{\ast}}\right).
\label{Appendix A: Mathematical Proofs-Equation33}
\end{align}

\vspace{0.3cm}

So, for $h<\left(  x_{2}-x^{\ast}\right) /2$, if $\frac{\lambda}{r}-1\geq0$, we have,

\begin{align}
g^{\prime}(h)\leq\left(\frac{\lambda}{r}\right)\left(\frac{1}{x_{2}-x^{\ast}}\right)\leq\left(\frac{\lambda}{r}\right)\left(\frac{1}{y-x^{\ast}}\right),
\label{Appendix A: Mathematical Proofs-Equation34}
\end{align}

and, if $\frac{\lambda}{r}-1<0$,

\begin{align}
g^{\prime}(h)\leq\left(\frac{\lambda}{r}\right)\left(\frac{1}{x_{2}-x^{\ast}}\right)\cdot2^{1-\frac{\lambda}{r}}\leq\left(\frac{\lambda}{r}\right)\left(\frac{1}{y-x^{\ast}}\right)\cdot 2^{1-\frac{\lambda}{r}}.
\label{Appendix A: Mathematical Proofs-Equation35}
\end{align}

\vspace{0.3cm}

As a consequence, from \eqref{Appendix A: Mathematical Proofs-Equation31} and \eqref{Appendix A: Mathematical Proofs-Equation28}, $V_{y}$ is (globally) Lipschitz for $y>x^{\ast}$. The same proof holds for $C(x)$ in any set $\left[w,\infty\right) \subset(x^{\ast},\infty)$ with $w>x^{\ast}$.
\end{proof}

\subsubsection{Proof of Lemma \ref{ThresholdStrategies-Section4-Lemma3}} \label{ProofofLemma4.3}

\vspace{0.3cm}

\begin{proof}

First, we note that the function $C$ is continuous at $\left[ 0,x^{\ast}\right) $ since it is linear. Also, it is continuous in $(x^{\ast },\infty )$ since it is locally Lipschitz by Lemma \ref{ThresholdStrategies-Section4-Lemma2}. Moreover, we have by definition of $C$, 

\begin{align}
\lim_{h\rightarrow 0^{-}}C(x^{\ast }+h)=C(x^{\ast}).
\label{Appendix A: Mathematical Proofs-Equation36}
\end{align}

Thus, it only remains to show that $\lim_{h\rightarrow 0^{+}}C(x^{\ast}+h)=C(x^{\ast })$. Take $h>0$ and let us consider the (uncontrolled) capital processes $X_{t}^{0}$ and $X_{t}^{1}$, starting with initial capital $x_{0}=x^{\ast}$ and $x_{1}=x^{\ast }+h$, respectively. We stop the processes at $\tau _{1}$ (the time of the first loss event). Thus, we have 

\begin{align}
X_{\tau _{1}}^{0}=Z_{1}\cdot x^{\ast }<x^{\ast}
\label{Appendix A: Mathematical Proofs-Equation37}
\end{align}

and that

\begin{align}
X_{\tau _{1}}^{1}=Z_{1}\cdot(x^{\ast }+he^{r\tau _{1}})<x^{\ast}
\label{Appendix A: Mathematical Proofs-Equation38}
\end{align}

if and only if

\begin{align}
Z_{1}\leq \frac{x^{\ast }}{x^{\ast }+he^{r\tau _{1}}}.
\label{Appendix A: Mathematical Proofs-Equation39}
\end{align}

Then, defining 

\begin{align}
A_{h}=\left\{Z_{1}\leq \frac{x^{\ast }}{x^{\ast }+he^{r\tau _{1}}}\right\},
\label{Appendix A: Mathematical Proofs-Equation40}
\end{align}

we have

\begin{align}
\lim_{h\rightarrow 0^{+}}\mathbb{P}(A_{h}^{c})=\lim_{h\rightarrow0^{+}}\int_{0}^{+\infty }\left(1-G_{Z}\left(\frac{x^{\ast }}{x^{\ast }+he^{rt}}\right)\right)\lambda e^{-\lambda t}dt=0.  
\label{Appendix A: Mathematical Proofs-Equation41}
\end{align}

Consider now the threshold strategy with threshold $x^{\ast }$. From \eqref{Appendix A: Mathematical Proofs-Equation37} and \eqref{Appendix A: Mathematical Proofs-Equation38}, in the event that $Z_{1}\leq x^{\ast}/\left(x^{\ast }+he^{r\tau _{1}}\right)$, the transfer\ at $\tau _{1}$ for the process $X_{t}^{1}$ is equal to $x^{\ast }-\ x^{\ast }\cdot Z_{1}$ and for the process $X_{t}^{2}$ is $x^{\ast}-\ \left( x^{\ast }+he^{r\tau _{1}}\right)\cdot Z_{1}$. Hence, we obtain

\begin{align}
C(x^{\ast })-C(x^{\ast }+h) & \leq  \mathbb{E}\left[he^{r\tau _{1}}Z_{1}\mathbbm{1}_{\{A_{h}\}}+C(x^{\ast })e^{-\delta \tau _{1}}\mathbbm{1}_{\{A_{h}^{c}\}}\right] \\ 
& \leq  \mathbb{E}\left[he^{\left(r-\delta\right)\tau _{1}}Z_{1}\mathbbm{1}_{\{A_{h}\}}+C(x^{\ast })\cdot \mathbb{P}(A_{h}^{c})\right].
\label{Appendix A: Mathematical Proofs-Equation42}
\end{align}

Also, we have that

\begin{align}
\mathbb{E}\left[he^{r\tau _{1}}Z_{1}\mathbbm{1}_{\{A_{h}\}}e^{-\delta \tau_{1}}\right]=\int_{0}^{\infty }\lambda he^{rt}\left( \int_{0}^{\frac{x^{\ast }}{x^{\ast }+he^{rt}}}zdG_{Z}(z)\right) e^{-\left(\lambda+\delta\right) t}dt,
\label{Appendix A: Mathematical Proofs-Equation43}
\end{align}

and, 

\begin{align}
\lambda he^{rt}\int_{0}^{\frac{x^{\ast }}{x^{\ast }+he^{rt}}}zdG(z)\leq he^{rt}\frac{x^{\ast }}{x^{\ast }+he^{rt}}\leq x^{\ast}.
\label{Appendix A: Mathematical Proofs-Equation44}
\end{align}

Thus, taking $h\rightarrow 0^{+},$ using the Bounded Convergence Theorem and from \eqref{Appendix A: Mathematical Proofs-Equation41}, we obtain the result.
\end{proof}

\setcitestyle{numbers} 
\bibliographystyle{chicago} 
\bibliography{main}

@book{Book:Abramowitz1964,
  title = {Handbook of Mathematical Functions with Formulas, Graphs, and Mathematical Tables},
  author = {Abramowitz, M. and Stegun, Irene A.},
  publisher = {U.S. Department of Commerce},
  address = {Washington, D.C.},
  year = {1972}
}

@book{Book:Slater1960,
  title = {Confluent Hypergeometric Functions},
  author = {Slater, Lucy J.},
  publisher = {Cambridge University Press},
  address = {New York},
  year = {1960},
}

@article{Article:Albrecher2011,
    author = {H. Albrecher and Gerber, H. U. and Shiu, E. S. W.},
    doi = {},
    issn = {},
    journal = {European Actuarial Journal},
    number = {1},
    pages = {43--55},
    title = {{The Optimal Dividend Barrier in the Gamma–Omega Model}},
    volume = {1},
    year = {2011}
}

@book{Book:Asmussen2010,
  title = {Ruin Probabilities},
  author = {Asmussen, S. and Albrecher, H.},
  publisher = {World Scientific},
  address = {Singapore},
  year = {2010},
}

@article{Article:Azais2015,
    author = {Aza{\"\i}s, R. and Genadot, A.},
    doi = {},
    issn = {},
    journal = {TEST},
    number = {2},
    pages = {341-360},
    title = {{Semi-Parametric Inference for the Absorption Features of a Growth-Fragmentation Model}},
    volume = {24},
    year = {2015}
}

@article{Article:Bellman1954,
    author = {Bellman, Richard},
    doi = {},
    issn = {},
    journal = {Proceedings of the National Academy of Sciences},
    number = {4},
    pages = {231--235},
    title = {{Dynamic Programming and a New Formalism in the Calculus of Variations}},
    volume = {40},
    year = {1954}
}

@article{Article:Crandall1983,
    author = {Crandall, Michael G. and Lions, Pierre-Louis},
    doi = {},
    issn = {},
    journal = {Transactions of the American Mathematical Society},
    number = {1},
    pages = {1--42},
    title = {{Viscosity Solutions of Hamilton-Jacobi Equations}},
    volume = {277},
    year = {1983}
}

@article{Article:Dfid2006,
    author = {{Department for International Development (DFID)}},
    doi = {},
    issn = {},
    journal = {Social Protection Briefing Note Series},
    number = {},
    pages = {},
    title = {{Social Protection in Poor Countries}},
    volume = {1},
    year = {2006}
}

@book{Book:Durrett2019,
  title = {Probability: Theory and Examples},
  author = {Durrett, Rick},
  publisher = {Cambridge University Press},
  address = {Cambridge, United Kingdom},
  year = {2019},
}

@article{Article:Farrington2006,
    author = {Farrington, John and Slater, Rachel},
    doi = {},
    issn = {},
    journal = {Development Policy Review},
    number = {5},
    pages = {499-511},
    title = {{Introduction: Cash Transfers: Panacea for Poverty Reduction or Money Down the Drain?}},
    volume = {24},
    year = {2006}
}

@article{Article:Firpo2014,
    author = {Sergio Firpo and Renan Pieri and Euclides Pedroso and André Portela Souza},
    doi = {},
    issn = {},
    journal = {EconomiA},
    number = {3},
    pages = {243-260},
    title = {{Evidence of Eligibility Manipulation for Conditional Cash Transfer Programs}},
    volume = {15},
    year = {2014}
}

@article{Article:Habimana2021,
    author = {Dominique Habimana and Jonathan Haughton and Joseph Nkurunziza and Dominique Marie-Annick Haughton},
    doi = {},
    issn = {},
    journal = {World Development Perspectives},
    number = {},
    pages = {100341},
    title = {{Measuring the Impact of Unconditional Cash Transfers on Consumption and Poverty in Rwanda}},
    volume = {23},
    year = {2021}
}

@article{Article:Handa2006,
  author = {Handa, S. and Davis, B.},
  doi = {},
  issn = {},
  journal = {Development Policy Review},
  number = {5},
  pages = {513-536},
  title = {{The Experience of Conditional Cash Transfers in Latin America and the Caribbean}},
  volume = {24},
  year = {2006}
}

@book{Book:Harvey2005,
  title = {Cash and Vouchers in Emergencies},
  author = {Paul Harvey},
  publisher = {Overseas Development Institute},
  address = {London, United Kingdom},
  year = {2005}
}

@article{Article:Henshaw2023b,
    author = {Henshaw, K. and Ramirez, Jorge M. and Flores-Contró, José M. and Thomann, Enrique A. and Loke, S. H. and Constantinescu, Corina D.},
    doi = {},
    issn = {},
    journal = {Working Paper},
    number = {},
    pages = {},
    title = {{On the Impact of Insurance on Households Susceptible to Random Proportional Losses: An Analysis of Poverty Trapping}},
    volume = {},
    year = {2023}
}

@article{Article:Flores-Contro2024a,
    author = {José Miguel Flores-Contró and Séverine Arnold},
    doi = {},
    issn = {},
    journal = {Scandinavian Actuarial Journal},
    number = {8},
    pages = {781--812},
    title = {{The Role of Direct Capital Cash Transfers Towards Poverty and Extreme Poverty Alleviation - An Omega Risk Process}},
    volume = {2024},
    year = {2024}
}

@article{Article:Flores-Contro2024b,
    author = {José Miguel Flores-Contró and Kira Henshaw and Sooie-Hoe Loke and Séverine Arnold and Corina Constantinescu},
    doi = {},
    issn = {},
    journal = {North American Actuarial Journal},
    number = {1},
    pages = {44--73},
    title = {{Subsidizing Inclusive Insurance to Reduce Poverty}},
    volume = {29},
    year = {2025}
}

@article{Article:Flores-Contro2025,
    author = {Flores-Contró, José Miguel},
    doi = {},
    issn = {},
    journal = {Scandinavian Actuarial Journal},
    number = {3},
    pages = {348--369},
    title = {{The Gerber-Shiu Expected Discounted Penalty Function: An Application to Poverty Trapping}},
    volume = {2026},
    year = {2026}
}

@article{Article:Keen1992,
    author = {Michael Keen},
    doi = {},
    issn = {},
    journal = {The Economic Journal},
    number = {410},
    pages = {67--79},
    title = {{Needs and Targeting}},
    volume = {102},
    year = {1992}
}

@article{Article:Kovacevic2011,
    author = {Kovacevic, Raimund M. and Pflug, Georg Ch.},
    doi = {},
    issn = {},
    journal = {Journal of Risk and Insurance},
    number = {4},
    pages = {1003-1027},
    title = {{Does Insurance Help to Escape the Poverty Trap? — A Ruin Theoretic Approach}},
    volume = {78},
    year = {2011}
}

@article{Article:Owusu-Addo2023,
    author = {Owusu-Addo, Ebenezer and Renzaho, Andre M. N. and Sarfo-Mensah, Paul and Sarpong, Yaw A. and Niyuni, William and Smith, Ben J.},
    doi = {},
    issn = {},
    journal = {Poverty \& Public Policy},
    number = {2},
    pages = {173-198},
    title = {{Sustainability of Cash Transfer Programs: A Realist Case Study}},
    volume = {15},
    year = {2023}
}

@book{Book:UnitedNations2015,
  title = {{Transforming Our World: The 2030 Agenda for Sustainable Development}},
  author = {{United Nations}},
  series = {},
  number = {A/RES/70/1},
  year = {2015},
  publisher = {United Nations},
  address = {New York}
}

@book{Book:Protter1992,
  title = {{Stochastic Integration and Differential Equations: A New Approach}},
  author = {Philip Protter},
  series = {},
  number = {},
  year = {1992},
  publisher = {Springer-Verlag Berlin Heidelberg},
  address = {New York}
}

@book{Book:Azcue2014,
  title = {{Stochastic Optimization in Insurance: A Dynamic Programming Approach}},
  author = {Pablo Azcue and Nora Muler},
  series = {},
  number = {},
  year = {2014},
  publisher = {Springer},
  address = {New York Heidelberg Dordrecht London}
}

@book{Book:Schmidli2007,
  title = {{Stochastic Control in Insurance}},
  author = {Schmidli, Hanspeter},
  series = {},
  number = {},
  year = {2007},
  publisher = {Springer},
  address = {London}
}

@book{Book:Slater2010,
  title = {{Appropriate, Achievable and Acceptable: A Practical Tool for Good Targeting}},
  author = {Slater, Rachel and Farrington, John},
  series = {},
  number = {},
  year = {2010},
  publisher = {ODI},
  address = {London}
}

@article{Article:Slater2011,
    author = {Slater, Rachel},
    doi = {},
    issn = {},
    journal = {International Journal of Social Welfare},
    number = {3},
    pages = {250-259},
    title = {{Cash Transfers, Social Protection and Poverty Reduction}},
    volume = {20},
    year  = {2011}
}

@article{Article:Tavor2002,
    author = {Tabor, Steven R},
    doi = {},
    issn = {},
    journal = {World Bank Social Protection Discussion Paper},
    number = {},
    pages = {79--97},
    title = {{Assisting the Poor with Cash: Design and Implementation of Social Transfer Programs}},
    volume = {223},
    year  = {2002}
}

@book{Book:Bardi1997,
  title = {Optimal Control and Viscosity Solutions of Hamilton-Jacobi-Bellman Equations},
  author = {Bardi, Martino and Capuzzo-Dolcetta, Italo},
  series = {Systems \& Control: Foundations \& Applications},
  publisher = {Springer Science+Business Media},
  address = {New York},
  year = {1997},
}

@book{Book:Fleming2006,
  title = {Controlled Markov Processes and Viscosity Solutions},
  author = {Fleming, Wendell H. and Soner, H. M.},
  series = {},
  publisher = {Springer Science+Business Media, Inc},
  address = {United States of America},
  year = {2006},
}

@book{Book:WorldBank2012,
  title   = {{The Cash Dividend: The Rise of Cash Transfer Programs in Sub-Saharan Africa}},
  author  = {Garcia, Marito and Moore, Charity M. T.},
  series = {Directions in Development - Human Development},
  publisher = {The World Bank},
  volume    = {},
  address   = {Washington, D.C.},
  pages   = {},
  year    = {2012}
}

@book{Book:Churchill2012,
  title = {Protecting the Poor: A Microinsurance Compendium --- Volume II},
  author = {Churchill, C. and Matul, M.},
  publisher = {International Labour Organization (ILO)},
  address = {Geneva, Switzerland},
  year = {2012}
}

@book{Book:Churchill2006,
  title = {Protecting the Poor: A Microinsurance Compendium --- Volume I},
  author = {Churchill, C.},
  publisher = {International Labour Organization (ILO)},
  address = {Geneva, Switzerland},
  year = {2006}
}

@misc{Misc:IAIS2015,
  title = {{Issues Paper on Conduct of Business in Inclusive Insurance}},
  author = {{International Association of Insurance Supervisors (IAIS)}},
  year = {2015},
  address = {Basel}
}

@incollection{InCo:Soner1988,
  title = {{Optimal Control of Jump-Markov Processes and Viscosity Solutions}},
  author = {Soner, Halil Mete},
  booktitle = {Stochastic Differential Systems, Stochastic Control Theory and Applications},
  pages = {501--511},
  year = {1988},
  address = {Springer-Verlag},
  publisher={Springer-Verlag}
}

@book{Book:Seaborn1991,
  title = {Hypergeometric Functions and Their Applications},
  author = {Seaborn, James B.},
  year = {1991},
  publisher = {Springer-Verlag},
  address = {New York}
}

\setcitestyle{authoryear}

\end{document}